\documentclass[12pt, letterpaper]{extarticle}
\usepackage{amsfonts}
\usepackage{eurosym}
\usepackage{titlesec}
\usepackage{titling}
\usepackage{setspace}
\usepackage{amssymb}
\usepackage{enumitem}
\usepackage{tabularx}
\usepackage{float}
\usepackage{caption}
\usepackage{amsmath}
\usepackage{indentfirst}
\usepackage{xcolor}
\usepackage[colorlinks = true,
            linkcolor = princetonorange!70!black,
            urlcolor  = princetonorange!70!black,
            citecolor = princetonorange!70!black,
            anchorcolor = princetonorange!70!black]{hyperref}
\usepackage{fancyhdr}
\usepackage{cleveref}
\usepackage{geometry}
\usepackage{atbegshi}
\usepackage{amsthm}
\usepackage[nopar]{lipsum}
\usepackage{mathtools}
\usepackage{calc}
\usepackage{natbib}
\usepackage{graphicx}
\usepackage{subcaption}

\definecolor{princetonorange}{rgb}{1.0, 0.561, 0.0}
\titleformat{\section}
  {\scshape\filcenter}{\thesection.}{1em}{}
\titleformat{\subsection}[runin]
        {\normalfont\bfseries}
        {\thesubsection.}
        {0.5em}
        {}
        [.]
\titleformat{\subsubsection}[runin]
        {\normalfont\bfseries}
        {\thesubsubsection.}
        {0.5em}
        {}
        [.]
\renewenvironment{abstract}
               {\list{}{\rightmargin\leftmargin}                \item[\textsc{Abstract.}]\relax}
               {\endlist}
\providecommand{\keywords}[1]{\noindent\textbf{{Keywords:}} #1}

\AtBeginDocument{\AtBeginShipoutNext{\AtBeginShipoutDiscard}}
\newtheorem*{assumption*}{\assumptionnumber}
\providecommand{\assumptionnumber}{}


\usepackage{bm}
\newtheoremstyle{special}
    {\topsep}
    {\topsep}
    {\itshape}
    {}
    {\bfseries}
    {}
    {0.5em}
    {{\thmname{#1}\thmnumber{ #2$^{\bm*}\!$.}\thmnote{\ \textmd{(#3)}}}}

\newtheorem{prop}{Proposition}
\newtheorem{lemma}{Lemma}
\newtheorem{theorem}{Theorem}

\theoremstyle{definition}

\theoremstyle{special}
\newtheorem{spprop}[prop]{Proposition}

\makeatletter
\renewenvironment{proof}[1][\proofname] {\par\pushQED{\qed}\normalfont\topsep6\p@\@plus6\p@\relax\trivlist\item[\hskip\labelsep\bfseries#1\@addpunct{.}]\ignorespaces}{\popQED\endtrivlist\@endpefalse}
\makeatother
\crefname{prop}{proposition}{props}
\crefname{corollary}{corollary}{corollaries}
\newenvironment{delayedproof}[1]
 {\begin{proof}[Proof of \Cref{#1}]}
 {\end{proof}}
\DeclareMathOperator{\Var}{Var}
\DeclareMathOperator{\E}{\mathbb{E}}

\DeclarePairedDelimiter\ceil{\lceil}{\rceil}
\DeclarePairedDelimiter\floor{\lfloor}{\rfloor}

\newtheorem*{Berry}{\textbf{Berry–Esseen theorem} {\normalfont (Theorem 2 in \citealp{petrov_estimates_1975})}}

\makeatletter
\newcommand\footnoteref[1]{\protected@xdef\@thefnmark{\ref{#1}}\@footnotemark}
\makeatother
\makeatletter
\newcounter{tmp@cnt}
\newcommand*\@labelpunc{.}
\newcommand*\combine[1][2]{    \refstepcounter{enumi}
    \setcounter{tmp@cnt}{\value{enumi}}
    \addtocounter{enumi}{#1-1}
    \item[\thetmp@cnt--\theenumi\@labelpunc]}
    
\newcommand{\biborder}[1]{}
\usepackage{relsize}

\usepackage{ifthen}
\newboolean{showcomments}
\setboolean{showcomments}{true}

\newcommand{\leeat}[1]{  \ifthenelse{\boolean{showcomments}}
{\textcolor{blue}{(L:  #1)}}{}}
\newcommand{\kirill}[1]{  \ifthenelse{\boolean{showcomments}}
{\textcolor{blue}{(K:  #1)}}{}}

\usepackage[multiple]{footmisc}
\begin{document}

\title{\bfseries{\large Dominance Solvability in Random Games}}
\author{%
\scshape{\normalsize Noga Alon\thanks{Department of Mathematics, Princeton University, \href{mailto:nalon@math.princeton.edu}{nalon@math.princeton.edu}} \qquad Kirill Rudov\thanks{Department of Economics, Princeton University, \href{mailto:krudov@princeton.edu}{krudov@princeton.edu}} \qquad Leeat Yariv\thanks{Department of Economics, Princeton University, \href{mailto:lyariv@princeton.edu}{lyariv@princeton.edu} }\ \thanks{We thank Roland B{\'e}nabou, Amanda Friedenberg, Faruk Gul, Annie Liang, and Drew Fudenberg for useful feedback. We gratefully acknowledge support from NSF grant SES-1629613.}}}
\date{{\normalsize \today}}

\begin{titlingpage}
\usethanksrule
\begin{changemargin}{0.5in}{0.5in} 
 \maketitle
\begin{abstract}
\normalsize We study the effectiveness of iterated elimination of strictly-dominated actions in random games. We show that dominance solvability of games is vanishingly small as the number of at least one player's actions grows. Furthermore, conditional on dominance solvability, the number of iterations required to converge to Nash equilibrium grows rapidly as action sets grow. Nonetheless, when games are highly imbalanced, iterated elimination simplifies the game substantially by ruling out a sizable fraction of actions. Technically, we illustrate the usefulness of recent combinatorial methods for the analysis of general games.
\end{abstract}
\keywords{Random Games, Dominance Solvability, Iterated Elimination}
\end{changemargin} 
\end{titlingpage}
\newgeometry{bottom=1in, top=1.5in, right=1in,left=1in}

\section{Introduction}
\subsection{Overview}

First introduced by \cite{moulin_dominance_1979} in the context of voting, dominance solvability relies on a straightforward prescription. If a player has an action that generates worse payoffs than another regardless of what other players select---a strictly dominated action---she should never use it. When the structure of the game is commonly known, other players can infer their opponents' strictly dominated actions and assume they will not be played. With those strictly dominated actions eliminated, the resulting, reduced game may have further strictly dominated actions that can then be eliminated, and so on and so forth. This iterative procedure allows players to restrict the set of relevant actions they consider. If it converges to a unique action profile, that profile constitutes a Nash equilibrium, and the game is dominance solvable.

Dominance-solvable games are appealing on both simplicity and robustness grounds. Players do not need to hold precise beliefs about opponents or even accurately assess the payoffs resulting from each action profile---whether or not a game is dominance solvable, and the resulting predictions, depend only on \textit{ordinal} comparisons of players' payoffs. These features have suggested suitability to a range of applications, and much effort has gone into identifying naturally-occurring dominant-solvable games and implementing desirable outcomes through protocols inducing dominance solvability. 

Despite the attention dominance-solvable games have received, little is known about the iterated-elimination procedure's features and outcomes. This paper provides a general analysis of the procedure in random games. Our results highlight its potential limitations in simplifying games meaningfully.




Dominance solvability is so fundamental in game theory. Why don't we have a full understanding of the implications of the iterated-elimination procedure in our canon of knowledge already? We suspect one reason might be that our analysis requires fairly recent results in combinatorics. The main difficulty in studying dominance solvability arises since, whatever distribution over payoff rankings of actions profiles is assumed, after each iteration, the remaining actions players consider are \textit{selected} and resulting payoff rankings are no longer distributed in the same way. 


We consider random games, where the ranking of payoffs resulting from all possible action profiles is determined uniformly at random. This allows us to analyze the likelihood of different dominance features within games of varying sizes. We later show that our results translate to a variety of other distributions corresponding to payoff structures commonly assumed in the literature.\footnote{We present all of our results for two-player games. We show in the Online Appendix that results become even starker for more than two players.} 

Perhaps confirming common wisdom, we show that the probability a game is  dominance solvable vanishes quickly as any player's action set grows. Even in $2\times n$ games, this probability is strictly decreasing in $n$ and proportional to $n^{-1/2}$.\footnote{Such games correspond to settings in which one player has a coarser action set---a seller deciding whether to sell an item or not to buyers who pick payment levels, a firm that chooses whether to hire an employee or not, where the employee selects an effort level, etc.} Our derivation of this probability is based on a link we uncover between the number of players' undominated actions and \textit{Stirling numbers of the first kind}, a prominent sequence in combinatorics, enumerating various constructs since at least the 18th century.\footnote{The $k$-th Stirling number of the first kind captures the number of permutations of $n$ items with precisely $k$ cycles. For a rich discussion of applications of these numbers, see \cite{stanley_enumerative_2011}.} 
The asymptotics of these numbers' distributions, which we employ, have been discovered only over the last couple of decades.

As we increase the action sets of both players, for $m \times n$ games with $m \leq n$, the probability a game is dominance solvable is $n^{-\Theta(m)}$ and vanishes more rapidly.\footnote{We write $f(n)=\Theta(g(n))$ if both $g(n)=\mathcal{O}(f(n))$ and $f(n)=\mathcal{O}(g(n))$. Informally, it means that $f$ is bounded both above and below by $g$ asymptotically.} 

These results indicate how special many of the games the literature focuses on are. They also make the classic virtual implementation results \`{a} la \cite{abreu_virtual_1992} appear even more remarkable than before: approximation of a large class of implementation problems can be done utilizing only dominance-solvable games, despite their rarity.


The experimental literature on level-k thinking and cognitive hierarchies (see, e.g., \citealp*{costa-gomes_cognition_2001, camerer_cognitive_2004}) suggests limited ability of individuals to go through more than two iterations. \cite{fudenberg_predicting_2019} conduct Amazon Mechanical Turk (MTurk) experiments using 200 two-player $3 \times 3$ games with payoffs determined uniformly at random. An analysis of their data, as depicted in  Figure~\ref{fig: bar_fl}, demonstrates that compliance with equilibrium decisions for the row player (or simply Row) is high in \textit{solvable} games in which Row needs one or two rounds to find her equilibrium decision, but significantly lower in \textit{solvable} games in which Row needs to perform at least three iterations, or \textit{non-solvable} games with exactly one Nash equilibrium that is pure.\footnote{We are grateful to Drew Fudenberg and Annie Liang for sharing their data with us.} Thus, dominance solvability alone may not guarantee the \textquotedblleft simplicity\textquotedblright\ of a game.

Conditional on a game being dominance solvable, we look at the number of iterations required to complete the elimination process. We show that this number is large, increasing rapidly as the number of actions of at least one of the players grows. As action sets expand, \textquotedblleft simple\textquotedblright\ games become rare---they are unlikely to be dominance solvable and, even when they are, they likely require tremendous sophistication of players to reach an equilibrium outcome. Our results also open the door to questions regarding the features of dominance-solvable games required to approximate various allocation objectives. Indeed, \cite*{katok_implementation_2002} illustrate the limitations of virtual implementation in the lab due to the limited number of dominance iterations participants can successfuly perform. 

\begin{figure}
    \centering
    \includegraphics[width=0.9\linewidth]{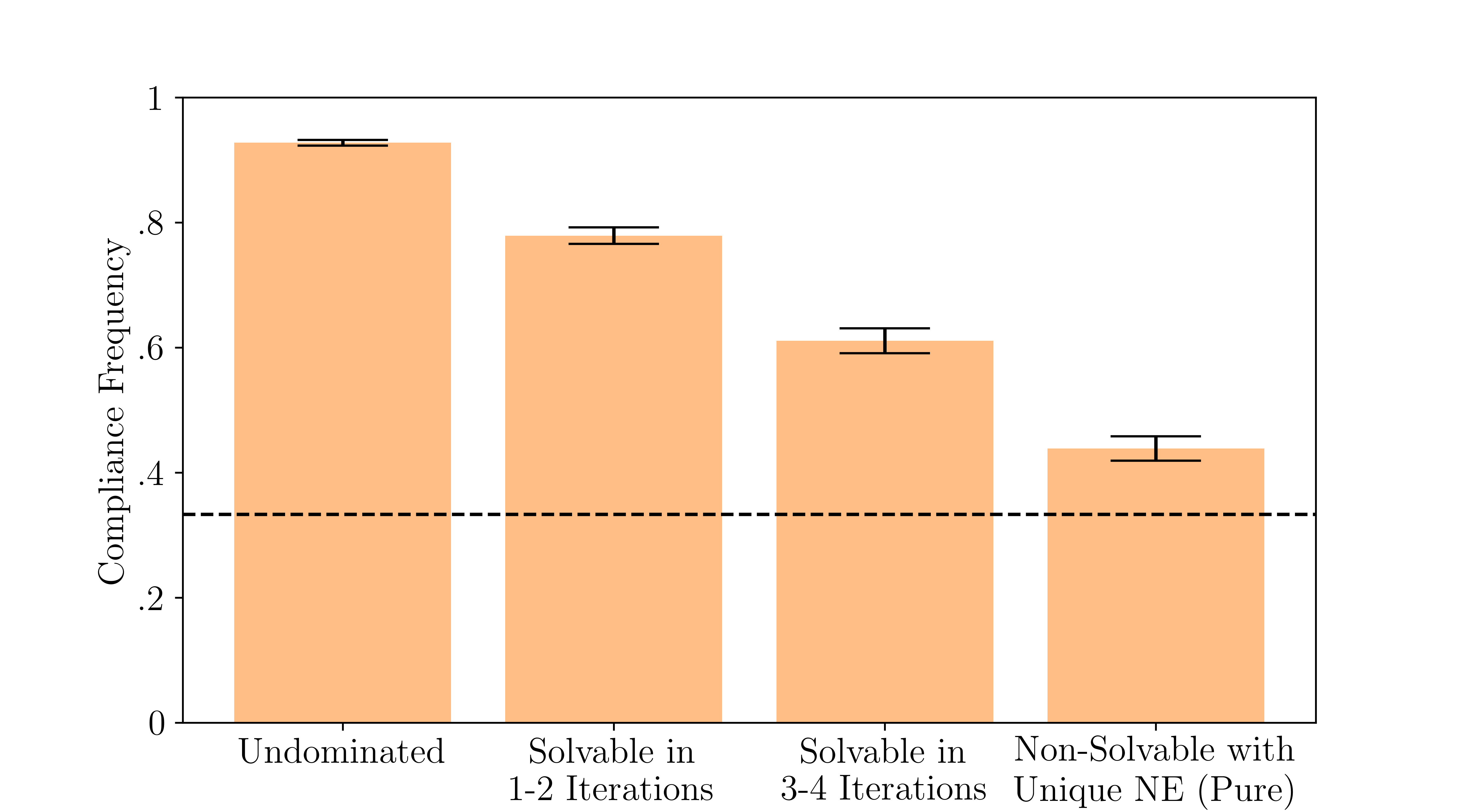}
    \caption{Frequency of Row's \textit{undominated} decisions in games with at least one Row's dominated action, \textit{iteratively undominated} decisions in \textit{solvable} games with 1 or 2 (3 or 4, respectively) rounds that Row needs to identify her equilibrium decision, \textit{equilibrium} decisions in \textit{non-solvable} games with one Nash equilibrium that is also pure}
    \label{fig: bar_fl}
    \end{figure}

Even without dominance solvability, iterated elimination of strictly dominated actions may still be effective in simplifying a game if the set of actions surviving it is relatively small. We show that whether this is the case depends on the relative number of actions each player has in the underlying game. For $2\times n$ games, the number of surviving actions for the second, column player has a mean of approximately $\ln n$ and is asymptotically normally distributed. Furthermore, for $m \times n$ games with relatively small $m=o(\ln{n})$, the fraction of surviving actions for the column player converges to zero asymptotically. This provides a silver lining to our previous results---$m\times n$ games with relatively small $m$ are dramatically simplified after the elimination process. Results are more discouraging when the first, row player has more actions. We show that in $m \times n$ games with $m=\log_2{n}+\omega(1)$, almost \textit{all} actions survive the iterative deletion process as $n$ grows.\footnote{We write $f(n)=\omega(g(n))$ if $g(n)=o(f(n))$. Informally, it means that $f$ dominates $g$ asymptotically.}

Throughout, we consider domination only via pure actions. Our notion of strict-dominance solvability closely relates to the rationalizability notion proposed by \cite{borgers_pure_1993}. Experimental evidence suggests that, indeed, identifying actions dominated by mixed strategies is far more challenging. Nonetheless, our results shed light on game complexity as viewed through the lens of the traditional rationalizability notion \citep{bernheim_rationalizable_1984, pearce_rationalizable_1984}. We show that our main insights carry over: rationalizability rarely yields a unique outcome and requires many iterations even when it does. Furthermore, the corresponding iterative procedure is frequently ineffective in limiting the actions agents need to consider.

\subsection{Literature Review}

Dominance solvability was first introduced by \cite{gale1953theory} and \cite{raiffa1957games}, with \cite{moulin_dominance_1979} offering one of its first uses as a weakening of strategy proofness in the context of voting. Dominance solvability has since been studied in a variety of applications, including auctions (see \citealp*{azrieli_dominance-solvable_2011}, and references therein), oligopolistic competition (\citealp*{borgers_dominance_1995}), and global games (\citealp*{carlsson_global_1993}).\footnote{Our study also closely relates to notions of rationalizability (\citealp*{bernheim_rationalizable_1984, pearce_rationalizable_1984}), see Section \ref{Rationalizability} for further links to that literature.}

Market designers use dominance-solvable games when implementing various allocation problems, in part because these games appear simple. In addition to \cite{abreu_virtual_1992}'s important work (mentioned above), recent research has addressed problems of robust implementation using dominance-solvable games, see e.g., \cite{bergemann_robust_2009}. This literature rarely considers the number of iterations required to reach an equilibrium in dominance-solvable games, although several recent papers have also identified strategically simple mechanisms; for example, \cite{borgers_strategically_2019} and \cite{li2017obviously}.\footnote{\cite{matsushima_mechanism_2007} and \cite{matsushima_detail-free_2008} consider incomplete-information settings with implementation in few rounds of iterated elimination of strictly-dominated strategies. Similarly, \cite*{kartik_simple_2014} consider agents with a taste for honesty and characterize social-choice functions that can be implemented using two rounds of iterated deletion. \cite{li2020simple} study the tradeoff between mechanisms' simplicity and optimality.} The experimental literature suggests that multiple iterations may not generate rationalizable outcomes, see \cite{sefton_abreumatsushima_1996} and \cite*{katok_implementation_2002}. \cite{abreu_response_1992} suggest one solution in a lively discussion with \cite{glazer_note_1992}: “[A]gents can simply be educated about how the mechanism is solved!” Our results imply that such education may be useful more often than not.

Some experimentalists advocate selecting games at random to test predictive theories about game play, see \cite{erev2007learning}. Our analysis provides some fundamental features of these games when used. While such experiments are rare, for illustration, we use data collected by \cite{fudenberg_predicting_2019}, who conducted experiments on two-player $3 \times 3$ games with payoffs determined uniformly at random. We also use data from a collection of $3 \times 3$ lab game experiments compiled by \cite{wright2014level}. In addition, we draw on experimental literature that suggests most individuals cannot perform many iterations, not without substantial experience (see, for instance, \citealp*{nagel_unraveling_1995, costa-gomes_cognition_2001, camerer_cognitive_2004}).



\cite{powers_limiting_1990} and \cite{mclennan_expected_2005} also consider random games and analyze the number of Nash equilibria, pure and mixed, while \cite{pei_rationalizable_2019} study the distribution of the number of point-rationalizable actions in such games. Our paper is related in spirit to these predecessors, though we address different questions and use different methodologies. 


We rely on recent results in combinatorics by \cite{hammett_how_2008} and \cite{hwang_asymptotic_1995, hwang_convergence_1998}. \cite{alon_2016} and \cite{stanley_enumerative_2011} provide a general overview of these methods. We hope the techniques we introduce can be useful for related problems.


\section{The Model}

\subsection{Random Games}

Consider a non-cooperative, simultaneous-move, one-shot game of complete information with two players, Row and Column. We consider only two players for presentation simplicity---in the Online Appendix, we show our main results extend to games with more than two players. Row has $m$ actions $[m]=\{1,2,\ldots ,m\}$ and Column has $n$ actions $[n]=\{1,2,\ldots ,n\}$, where $m,n$ are positive integers. Let $R=(r_{ij})\in \mathbb{R}^{m\times n}$ and $C=(c_{ij})\in \mathbb{R}^{m\times n}$ denote the $m\times n$ Row's and Column's payoff matrices respectively. We can represent this normal-form game by a bimatrix of the form 
\begin{equation*}
(R,C)=%
\begin{pmatrix}
r_{11},c_{11} & r_{12},c_{12} & \dots & r_{1n},c_{1n} \\ 
r_{21},c_{21} & r_{22},c_{22} & \dots & r_{2n},c_{2n} \\ 
\vdots & \vdots & \ddots & \vdots \\ 
r_{m1},c_{m1} & r_{m2},c_{m2} & \dots & r_{mn},c_{mn}%
\end{pmatrix}%
.
\end{equation*}

In order to study the general properties of this class of games, we assume all payoffs are randomly generated. Since dominance solvability hinges on ordinal comparisons alone, we can focus on the randomness of payoff rankings, abstracting from the underlying cardinal payoff distributions.\footnote{We ignore indifferences, which would arise with measure $0$ for any continuous distribution of payoffs. In Section \ref{Rationalizability}, we also consider dominance via mixed strategies and rationalizability.} In that sense, our analysis is ``distribution-free.'' Let $S_{m}$ denote the symmetric group of permutations of $[m]$. We maintain the following \textit{ordinal randomness assumption} throughout our analysis:

\begin{enumerate}
\item for each $j\in \lbrack n]$, Row's payoffs ${r_{\cdot j}}$ are uniform
on $S_{m}$;

\item for each $i\in \lbrack m]$, Column's payoffs ${c_{i\cdot }}$ are
uniform on $S_{n}$;

\item random permutations $\{r_{\cdot j}, c_{i \cdot}\}$, $i \in \lbrack m
\rbrack$ and $j \in \lbrack n \rbrack$, are mutually independent.
\end{enumerate}

In other words, for each fixed action of Row or Column, Column's or Row's ordinal rankings over its actions are uniform, and all ordinal rankings are mutually independent. In Section \ref{AlternativeDistributions}, we consider alternative distributions that correspond to a variety of commonly studied classes of games; they yield qualitatively identical results.

Let $G(m,n)$ denote the corresponding \textit{random game}.

\subsection{Three Dimensions of Pure-Strategy Strict Dominance}

We examine the general properties of random games related to pure-strategy strict dominance. An action is \textit{pure-strategy strictly dominated} if it always yields a worse outcome than some other action, regardless of other players' actions. If an action is not pure-strategy strictly dominated, it is called \textit{pure-strategy strictly undominated}. An action is \textit{strictly dominant} if all alternative actions are strictly dominated.

The elimination procedure that iteratively discards of all pure-strategy strictly dominated actions until there is no pure-strategy strictly dominated action is called \textit{iterated elimination of pure-strategy strictly dominated actions}. We also call rounds of this elimination procedure \textit{iterations}.\footnote{For finite games, the order in which pure-strategy strictly dominated actions are eliminated does not matter. To define the number of iterations, we suppose that, at each iteration (or round) of the elimination procedure, all players delete all pure-strategy strictly dominated actions.} If by iterated elimination of pure-strategy strictly dominated actions there is only one action left for each player, the game is called a \textit{pure-strategy strict-dominance solvable} game. To simplify the terminology in this paper, we will often omit the \textquotedblleft pure-strategy\textquotedblright\ preamble.

Our analysis focuses on the following \textit{three dimensions of strict dominance} for any random game $G(m,n)$. First, we ask how common strict-dominance solvable games are. We address this question by studying the \textit{probability of strict-dominance solvability}, denoted by $\pi (m,n)$. Second, we ask how \textquotedblleft complex\textquotedblright\ strict-dominance solvable games are. We use the number of iterations required conditional on strict-dominance solvability as our complexity measure for strict-dominance solvable games. We call that measure the  \textit{conditional number of iterations} and denote it by $I(m,n)$. Last, we inspect the complexity of games surviving iterated elimination of pure-strategy strictly dominated actions. As a complexity measure for surviving games, we analyze the \textit{number of surviving actions} after the iterated procedure, which we denote by $S^{R}(m,n)$ for Row and $S^{C}(m,n)$ for Column.

As a by-product, we also examine the number of strictly undominated actions denoted by $U^{R}(m,n)$ for Row and $U^{C}(m,n)$ for Column. It provides insights on the likelihood of games with a dominant-strategy equilibrium, where $U^{R}(m,n)$ and $U^{C}(m,n)$ are singletons.

\section{Motivating Example: Two Actions for One Player}

In this section, we fix the number of Row's actions to $m=2$ and vary the number of Column's actions $n$.\footnote{Due to symmetry, if we instead fix the number of Column's actions, the analysis is identical.} If $m=1$, all realized games are dominance solvable within one iteration. Therefore, the minimal non-trivial case corresponds to $m=2$.



\subsection{Undominated Actions}

Since there are only two actions for Row, the distribution of her number of strictly undominated actions is straightforward. One action is strictly dominated by another action with probability $\left(  \frac{1}{2}\right) ^{n}$, where $n$ is the number of Column's actions. In addition, there are two, mutually exclusive, ways to choose a strictly dominated action. Thus, Row has one strictly undominated action with probability $\left( \frac{1}{2}\right) ^{n-1}$.

Unfortunately, we cannot follow the same argument for Column with $n$ actions and the payoff matrix 
\begin{equation*}
C=%
\begin{pmatrix}
c_{11} & c_{12} & \ldots & c_{1j} & \ldots & c_{1n} \\ 
c_{21} & c_{22} & \ldots & c_{2j} & \ldots & c_{2n}%
\end{pmatrix}
,
\end{equation*}
\noindent where rows $\{c_{1\cdot },c_{2\cdot }\}$ are \textit{i.i.d.} uniform on $S_{n}$. One action is strictly dominated by another with probability $\frac{1}{4}$ in isolation. However, there are many ways by which one action can be dominated by various others, and they are not mutually exclusive.

Instead, we employ combinatorial techniques. Because we care only about the number of undominated actions and not their labels, we can set either of $\{c_{1\cdot },c_{2\cdot }\}$ to any fixed permutation. Without loss of generality, we fix $c_{1\cdot }=e_{n}\equiv (1,2,\ldots ,n)$ and focus on 
\begin{equation*}
C=%
\begin{pmatrix}
1 & 2 & \ldots & j & \ldots & n \\ 
c_{21} & c_{22} & \ldots & c_{2j} & \ldots & c_{2n}%
\end{pmatrix}
,
\end{equation*}
where $c_{2\cdot }$ is uniform on $S_{n}$. Formally, our notation above with two rows that are \textit{i.i.d.} uniform on $S_{n}$ is equivalent to the two-row notation with one fixed row and another drawn uniformly from $S_{n}$. 

Several conclusions follow immediately. First, the $n$-th action is always strictly undominated. Furthermore, for any $1\leq j\leq n-1$, the $j$-th action is strictly undominated if and only if $c_{2j}>c_{2i}$ for all $i>j$, which occurs with probability $\frac{1}{n-j+1}$. Thus, because of linearity of expectations, the expected number of Column's undominated actions is $H_{n}$, where $H_{n}\equiv 1+\frac{1}{2}+\ldots + \frac{1}{n}$ is the $n$-th harmonic number. Since $H_{n}\sim \ln {n}$, asymptotically, the fraction of Column's undominated actions is negligible. 

To establish the distribution of the number of undominated actions, we rely on an underlying recursive structure. There are $n!$ combinations in total for $c_{2\cdot }$. Let $s(n,k)$ denote the number of combinations corresponding to exactly $k$ of Column's actions being strictly undominated, $k\in \lbrack n]$. There are two relevant cases. If $c_{21} \in [n-1]$, then Column's first action is strictly dominated and we need to have $k$ undominated actions among the remaining $(n-1)$ actions. If $c_{21}=n$, Column's first action is strictly undominated and we need to have $k-1$ undominated actions among the remaining $(n-1)$ actions. Thus, 
\begin{equation*}
s(n,k)=(n-1)s(n-1,k)+s(n-1,k-1).
\end{equation*}

This expression corresponds to the recurrence relation of the \textit{unsigned (or signless) Stirling numbers of the first kind}, commonly denoted by $s(n,k)$, 
with the initial conditions $s(n,k)=0$ if $n<k$ or $k=0$, except for $s(0,0)=1$. Therefore,

\begin{lemma}
\label{lem: u_2_n} Consider a random game $G(2,n)$. Then, for any $n \geq 1$,

\begin{enumerate}
\item {$\Pr\left(U^{R}(2,n)=1\right)=\dfrac{1}{2^{n-1}}$;}

\item {for any $k\in \lbrack n]$, $\Pr \left( U^{C}(2,n)=k\right) =\dfrac{s(n,k)}{n!}$.}
\end{enumerate}
\end{lemma}

The combinatorics literature offers various interpretations for the unsigned Stirling numbers of the first kind. The original definition of $s(n,k)$ is algebraic. Namely, they are the coefficients in the expansion of the rising factorial: 
\begin{equation*}
x^{(n)}\equiv x(x+1)\ldots (x+n-1)=\sum_{k=0}^{n}s(n,k)x^{k}.
\end{equation*}
We use this definition to find the probability of dominance solvability in Proposition~\ref{prop: p_2_n} below. There are various other interpretations. For instance, $s(n,k)$ corresponds to the number of permutations $\sigma \in S_{n}$ with exactly $k$ cycles.\footnote{For other enumerative interpretations see \cite{stanley_enumerative_2011}.} 

\subsection{Dominance Solvability}
In order to express the probability of dominance solvability, let $n!!$ denote the \textit{double factorial} of a positive integer $n$, defined as the product of all the integers from $1$ up to $n$ with the same parity (odd or even) as $n$.\footnote{By definition, $(2n-1)!!=1\cdot 3\cdot \ldots \cdot (2n-1)$ and $(2n)!!=2\cdot 4\cdot \ldots \cdot (2n)=n!\cdot 2^{n}$.} In addition, let $W(n)$ denote the so-called \textit{Wallis ratio} (\citealp{qi_best_2015}) defined as 
\begin{equation*}
W(n)\equiv \frac{(2n-1)!!}{(2n)!!}=\frac{\Gamma \left( n+1/2\right) }{\Gamma
(1/2)\Gamma (n+1)},
\end{equation*}%
\noindent where $\Gamma (x)$ is the gamma function with $\Gamma (1/2)=\sqrt{%
\pi }$.

Proposition~\ref{prop: p_2_n} provides analytical formulas for the probability of dominance solvability.

\begin{prop}
\label{prop: p_2_n} Consider a random game $G(2,n)$. Then,

\begin{enumerate}
\item for any $n \geq 1$, $\pi(2,n) = 2W(n)=\dfrac{(2n-1)!!}{2^{n-1}\cdot n!}
$;

\item $\pi(2,n)$ is strictly decreasing in $n$;

\item $\lim\limits_{n\rightarrow \infty }n^{1/2}\cdot \pi (2,n)=\dfrac{2}{%
\sqrt{\pi }}$.
\end{enumerate}
\end{prop}

Intuitively, we derive the exact formula for $\pi (2,n)$ as follows. Recall that the order in which strictly-dominated actions are eliminated does not matter. There are $n$ possible mutually exclusive events corresponding to the number $k$ of strictly undominated actions for Column, $k\in \lbrack n]$, \ each occurring with probability $\frac{s(n,k)}{n!}$ respectively. The induced $2\times k$ game, derived from eliminating all of Column's dominated actions, is strict-dominance solvable if and only if Row has exactly one strictly undominated action. This occurs with probability $\left( \frac{1}{2}\right) ^{k-1}$ since Row's and Column's payoffs are independent. By summing over all possible cases $k\in \lbrack n]$, and using the algebraic definition of $s(n,k)$ together with various well-known identities, we get the desired expression. The monotonicity of $\pi (2,n)$ follows from the identity $\Gamma (x+1)=x\Gamma (x)$. The asymptotic characterization follows from Stirling's formula applied to the gamma function.

It is interesting to note that dominance solvability is rare even conditional on there being a unique pure-strategy Nash equilibrium. Indeed, from \cite{powers_limiting_1990}, the asymptotic number of pure-strategy Nash equilibiria in $G(2,n)$ follows a binomial distribution $B(2,1/2)$.
In particular, a $2 \times n$ game has exactly one pure equilibrium with probability close to $1/2$ for large $n$. Dominance-solvable games then account for a vanishing fraction of those.  

\subsection{Conditional Iterations}

There is exactly one iteration conditional on $G(2,n)$ being strict-dominance solvable if and only if both Row and Column have strictly-dominant actions. Thus, using the identity $\Gamma (n+1)=n\Gamma (n)$, we have: 
\begin{equation*}
\Pr (I(2,n)=1) =\left( \frac{1}{2^{n-1}}\cdot \frac{1}{n}\right) \cdot 
\frac{1}{\pi (2,n)} =\frac{\sqrt{\pi }}{2^{n}}\cdot \frac{\Gamma (n)}{\Gamma (n+1/2)}\sim 
\sqrt{\pi }\cdot \frac{1}{2^{n}\cdot n^{1/2}}.
\end{equation*}

Next, there are exactly two iterations conditional on $G(2,n)$ being strict-dominance solvable if and only if either Row or Column have a dominant action, not both. That is, 
\begin{equation*}
\Pr (I(2,n)=2) =\left( \frac{1}{2^{n-1}}+\frac{1}{n}-\frac{1}{2^{n-2}}\cdot 
\frac{1}{n}\right) \cdot \frac{1}{\pi (2,n)} =\frac{n+2^{n-1}-2}{2^{n}}\cdot \sqrt{\pi } \cdot \frac{\Gamma (n)}{\Gamma(n+1/2)}\sim \frac{\sqrt{\pi }}{2}\cdot \frac{1}{n^{1/2}}.
\end{equation*}
\label{subsection: i_2_n}Finally, there are three conditional iterations in all remaining cases: 
\begin{equation*}
\Pr (I(2,n)=3)=1-\frac{n+2^{n-1}-1}{2^{n}}\cdot \sqrt{\pi }\cdot \frac{\Gamma (n)}{\Gamma (n+1/2)}\sim 1-\frac{\sqrt{\pi }}{2}\cdot \frac{1}{n^{1/2}}.
\end{equation*}

In the Online Appendix, we show that $\Pr (I(2,n)=1)$ is monotonically, and exponentially, decreasing to zero as the number of Column's actions $n$ goes to infinity. In particular, it is unlikely for games to be solvable in strictly-dominant actions. In fact, as the derivation above suggests, it is rare to have a strictly dominant action even for only one of the players. Formally, $\Pr (I(2,n)=2)$ is also monotonically decreasing, albeit not exponentially, to zero. Therefore, asymptotically, the more pervasive manner by which dominance solvability is achieved involves the maximum of three elimination iterations, where Column is the first to eliminate actions. Proposition~\ref{prop: i_2_n} summarizes this discussion by focusing on the expected number of conditional iterations.

\begin{prop}
\label{prop: i_2_n} Consider a random game $G(2,n)$. Then,

\begin{enumerate}
\item $\E\left[I(2,n)\right]=3-\dfrac{n+2^{n-1}}{2^n} \cdot \sqrt{\pi} \cdot 
\dfrac{\Gamma(n)}{\Gamma(n+1/2)}$;

\item $\E\left[I(2,n)\right]$ is strictly increasing in $n$;

\item $\lim\limits_{n\rightarrow \infty }n^{1/2}\cdot \left( 3-\E\left[
I(2,n)\right] \right) =\dfrac{\sqrt{\pi }}{2}$.
\end{enumerate}
\end{prop}

As mentioned in the introduction, experimental evidence suggests individuals' limited ability to go beyond two iterations. Proposition~\ref{prop: i_2_n} then implies that most $2\times n$ games that are dominance solvable may be de-facto challenging to reason through. As we will soon show, this point becomes even starker when both players have a substantial number of actions.

\subsection{Surviving Actions}

Row has exactly one action surviving iterated elimination of strictly dominated actions if and only if the game is strict-dominance solvable. Therefore, 
\begin{equation*}
\Pr \left( S^{R}(2,n)=1\right) =\pi (2,n),\quad \Pr \left( S^{R}(2,n)=2\right) =1-\pi (2,n)\text{,\quad and \quad} \E\left[ S^{R}(2,n)\right] =2-\pi (2,n),
\end{equation*}%
\noindent and both comparative statics and asymptotic features follow directly from Proposition~\ref{prop: p_2_n}.

As for Column, by similar arguments we obtain 
\begin{equation*}
\Pr \left( S^{C}(2,n)=1\right) =\pi (2,n).
\end{equation*}%
\noindent For any $k\neq 1$, $k\in \lbrack n]$, Column has exactly $k$ surviving actions if and only if he has exactly $k$ undominated actions and the considered game is not strict-dominance solvable, so that 
\begin{align*}
\Pr \left( S^{C}(2,n)=k\right) & =\Pr \left( U^{C}(2,n)=k\right) \cdot \Pr\left( S^{R}(2,n)\neq 1\mid U^{C}(2,n)=k\right) \\
& =\Pr \left( U^{C}(2,n)=k\right) \cdot \Pr \left( U^{R}(2,k)\neq 1\right) =\frac{s(n,k)}{n!}\cdot \left( 1-\frac{1}{2^{k-1}}\right) ,
\end{align*}
\noindent where the second equality follows from independence between Row's
and Column's payoffs.

It follows from Proposition~\ref{prop: p_2_n} that, asymptotically, Row has nothing to eliminate, so that Column can eliminate actions only in his first iteration. Intuitively, then, the difference between the number of Column's strictly-undominated actions and the number of Column's surviving actions vanishes asymptotically.

Formally, the distribution of the number of Column's surviving actions is similar to that pertaining to strictly-undominated actions with two exceptions. First, for any $k\neq 1$, $k\in \lbrack n]$, the corresponding probabilities are discounted by $\left( 1-\frac{1}{2^{k-1}}\right) $ with smaller discounts for larger $k$. Second, there is a spike at  $k=1$ in the distribution of the number of Column's surviving actions that corresponds to the probability of strict-dominance solvability. Indeed, for any $n\geq 2$, because $\Pr \left( U^{R}(2,2)=1\right) =1/2$, we have  $\Pr \left( S^{C}(2,n)=1\right) >\Pr \left( S^{C}(2,n)=2\right) $.

For any $n\geq 1$, the sequence of numbers $s(n,k)$, $k=0,1,\ldots ,n$, is log-concave\footnote{A sequence $a=(a_{0},a_{1},\cdots ,a_{n})$ of nonnegative real numbers is log-concave if $a_{k}^{2}\geq a_{k-1}a_{k+1}$ for any $k\in \lbrack n-1]$.} and, hence, unimodal (\citealp{stanley_enumerative_2011}). In addition, the signless Stirling numbers $s(n,k)$ are maximized at $k(n)$ that is either $\floor{H_n}$ or $\ceil{H_n}$. That is, $k(n)\sim \ln {n}$ asymptotically. Results by \cite{hwang_asymptotic_1995} suggest that, although for any fixed  $k\neq 1$, $k\in \lbrack n]$, the corresponding probability $\Pr \left( S^{C}(2,n)=k\right) $ converges to zero faster than $\Pr \left( S^{C}(2,n)=1\right) =\pi (2,n)$, the probability of the distibution mode---corresponding to $k(n)\sim \ln {n} $---converges to zero slower than $\pi (2,n)$. We formalize this claim in the Online Appendix. In fact, results from probabilistic combinatorics also suggest that the distribution induced by $\frac{s(n,k)}{n!}$ is asymptotically normal (\citealp{gontcharoff_du_1944, hwang_convergence_1998}). This implies that the number of surviving Column's actions is asymptotically normal. Proposition~\ref{prop: s_c_2_n} formalizes this intuition.

\begin{prop}
\label{prop: s_c_2_n} Consider a random game $G(2,n)$. Then, 
\begin{equation*}
\Pr \left( S^{C}(2,n)-\E\left[ S^{C}(2,n)\right] \leq x\cdot \sqrt{\Var\left[
S^{C}(2,n)\right] }\right) =\Phi (x)+\mathcal{O}\left( \frac{1}{\sqrt{\ln n}%
}\right) ,
\end{equation*}
\noindent where $\Phi (\cdot )$ is the distribution function of the standard
normal distribution, 
\begin{align*}
\E\left[ S^{C}(2,n)\right] =\ln {n}+\gamma +o(1)\text{,\quad and \quad}
\sqrt{\Var\left[ S^{C}(2,n)\right] }=\sqrt{\ln {n}}- \frac{{\pi ^{2}}-6\gamma}{12\sqrt{\ln {n}}}+o\left( \frac{%
1}{\sqrt{\ln {n}}}\right) .
\end{align*}
\end{prop}


The formal proof 
follows similar lines to those appearing in the analysis of \cite{hwang_convergence_1998}.\footnote{Since we have a spike and discounted probabilities, the problem does not belong to the exp-log class that \cite{hwang_convergence_1998} studies. Therefore, we cannot use his results directly.} It uses the Berry-Esseen inequality (\citealp{petrov_estimates_1975}) stated in terms of characteristic functions to find convergence rates. Namely, using the algebraic definition of $s(n,k)$, we compute the characteristic function of the number of surviving Column's actions and compare it to the characteristic function of the standard normal distribution.

\begin{figure}[t]
\centering
\begin{subfigure}[t]{.49\textwidth}
\centering
\includegraphics[width=\linewidth]{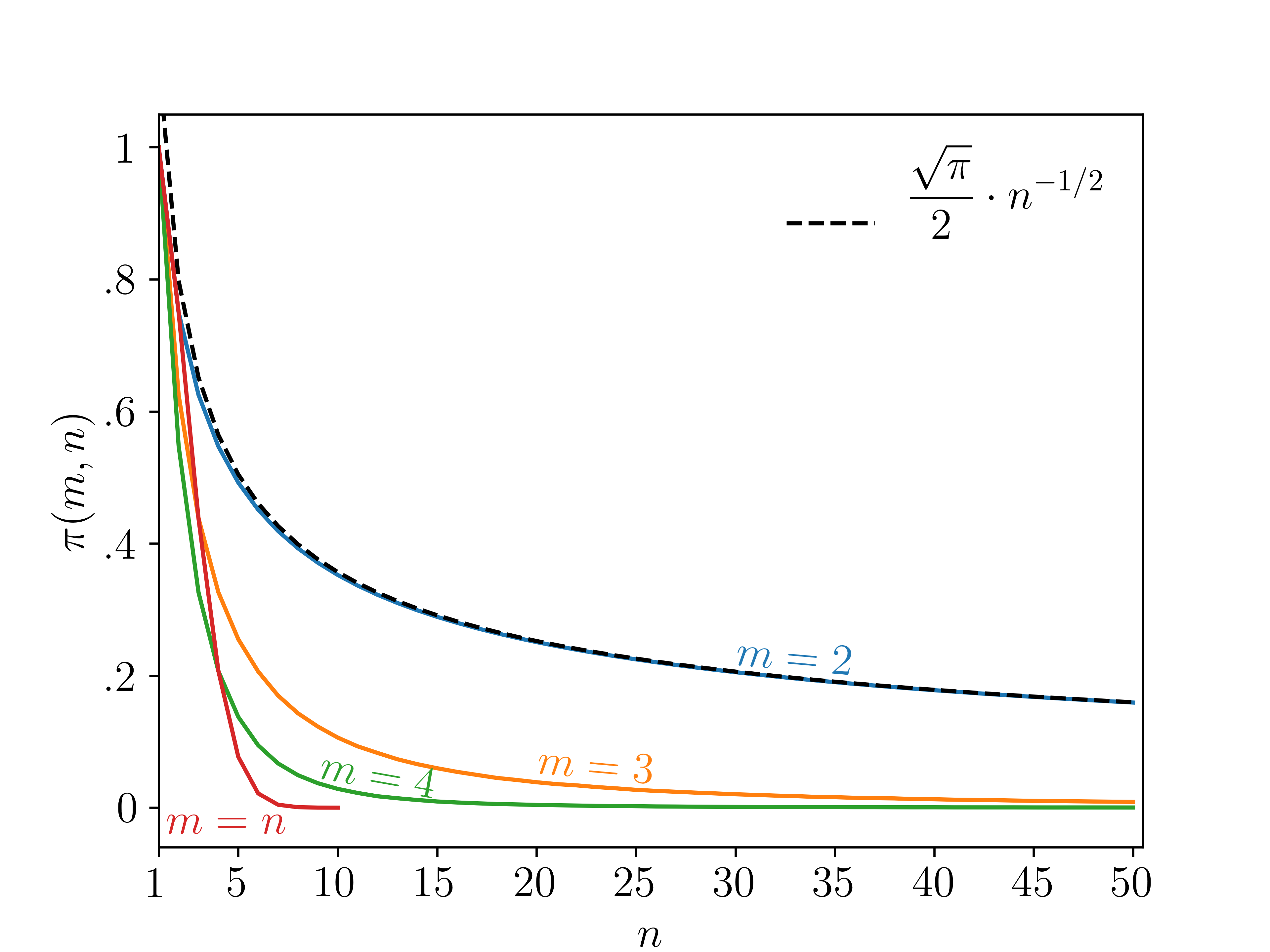}
        \caption{Probability of strict dominance}\label{fig: pi}
\end{subfigure}
\begin{subfigure}[t]{.49\textwidth}
\centering
\includegraphics[width=\linewidth]{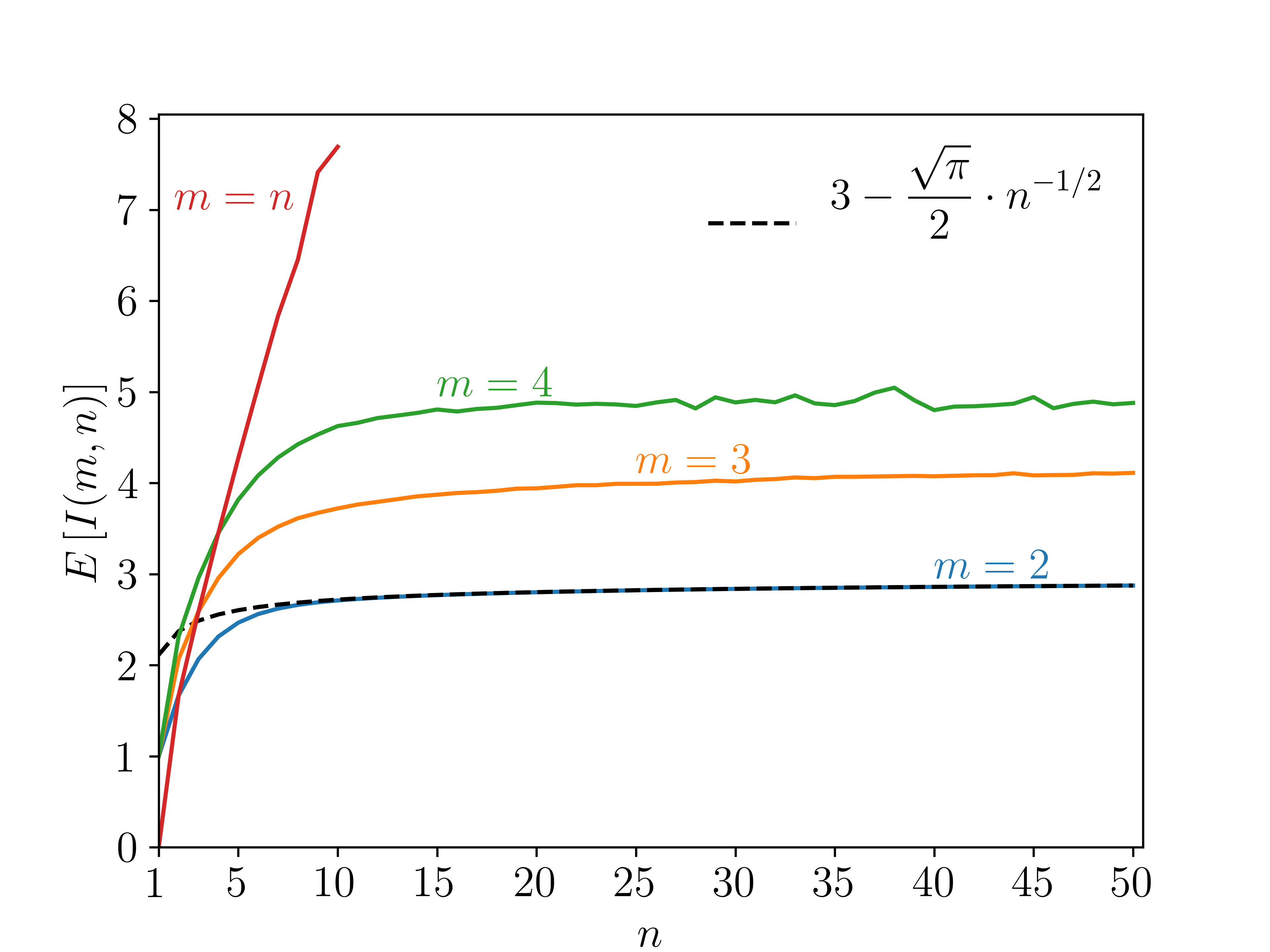}
\caption{Conditional number of iterations}\label{fig: I}
\end{subfigure}
\par
\medskip
\par
\begin{subfigure}[t]{\textwidth}
\centering
\vspace{0pt}
\includegraphics[width=0.49\linewidth]{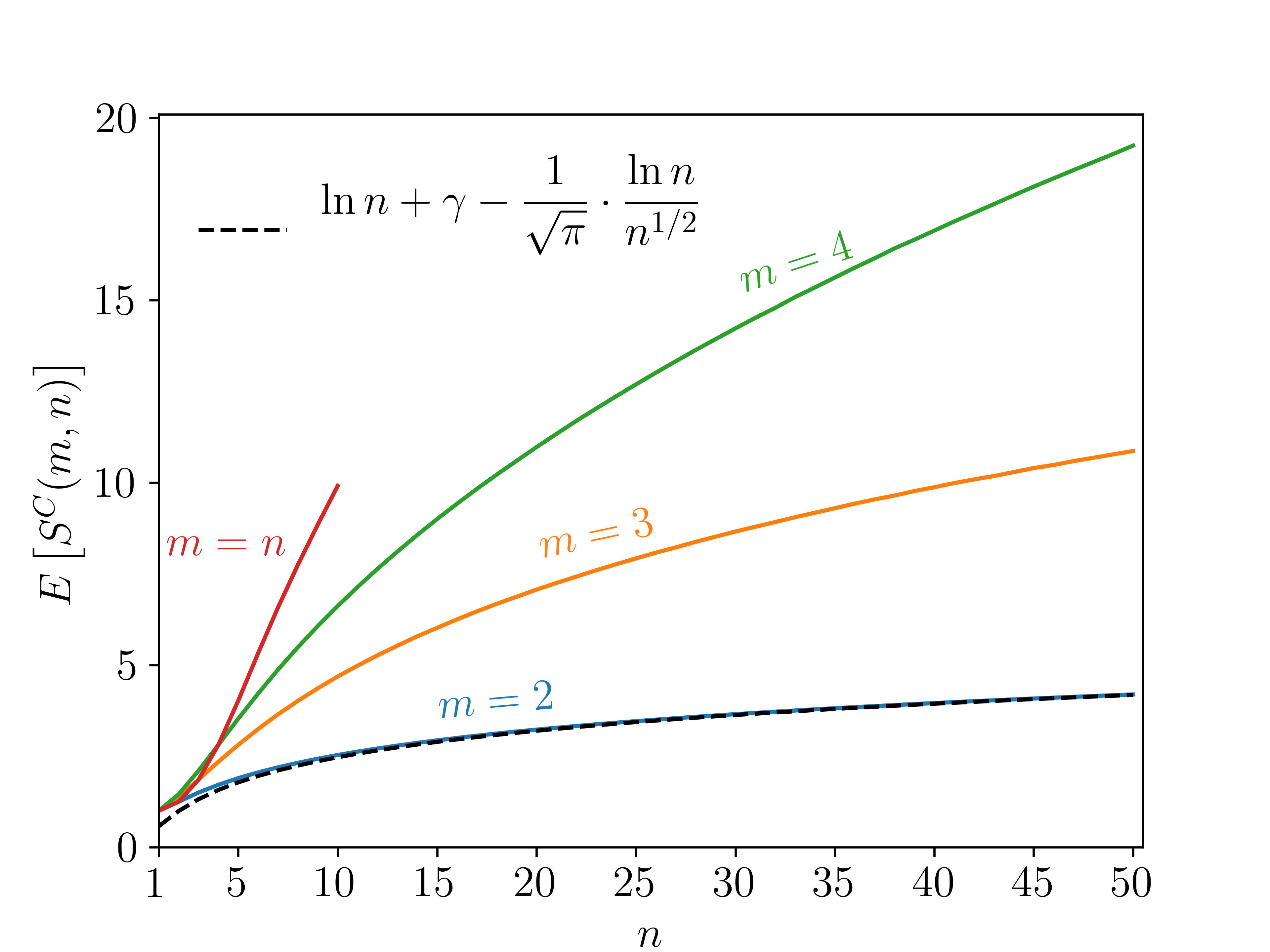}
\includegraphics[width=0.49\linewidth]{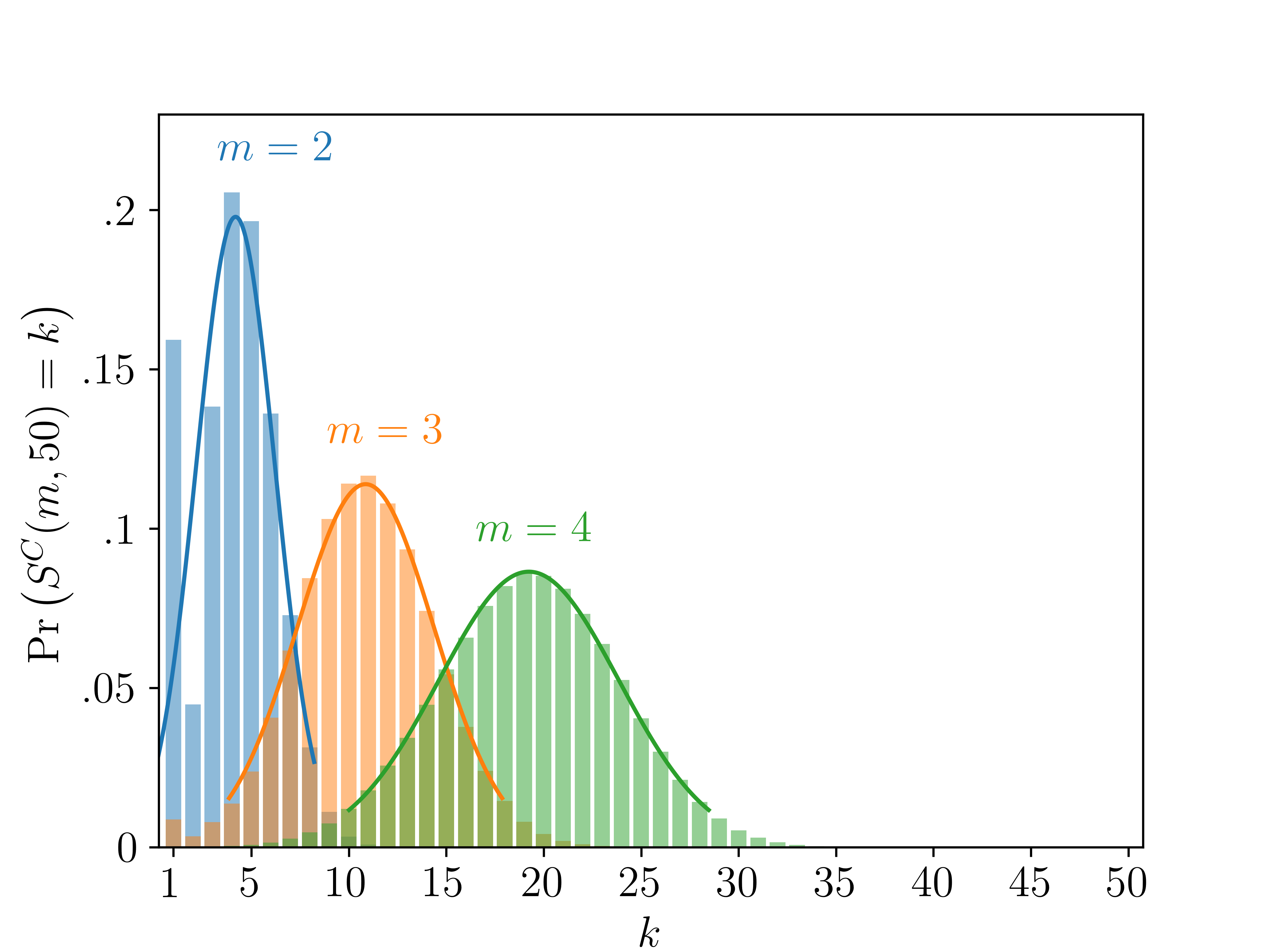}
\caption{Surviving actions}\label{fig: S_column}
\end{subfigure}
\caption{Three dimensions of dominance solvability}
\label{fig: m_n}
\end{figure}

Figure~\ref{fig: m_n} summarizes our discussion in this subsection, when focusing on the $m=2$ curves. The panels of the figure depict the different objects we analyze for random games varying in size, $n=1,...,50$: the probability of dominance solvability in panel~(\subref{fig: pi}), the expected number of conditional iterations in panel~(\subref{fig: I}), and the expectation and distribution of the number of surviving Column actions in panel~(\subref{fig: S_column}). In addition to their exact values, we also depict the asymptotic behavior analytically described in our results. As can be seen, our asymptotic characterizations provide remarkably close approximations for $2\times n$ games in which $n>5$.

\section{Arbitrary Action Sets: General Analysis}

We now turn to general $m\times n$ games. Without loss of generality, we fix the number of Column's actions to $n$ and vary the number of Row's actions $m \leq n$, as a function of $n$. 

Our general analysis suggests several main insights. First, dominance solvability is rare---for $2<m\leq n$, $\pi (m,n)=n^{-\Theta(m)}$ and it converges to zero as $n\rightarrow \infty $ at convergence rates that increase in $m$. Second, the conditional number of iterations is large even for relatively small games. Third, the iterated-elimination procedure is effective in allowing players to eliminate a significant fraction of their actions only in sufficiently imbalanced games.


In contrast to the $m=2$ case, the analysis here involves various novel enumerative issues that have not yet been studied in the combinatorics literature. We employ mostly probabilistic methods to obtain closed-form expressions for our variables of interest and study their asymptotic patterns. 

\subsection{Undominated Actions}

In general, an action can be dominated by multiple other actions. When calculating the number of undominated actions, one needs to consider the various interdependent possibilities of domination of any set of actions. For general $m\times n$ games, the exact distribution of the number of undominated actions of either player is challenging to characterize. Indeed, in the Online Appendix, we show that even for the case of $m=3$, the basic problem of finding the probability that Column has no strictly dominated actions turns out to be mathematically equivalent to a specific problem from the so-called ``permutation avoidance'' literature (e.g., see \citealp{gunby_asymptotics_2019}) and cannot be calculated explicitly \citep{hammett_how_2008}.\footnote{The Online Appendix offers a general description of the connection between the set of problems we consider and permutation patterns.} Nonetheless, we now illustrate a recursive structure of the \textit{expected} number of each player's undominated actions, which generalizes some of our observations from the $m=2$ case. 


We focus on Column, as results for Row are symmetric. Consider the expected number of Column's undominated actions. By symmetry, the probability that a given Column's action is undominated in $G(m,n)$ is independent of its label. Therefore, by the linearity of expectations, the expected proportion of Column's undominated actions $\frac{\E\left[U^C(m,n)\right]}{n}$ is also the probability that the first of Column's actions is undominated. To glean some intuition for the recurrence relation governing this probability, and hence for $\E\left[U^C(m,n)\right]$, consider $n$ mutually exclusive events, each corresponding to $c_{m1}=k$ for some $k \in [n]$. That is, we consider every possible payoff of Column from his first action, and Row's $m$-th action. Any such event occurs with probability $\frac{1}{n}$. If $c_{m1}=k$, the first action is undominated if and only if it is undominated in the reduced $(m-1)\times(n-k+1)$ game formed by removing all columns $j$ with $c_{mj}<k$ and the last row. This event occurs with probability $\frac{\E\left[U^C(m-1,n-k+1)\right]}{n-k+1}$. By summing over all possible cases $k \in \lbrack n \rbrack$, we achieve a recurrence relation, which is useful for bounding the expected number of players' undominated actions, as in the following lemma.


\begin{lemma}
\label{lem: u_m_n}
Consider a random game $G(m,n)$. Then, for any $m,n\geq2$,
\begin{enumerate}
    \item $U^{C}\left(m,n\right)$ first-order stochastically dominates $%
U^{C}\left(m-1,n\right)$;
    \item $\E\left[U^C(m,n)\right]=\mathlarger{\sum}\limits_{k=1}^n \dfrac{\E\left[U^C(m-1,k)\right]}{k}$ and it is component-wise strictly increasing;
    \item \label{lem: u_m_n_3} $\dfrac{(\ln{n})^{m-1}}{(m-1)!} \leq \E\left[U^C(m,n)\right] \leq \mathlarger{\sum}\limits_{k=0}^{m-1} \dfrac{(\ln{n})^{k}}{k!}$.
\end{enumerate}
\end{lemma}

Intuitively, it is harder for Column to eliminate his actions as the number of Row's actions $m$ becomes larger, which is at the heart of the lemma. This lemma can be seen as a generalization of Lemma~\ref{lem: u_2_n}. The recurrence relation extends our observations from Section 3.1. By employing the recurrence relation for $m=2$, we immediately verify that $E\left[U^C(2,n)\right]=\sum_{k=1}^n \frac{1}{k}=H_n$. For $m=3$, it gives $E\left[U^C(3,n)\right]=\sum_{k=1}^n \frac{H_k}{k}=\frac{H_n^2+H_n^{(2)}}{2}$, where $H^{(m)}_n\equiv 1+\frac{1}{2^m}+\ldots+\frac{1}{n^m}$ is the $n$-th generalized harmonic number of order $m$ and the last identity follows, say, from \cite*{alzer_series_2006}. Nonetheless, deriving closed-form solutions for $n \geq 4$ becomes challenging. Indeed, to our knowledge, this recurrence has not been studied before and has no known explicit solution. 

The bounds in part~\ref{lem: u_m_n_3} of the lemma can be rewritten as
\begin{equation*}
    n \cdot \Pr\big({\rm Poisson}(\ln{n})= m-1\big) \leq \E\left[U^C(m,n)\right] \leq n \cdot \Pr\big({\rm Poisson}(\ln{n})\leq m-1)\big),
\end{equation*}
where ${\rm Poisson}(\lambda)$ is a Poisson random variable with parameter $\lambda>0$. 
These bounds are derived recursively. For example, consider the upper bound. For $m=2$, the upper bound holds since $E\left[U^C(2,n)\right] = H_n \leq \ln{n}+1$ for any $n\geq 1$. Now, let $f(x)=\frac{1}{x}\left(\ln{x}+1\right)$, which is strictly decreasing for $x \geq 1$. We can then write:
\begin{equation*}
E\left[U^C(3,n)\right]=\sum_{k=1}^n \frac{E\left[U^C(2,k)\right]}{k}\leq \sum_{k=1}^n f(k) \leq f(1)+\int_{1}^n f(x)=1+\ln{n}+\frac{(\ln{n})^2}{2},
\end{equation*}
as desired. At each subsequent step, as we increase $m$, we can redefine the function $f(x)$ accordingly and show it is strictly decreasing for $x \geq 1$. Similar, albeit somewhat more intricate, arguments are needed to illustrate the lower bound stated in the lemma.

We conclude this subsection with the asymptotic analysis of undominated actions for both players. Consider Row, who has $m(n) \leq n$ actions. It is immediate to see that Row cannot eliminate any of her actions as $n$ goes to infinity, irrespective of the dependence of $m(n)$ on $n$. Indeed,
\begin{equation*}
    \Pr\left(U^R(m,n) < m \right) \leq m(m-1) \times 2^{-n} \to 0 \quad \text{as} \quad n \to \infty, \, m=m(n)\leq n,
\end{equation*}
where the inequality follows from the fact that for Row, there are $m(m-1)$ pairs of strategies and each strategy can dominate another with probability $(\frac{1}{2})^n$.\footnote{A variant of this observation for $m=n$ is also stated as Proposition 5 in \cite{pei_rationalizable_2019}.}


The asymptotic analysis for Column, who has a greater number of actions, is more delicate. Whether Column can significantly alter a game by eliminating his actions in the first iteration depends on the relative sizes of the players' action sets. First, when Row has relatively few actions, $m(n)=o(\ln{n})$, the proportion of Column's undominated actions converges to zero asymptotically. In other words, such $m(n) \times n$ games are greatly simplified. We show this by first illustrating that the bounds in part~\ref{lem: u_m_n_3} of the lemma are asymptotically equivalent when $m(n)=o(\ln{n})$, and then noticing that the lower bound $\frac{(\ln{n})^{m-1}}{(m-1)!n}=\Pr\big({\rm Poisson}(\ln{n})= m-1\big)$ for the proportion of Column's undominated actions converges to zero for any $m=m(n)$---intuitively, the probability that ${\rm Poisson}(\ln{n})$ equals the specific value $(m(n)-1)$ becomes negligible.



A different picture emerges when Row's action set is large, namely when $m(n)=\log_2{n}+\omega(1)$. In this case, asymptotically, almost all of Column's actions are undominated. When $m$ grows, as discussed above, the lower bound for the proportion of Column's undominated actions converges to zero and is, hence, less useful. However, we can construct a different, more useful bound. For any given Column's action, there are $n-1$ other actions that may strictly dominate it, each with probability $\frac{1}{2^m}$. Using a union bound, $\frac{\E\left[U^C(m,n)\right]}{n} \geq 1-\frac{n-1}{2^m}$. When $m(n)=\log_2{n}+\omega(1)$,
\begin{equation*}
    \lim\limits_{n\rightarrow \infty} \frac{\E\left[U^C(m,n)\right]}{n}=1.\footnote{We can use the same union bound to show that for $m=2\log_2{n}+\omega(1)$, Column cannot delete any action at all. Namely, $\E\left[U^C(m,n)\right] = n-o(1)$.}
\end{equation*}

Figure~\ref{fig:p_m_n} summarizes these asymptotic observations by depicting the proportion of Column's undominated actions for large $n=10^9$ as a function of $m$ (the solid black line) along with its bounds (encompassing the region shaded in orange). As can be seen, our bounds are accurate for $m=o(\ln{n})$ and  $m=\log_2{n}+\omega(1)$ and provide a narrow band for the majority of relevant cases of $m=m(n) \leq n$.\footnote{We conjecture that for the small subset of intermediate $m(n)$ not covered formally by our analysis, the proportion of Column's undominated actions is very close to its upper bound $\Pr\big({\rm Poisson}(\ln{n})\leq m-1)\big)$. This conjecture is confirmed by our numerical exercises. It suggests that, for almost all intermediate $m(n)$, the proportion of Column's undorminated actions is large and close to one.} 


\begin{figure}[H]
    \centering
    \includegraphics[width=0.57\linewidth]{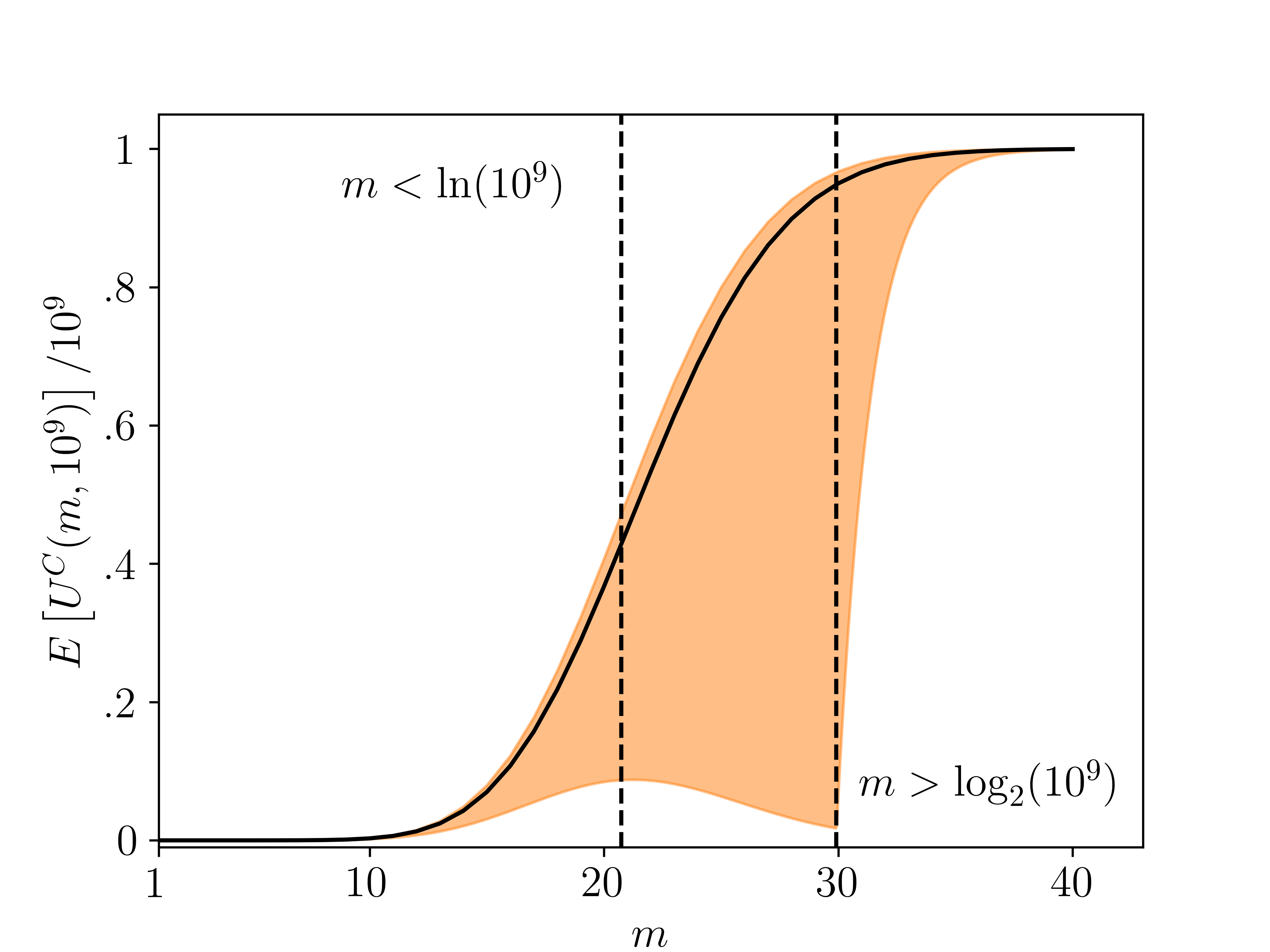}
    \captionsetup{justification=centering}
    \caption{Expected proportion of Column's undominated actions $\E\left[U^C(m,10^9)\right]/{10^9}$ with bounds (in orange)}
    \label{fig:p_m_n}
\end{figure}

\subsection{Dominance Solvability}

When analyzing the probability of dominance solvability, an additional enumerative issue emerges in general $m \times n$ games. In $2 \times n$ games, there are at most $3$ iterations. When, say, Column eliminates his actions first, the Row's payoffs in the induced, restricted game are still independent. And the third iteration, when it exists, corresponds to a simple maximization. In contrast, in general games, there can be many iterations.\footnote{The maximum number of iterations for an $m$ by $n$ game is $2m-1$ when $m<n$, and $2m-2$ if $m=n$.} These introduce non-trivial correlations: the fact that, say, a Column's action is not eliminated in the first iteration provides information on the payoffs it can generate. This information cannot be ignored when considering the third iteration. In particular, we cannot generally emulate the construction underlying the calculation of $\pi(m,n)$ obtained for $m=2$, which effectively considered each player's eliminated actions in isolation. In this subsection, we analyze $\pi(m,n)$ by taking into direct account the iterative nature of deletion. 

Our analysis in Section 4.1 suggests that for $m(n) \times n$ games, when the number $m(n)$ of Row's actions is relatively large, namely when $m(n)=\log_2{n}+\omega(1)$, Row cannot eliminate any of her actions and almost all of Column's actions survive in his first iteration with high probability. As a result, the probability a game is dominance solvable is vanishingly small when both action sets grow at these rates. In particular, perhaps confirming common wisdom, $n \times n$ games are rarely simplified for large $n$.

What is possibly less transparent is that the probability a game is dominance solvable vanishes quickly as long as \textit{any} action set grows, irrespective of their relative sizes. Indeed, in the previous section we showed that even seemingly simple $2 \times n$ games are rarely dominance solvable. Theorem~\ref{thm:p_m_n} below generalizes this statement to any $m(n) \times n$ games with $m(n) \leq n$. That is, even though particular games can be greatly simplified in just one iteration as we showed in Section 4.1, such games are not solvable with high probability.

\begin{theorem}
\label{thm:p_m_n}
There exist $C_1,C_2>0$ such that, for any $m \leq n$,
\begin{align*}
    n^{-(m-1)} \leq \pi(m,n) \leq C_1 \cdot n^{-C_2m}.
\end{align*}
In particular, $\pi(n,n)\leq n^{-\left(\frac{1}{3}-o(1)\right)n}$.
\end{theorem}

Despite the complex iterative nature of elimination, we can characterize an asymptotically tight bound for the convergence rate of the probability of dominance solvability.\footnote{Formally, this theorem can be restated as $\pi (m,n)=n^{-\Theta(m)}$, where $\Theta(m)$ denotes functions of $m$ that are asymptotically bounded both above and below by a linear function in $m$.} For games with action sets of comparable sizes, and particularly for balanced $n \times n$ games, this estimate is significantly more accurate than the crude union bound from Section 4.1:
\begin{equation*}
\pi(n,n) \leq \Pr\left(U^R(n,n) < n \right) + \Pr\left(U^C(n,n) < n \right) \leq 2n(n-1) \times 2^{-n}=\mathcal{O}\left(\frac{n^2}{2^n}\right).
\end{equation*}

We discuss the ideas guiding the proof in Section 4.5. Intuitively, the lower bound corresponds to the probability that Column has a strictly dominant action, in which case the game is strict-dominance solvable. The upper bound is derived from the consideration of two cases. If an $m \times n$ game is sufficiently imbalanced, we bound $\pi(m,n)$ by the probability $\Pr\left(S^R(m,n) < m \right)$ that Row deletes at least one action in the iterative procedure. Otherwise, if a game is roughly balanced, Row needs to eliminate many actions, not just one, when the game is dominance solvable. We account for the iterative elimination procedure to derive a bound in that case.

As compared with games exhibiting a unique pure equilibrium, dominance-solvable games are still vanishingly rare. Indeed, \cite{powers_limiting_1990} implies that the distribution of the number of pure-strategy Nash equilibria approaches the Poisson distribution $\rm Poisson(1)$ with mean 1 as both players' action sets expand. In addition, if one lets the number $n$ of Column's actions go to infinity while keeping the number $m$ of Row's actions finite, the limiting distribution of the number of pure-strategy Nash equilibria is a binomial random variable $B(m,1/m)$. In either case, asymptotically, a random game has exactly one pure-strategy Nash equilibrium with a positive probability. Hence, the probability of strict-dominance solvability conditional on there being a unique pure-strategy Nash equilibrium converges to zero.


\subsection{Conditional Iterations} 

As discussed in the introduction, solvable games are used for the implementation of desirable outcomes in a variety of applications due to their perceived simplicity and robustness. Nonetheless, empirically, people seem to have difficulty applying more than two or three iterations. We therefore ask what is the number of iterations players need to go through in generic games, conditional on dominance solvability.

Ex ante, it is unclear how the number of actions of each player affects the number of deletion rounds solvable games entail on average. Indeed, if a game is dominance solvable, the number of actions eliminated at each iteration could affect the number of iterations needed. Our question here regards the identification of the most pervasive way by which dominance solvability is obtained. 
As discussed in Section 4.2, the analysis of the general case is challenging since remaining games after each round of deletion exhibit non-trivial correlations between the conditional payoff distributions corresponding to different actions. Consequently, when examining $m \times n$ games with $m \geq 3$, we rely on simulations to compute the conditional number of iterations.

Specifically, we simulate $10^{6}$ games for each game dimension. 
Panel~(\subref{fig: I}) of Figure~\ref{fig: m_n} illustrates the resulting number of conditional iterations for various $m \times n$ games. As can be seen, the number of iterations required grows with the number of actions both players can use. In particular, for $n \times n$ games, this growth appears rapid and nearly linear, suggesting a large number of iterations, even for relatively small games. For instance, for $10 \times 10$ games, on average, more than $7$ iterations are required conditional on dominance solvability.

\subsection{Surviving  Actions}

Even though dominance solvability is scarce, iterated deletion of dominated actions might still be effective in simplifying a game as long as the set of actions surviving the elimination procedure is relatively small. In Section 4.1, we showed that the effectiveness of the first iteration depends on the relative size of players' action sets. For $m \times n$ games with relatively small $m=o(\ln{n})$, the proportion of undominated actions for Column converges to zero asymptotically. This provides a silver lining to Theorem~\ref{thm:p_m_n}---$m\times n$ games with relatively small $m$ are significantly simplified even after the first iteration. Can subsequent iterations simplify games even further?

It turns out that subsequent iterations are not effective asymptotically. In other words, one iteration is asymptotically sufficient. We prove this by showing that Row eliminates at least one of her actions with vanishing probability. Therefore, as for $2 \times n$ games, asymptotically, there is no difference between the number of Column's undominated actions and the number of Column's surviving actions.

In contrast, for $m \times n$ games with larger $m=\log_2{n}+\omega(1)$, almost all actions are undominated, as $n$ grows. Asymptotically, such games cannot be simplified at all.

Theorem~\ref{thm:s_m_n} summarizes our results pertaining to the effectiveness of the iterated elimination procedure in simplifying general games.

\begin{theorem}
\label{thm:s_m_n}
Consider a random game $G(m,n)$. Then, for any $n \geq m=m(n) \geq 2$,
\begin{enumerate}
    \item $\lim\limits_{n\rightarrow \infty}\Pr\left(S^R(m,n) < m \right)=0$, i.e. Column can proceed with at most one iteration asymptotically;
    \item for small $m=o(\ln{n})$, we have $\lim\limits_{n\rightarrow \infty} \dfrac{(m-1)!}{(\ln{n})^{m-1}} \times \E\left[S^C(m,n)\right] = 1$, and therefore $\lim\limits_{n\rightarrow \infty} \dfrac{\E\left[S^C(m,n)\right]}{n}=0$;
    \item for larger $m=\log_2{n}+\omega(1)$, we have $\lim\limits_{n\rightarrow \infty} \dfrac{\E\left[S^C(m,n)\right]}{n}=1$.\footnote{By our analysis in Section 4.1, for $m=2\log_2{n}+\omega(1)$, all Column's actions survive the iterated elimination, or formally $\E\left[S^C(m,n)\right] = n-o(1)$.}
\end{enumerate}
\end{theorem}

In Section 4.5, we provide more detailed bounds on $\Pr\left(S^R(m,n) < m \right)$ and discuss their use for the proof of this theorem. Intuitively, the first part holds for fixed $m=2$ from our analysis in Section 3. As $m(n)$ grows, there are two competing forces. On the one hand, Row has more actions, so it is easier to eliminate at least one of them. On the other hand, it is harder to eliminate any particular action of Row since, by Lemma~\ref{lem: u_m_n}, Column deletes fewer actions in his first iteration as the number of Row's actions grows. The proof illustrates that the latter force asymptotically dominates the first, irrespective of the dependence of $m(n)$ on $n$. The last two parts of the theorem follow immediately from our analysis in Section 4.1 and the theorem's first part.

Despite the iterative nature of the deletion procedure, our analysis suggests that players cannot go far beyond their first iteration \textit{unconditionally}. This is comforting news for $m(n) \times n$ games with relatively small $m(n)$. These games can be greatly simplified in only a few iterations. Thus, even taking into account experimental evidence suggesting the limited ability of individuals to go through an extensive number of deletion iterations, when games are highly imbalanced, the iteration procedure can be quite useful in simplifying games. This insight reverses when players' action sets are of comparable sizes. For larger games, regardless of the number of iterations one contemplates as plausible, the iterative elimination procedure does not alter substantially the game players need to consider.\bigskip

Figure~\ref{fig: m_n} uses simulations to depict the objects of our analysis for $m\times n$ games, $m=3,4$ and $n \in [50]$, as well as $n\times n$ games, $n \in [10]$.\footnote{As already described, we use $10^{6}$ simulations for each game size, in addition to exact values corresponding to the $m=2$ case discussed in the previous section. For $n \times n$ games, we restrict $n$ to be no larger than $10$ for computational reasons.} 

For $m\times n$ games with $m \geq 3$, all qualitative conclusions resemble those derived for the $m=2$ case. In particular, the probability of strict-dominance solvability converges to zero, albeit more rapidly. We already described the pattern that emerges for the number of conditional iterations: they grow with $m$, and more rapidly so when players' action sets are of the same size. Last, the number of surviving Column actions grow with $m$.\footnote{Furthermore, the distribution of the number of surviving actions appears approximately normal even when fixing $n$ at $50$. Intuitively, since Row cannot eliminate any of her actions asymptotically, we expect this distribution to be close to the one of undominated actions. The number of undominated actions can be represented as the sum of many rare ``almost independent'' indicator random
variables, each corresponding to whether a particular action is undominated or not. The so-called ``Poisson Paradigm'', see \cite{alon_2016}, would suggest that the distribution of the number of undominated actions approximate a Poisson distribution, which in turn approaches a normal distribution as the Poisson mean goes to infinity.}

Figure~\ref{fig: m_n} also highlights the contrast between games in which both players' action sets expand and those in which one of the players has a fixed action set. The likelihood of dominance solvability $\pi(n,n)$ vanishes quickly, standing at less than $ 5\%$ when $n\geq 7$, the number of conditional iterations $I(n,n)$ increases indefinitely and exceeds 3 starting from $n=4$. The expected number of surviving actions coincides with the full action set even for small $n$, namely any $n\geq 5$.

\subsection{Structure of Proofs}

In this section, we sketch for the interested reader the arguments underlying our main results. Detailed proofs appear in the Appendix. We first discuss Proposition~\ref{prop:s_m_n} that, together with our analysis in Section 4.1, implies Theorem~\ref{thm:s_m_n}. We then use the proposition to illustrate the ideas generating the proof of Theorem~\ref{thm:p_m_n}. 

\begin{prop}
\label{prop:s_m_n}
For any $m \leq n$, 
\begin{equation*}
\Pr\left(S^R(m,n) < m \right) \leq m(m-1)\cdot\left(\dfrac{m}{n}\right)^{\frac{m-1}{4}}.
\end{equation*}
\end{prop}

The intuition for this proposition is the following. The number of strategy pairs for Row is $m(m-1)$. We can then use symmetry to bound $\Pr\left(S^R(m,n) < m \right)$ by $m(m-1)$ times the probability that Row's second action strictly dominates her first \textit{after} the first round of Column's elimination. Namely, we restrict the problem to a collection of particular ``subgames'' in which Row considers only two actions. 

To complete the heuristic argument underlying this proposition, notice that for any pair of Row's actions, Row's second action does not dominate her first after Column's first elimination round only if, for some $j \in [n]$, $j$-th Column's action is undominated and $r_{1j}>r_{2j}$---denote by $\mathcal{R}_j\equiv\{r_{1j}>r_{2j}\}$ the corresponding event. Without loss of generality, we can fix $c_{1\cdot}=(n,n-1,\ldots,1)$. For Column's $j$-th action to be undominated, it suffices for it to deliver the largest payoff among the first $j$ of Column's actions for at least one of Row's action. Formally, it corresponds to the event $\mathcal{C}_j\equiv \bigcup_{i \geq 2} \{c_{ij}=\max_{k\leq j} c_{ik}\}$. In the proof, we show that these events $\{E_j \equiv \mathcal{C}_j \bigcap \mathcal{R}_j\}_{j \in [n]}$ are mutually independent. This observation allows us to bound the probability that Row's second action strictly dominates her first one after the first Column's iteration by 
\begin{equation*}
\Pr\left(\bigcap_{j \in [n]}\overline{E}_j\right)=\prod_{j \in \lbrack n \rbrack} \left(1- \Pr\left(E_j\right)\right)\leq \exp(-\sum_{j \in \lbrack n \rbrack}\Pr\left(E_j\right)). 
\end{equation*}
Since each event $E_j$, $j \in [n]$, itself represents a finite union of events, we apply Bonferroni's inequality to find the lower bound on $\Pr\left(E_j\right)$ and use it with other well-known inequalities to obtain the proposition's claim. In most cases, the proposition's bound is sharper than the crude union bound $m(m-1)\cdot2^{-n}+n(n-1)\cdot2^{-m}$ for the probability of deleting any action in the first iteration.\footnote{In fact, it is possible to use the inclusion–exclusion principle to obtain an exact expression for $\Pr(E_j)$. In the Appendix, we use it to prove that for any \textit{finite} $m$, we have $\sum_{j \in \lbrack n \rbrack}\Pr\left(E_j\right)=\mathcal{O}\left((m-1)/2 \cdot \ln{n} \right)$, and thus $\Pr\left(S^R(m,n) < m \right)=\mathcal{O}\left(n^{-{(m-1)/2}}\right)$, which is tight when $m=2$.}

The proposition, together with the union bound above, allows us to immediately conclude that in $m(n) \times n$ games with an arbitrary specification of $m(n) \leq n$, Row cannot eliminate any of her actions asymptotically. Formally, $\lim\limits_{n\rightarrow \infty}\Pr\left(S^R(m,n) < m \right)=0$, as desired. The last two parts of Theorem~\ref{thm:s_m_n} follow from our results in Section 4.1.

As concerns the proof of Theorem~\ref{thm:p_m_n}, we obtain the lower bound by noting that an $m \times n$ game is strict-dominance solvable when Column has a strictly dominant action. This event occurs with probability $n^{-(m-1)}$. 

The derivation of the upper bound is more intricate and divided into two separate cases. First, we focus on sufficiently imbalanced $m \times n$ games, for which $m$ is relatively small compared to $n$. In that case, we bound $\pi(m,n)$ by the probability that Row eliminates at least one action: $\pi(m,n)=\Pr\left(S^R(m,n) = 1 \right) \leq \Pr\left(S^R(m,n) < m \right)$. We use Proposition~\ref{prop:s_m_n} to get the desired bound. 

Second, we consider relatively balanced $m \times n$ games, for which Row has almost as many actions as Column. Then, the iterative elimination procedure needs to remove many of Row's actions, not just one. Because the corresponding bound for $\Pr\left(S^R(m,n) < m \right)$ becomes less relevant when dealing with games that feature many Row actions and specifically $n \times n$ games, we use an alternative argument. 

Specifically, we first show that for perfectly balanced $n \times n$ games, $\pi(n,n)\leq n^{-\Theta(n)}$. An $n \times n$ game is solvable if and only if exactly $(n-1)$ actions for Row or Column are eliminated in the iterative procedure. We obtain the stated upper bound by considering a relaxed problem of finding the probability of eliminating at least an $\alpha$-fraction of Row's or Column's actions, $\alpha \in (0,1-1/n]$. Intuitively, we expect the latter probability to be sufficiently small for large enough $\alpha>0$. 

For the relaxed problem, we iteratively (if needed) eliminate strictly dominated actions in an arbitrary order and stop exactly when an $\alpha$-fraction of Row's or Column's actions is eliminated. By symmetry with respect to players' labels, it is without loss of generality to suppose that at the stopping point, Row has eliminated approximately an $\alpha$-fraction of her actions, while Column has eliminated a smaller fraction of actions. 


These iterations introduce non-trivial correlations between conditional payoffs corresponding to different actions. However, in the next step, we simplify it further to establish our stated bounds. Consider the final subgame at which our iterative process above stops, when approximately an $\alpha$-fraction of Row's actions have been eliminated. Suppose action $r$ of Row has been eliminated and let $X$ denote the set of Column's actions in this subgame, his surviving actions. In the original game, action $r$ would also be dominated for Row were Column's actions restricted to $X$. Recall that the number of actions $r$ as such accounts for approximately an $\alpha$ fraction of Row's actions. Furthermore, the set $X$ accounts for at least a $(1-\alpha)$-fraction of Column's actions. To derive the desired upper bound for $\pi(n,n)$, we therefore assess the probability that at least a fraction of $\alpha$ of Row's actions is dominated in a subgame with at least $(1-\alpha)$ of Column's actions. This allows us to circumvent the correlations between payoffs in games selected through the iterative process. Our stated bound $\pi(n,n)\leq n^{-\left(\frac{1}{3}-o(1)\right)n}$ is obtained by picking $\alpha=1/3$.\footnote{Various values of $\alpha$ could generate our desired bound.}


We next connect imbalanced $m \times n$ games, $m\leq n$, to the already examined perfectly balanced $m \times m$ games. When considering the likelihood of dominance solvability, it is important to note that not every subgame of a strict-dominance solvable $m \times n$ game is solvable in itself.\footnote{Furthermore, in general, we cannot partition a non-random strict-dominance solvable game into two non-intersecting strict-dominance solvable subgames of given dimensions. That is, $\pi (m,n)$ is not component-wise sub-multiplicative.} Nonetheless, we can pick \textit{particular} subgames of any size that are solvable. In particular, any subgame generated by eliminating some of the players' dominated actions would be dominance solvable as well. Therefore, we have the following lemma.

\begin{lemma}
\label{lem: subgame} Consider a  strict-dominance solvable $m \times n$ game, where $m, n \geq 1$. Then, for any $m^{\prime }\in \lbrack m \rbrack$, $n^{\prime }\in \lbrack n \rbrack$, there exists a strict-dominance solvable $m^{\prime }\times
n^{\prime }$ subgame.\footnote{As an immediate corollary, for any fixed $m \geq 2$, $\pi(m,n) \leq m \cdot \pi(m-1,n)$, and hence convergence rates are weakly increasing.}
\end{lemma}

For all remaining cases of the number $m$ of Row's actions, we use Lemma~\ref{lem: subgame} to bound  $\pi(m,n) \leq \binom{n}{m} \cdot \pi(m,m)$, and then apply $\pi(m,m)\leq m^{-\left(\frac{1}{3}-o(1)\right)m}$ to obtain our results. 

\section{Alternative Distributional Assumptions} \label{AlternativeDistributions}

Our analysis focused on random games. However, strategic interactions that have received attention in the literature, theoretically and empirically, are inspired by applications, and could be far from random. One immediate concern could be that real-world interactions correspond to games that are more amenable to the iterated elimination procedure. To what extent are our qualitative results driven by our uniform determination of game structures? In this section, we present data from lab experiments and from simulations indicating that our results 
hold for a wide variety of alternative distributional assumptions that correspond to commonly-studied strategic interactions.

\subsection{Comparison of Lab and Random Games}

We start by comparing random games with those played in lab experiments. Our analysis below uses data on initial play within $86$ symmetric $3 \times 3$ games from $6$ different experiments collected by \cite{wright2014level} and utilized by \cite{fudenberg_predicting_2019} (in addition to simulated random games as analyzed in the current paper).

\begin{figure}[h]
    \centering
\includegraphics[width = 0.48\linewidth]{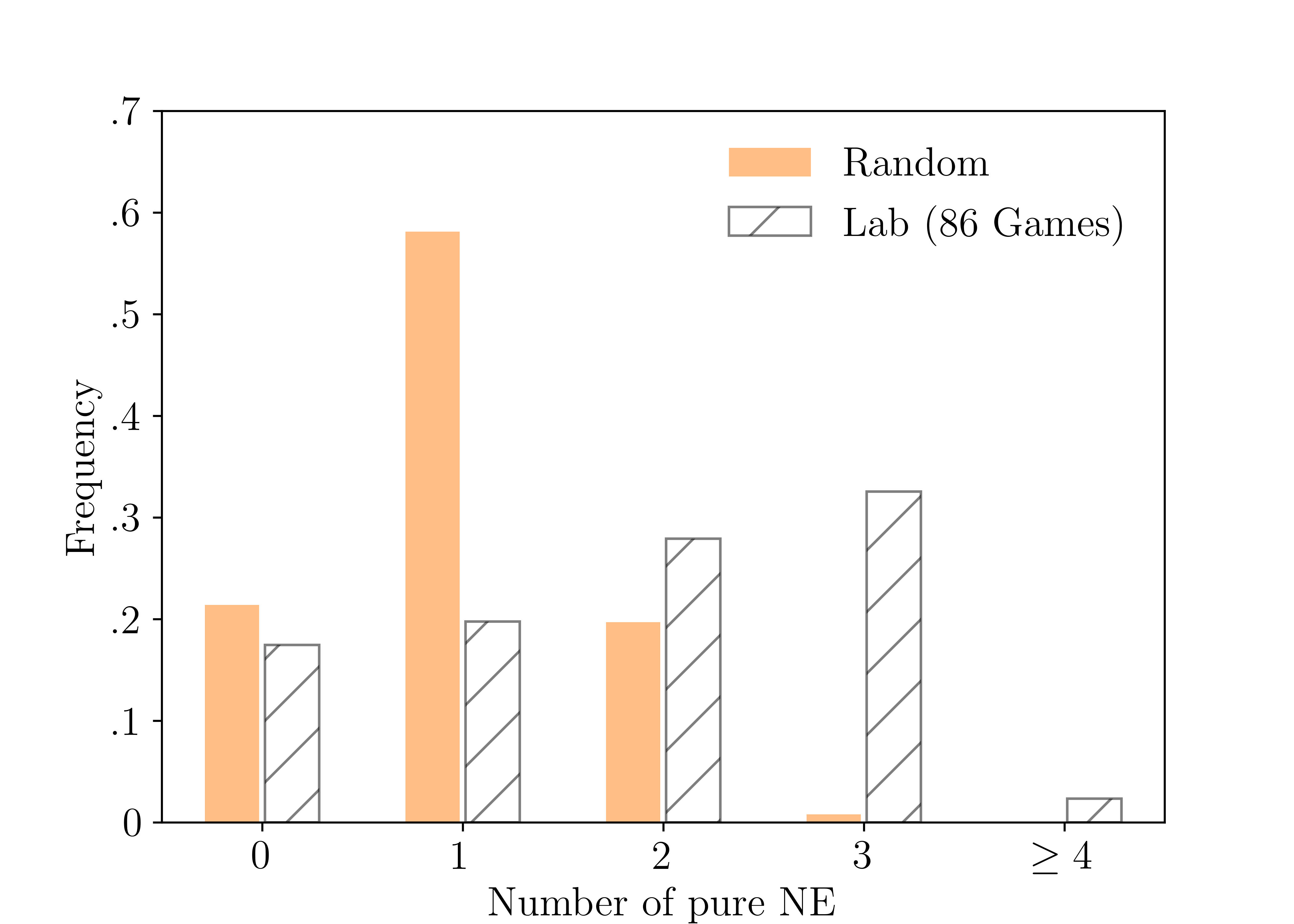}   
\includegraphics[width = 0.48\linewidth]{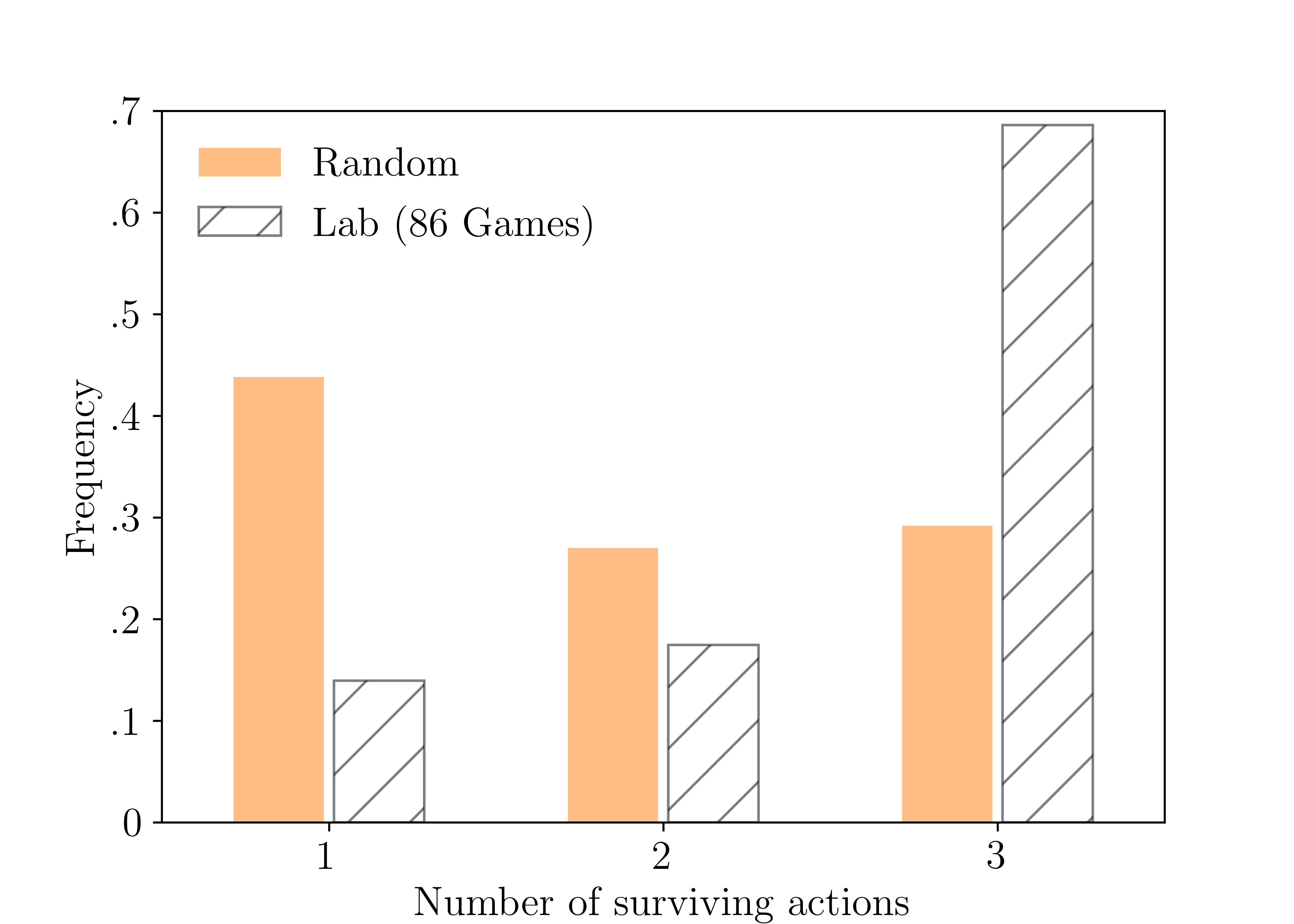}
\captionsetup{justification=centering}
\caption{Frequency of pure-strategy Nash equilibria ({\color{princetonorange}Left Panel}) and surviving actions ({\color{princetonorange}Right Panel}) in random and lab games}
\label{fig: lab vs random}
\end{figure}

Figure~\ref{fig: lab vs random} depicts the frequency of pure Nash equilibria on the left panel and the number of actions surviving elimination of strictly-dominated actions, in both random $3 \times 3$ games that we study, and those collected from the experimental literature. The figure already suggests that, if anything, randomly-generated games tend to be ``strategically simpler'' than experimental games---they feature fewer pure Nash equilibria and greater impact of the elimination procedure in that fewer actions survive.\footnote{In fact, our random games are also simpler than random games from \cite{fudenberg_predicting_2019}. Intuitively, indifferences that they allow for generate larger best-response sets and hinder the deletion of strictly-dominated strategies.} As we soon show, this increased simplicity translates to the three dimensions of complexity we inspect throughout the paper. 

\subsection{Alternative Classes of Games} 

In this subsection, we compare the three dimensions of complexity inspected throughout the paper across random games exhibiting commonly-studied structures, lab experimental games, and our uniformly random games. We focus the comparison on balanced games, which allows us to consider symmetric games as well.\footnote{We see qualitatively similar results for imbalanced games, see details in the Online Appendix.} 

In addition to standard games $G(n,n)$ with payoff matrices $(R,C)$ studied before, we also analyze randomly generated games with the following constraints:\footnote{The constraints imposed by these classes of games and the ones discussed next are ordinal in nature. In particular, their analysis is, again, distribution-free. In all our simulations, we randomly generate games as we do for our baseline games $G(n,n)$, but impose additional constraints corresponding to the various structures described here. \cite{pei_rationalizable_2019} consider point-rationalizability in similar classes of games, see also our discussion in Section \ref{Rationalizability}.}
\begin{enumerate}
    \item \textit{symmetric} games $G^{sym}(n)$, where payoff matrices are transposes of each other, $(R, R^T)$;
    \item \textit{potential} or \textit{common interest} games Potential$(n)$, in which it is possible to capture both players' payoff matrices by a single matrix called an ordinal potential, so that payoff matrices are identical, $(R, R)$;
    \item \textit{constant-sum} games Const-Sum$(n)$: $(R, c-R)$, where $c$ is an arbitrary constant matrix.
\end{enumerate}

Furthermore, we consider random games with strategic complementarities, i.e. games having nondecreasing best-response functions. Formally, a game has strategic complementarities if, given an order on players' strategies, an increase in one player's strategy induces other players to increase their strategies, see \cite{topkis1979equilibrium}, \cite*{bulow1985multimarket}, and \cite{vives1990nash}. 
Specifically, we consider random games as follows:
\begin{enumerate}
    \setcounter{enumi}{3}
    \item games \textit{with strategic complementarities} Strat-Complements$(n)$ with payoff matrices $(\bar{R},\bar{C})$, where random $\bar{R}$ is equal to $R$ conditional on it having a nondecreasing best-response function with respect to natural orderings $\{1,2,\ldots,n\}$ for both players and $\bar{C}$ is defined similarly;\footnote{The conditional distribution can be stated in terms of the unconditional one. It implies that best-response functions are chosen uniformly from 
    nondecreasing functions for given natural orderings.}
    \item \textit{symmetric} games \textit{with strategic complementarities} Strat-Complements$^{sym}(n)$ with payoffs $(\bar{R},\bar{R}^T)$, where $\bar{R}$ is defined as above.
\end{enumerate}

Figure~\ref{fig: alt_bal} depicts the three complexity dimensions we study for these simulated classes of games, and for the 86 symmetric lab games assembled by \cite{wright2014level} and discussed in the previous subsection. 


Panel~(\subref{fig: pi_alt_bal}) of the figure indicates that almost all of the above games yield even lower probability of dominance solvability. In particular, lab games correspond to a substantially lower probability of dominance solvability than random games of the same size. The one exception is games with strategic complementarities, which are somewhat more likely to be solvable. Nonetheless, even for those games, the probability of dominance solvability converges to zero rapidly, standing at less than $2\%$ for $8 \times 8$ games. 

Panel~(\subref{fig: I_alt_bal}) of the figure displays the number of iterations conditional on dominance solvability.  The random games we study are, to some extent, more complex in that respect, corresponding to a greater number of necessary iterations. Nonetheless, the differences are not vast. Furthermore, even for the ``simplest'' games in this respect, symmetric games with complementarities, the number of conditional iterations exceeds $2$ for $8 \times 8$ games.\footnote{We do not simulate larger games since the probability of dominance solvability is then very low across all structures we consider, and the number of simulations required to establish reasonable precision becomes prohibitively large.}


\begin{figure}[t!]
    \centering
    \begin{subfigure}[t]{.49\textwidth}
    \centering
    \includegraphics[width=\linewidth]{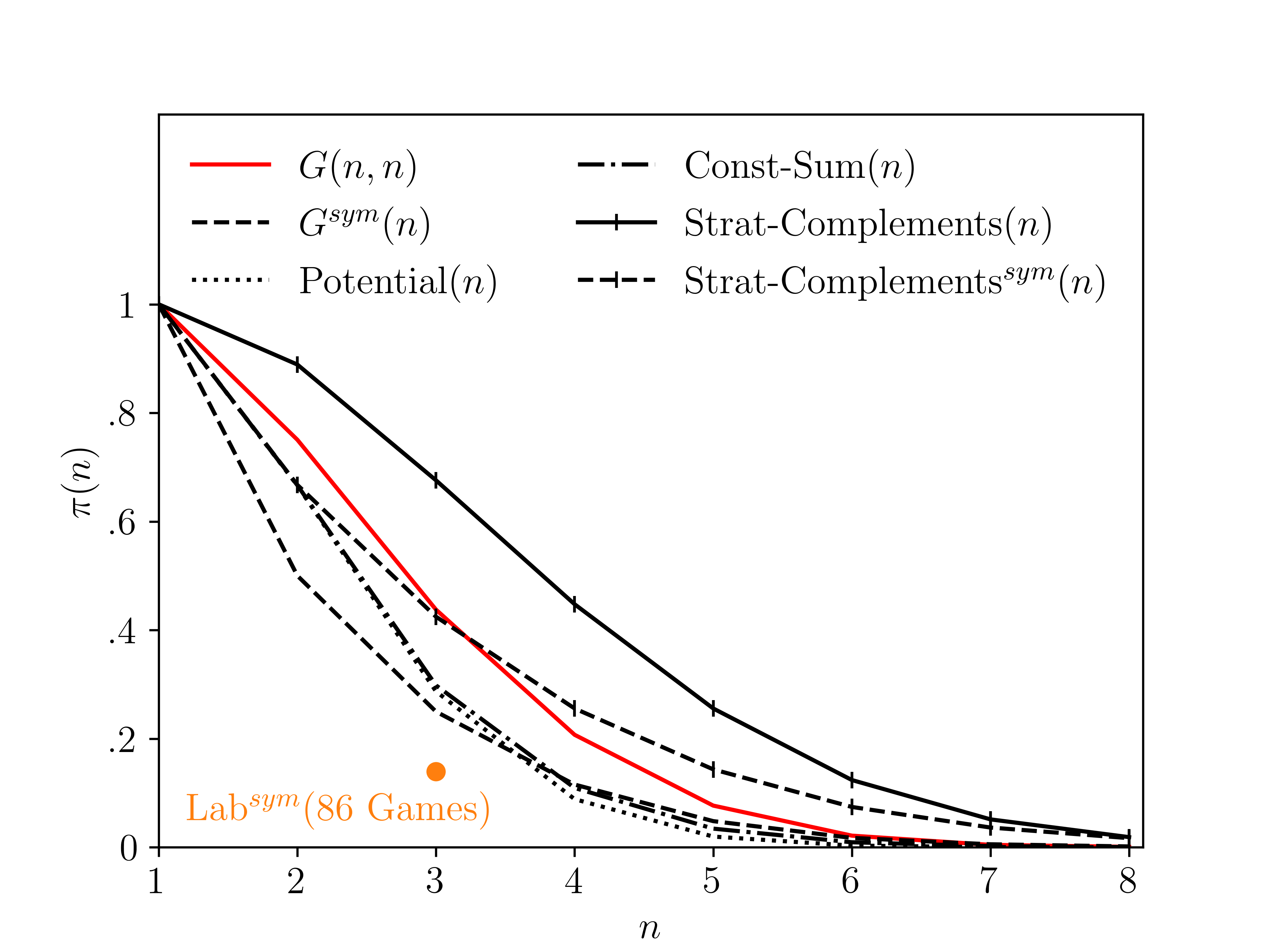}
            \caption{Probability of strict dominance}\label{fig: pi_alt_bal}
    \end{subfigure}
    \begin{subfigure}[t]{.49\textwidth}
    \centering
    \includegraphics[width=\linewidth]{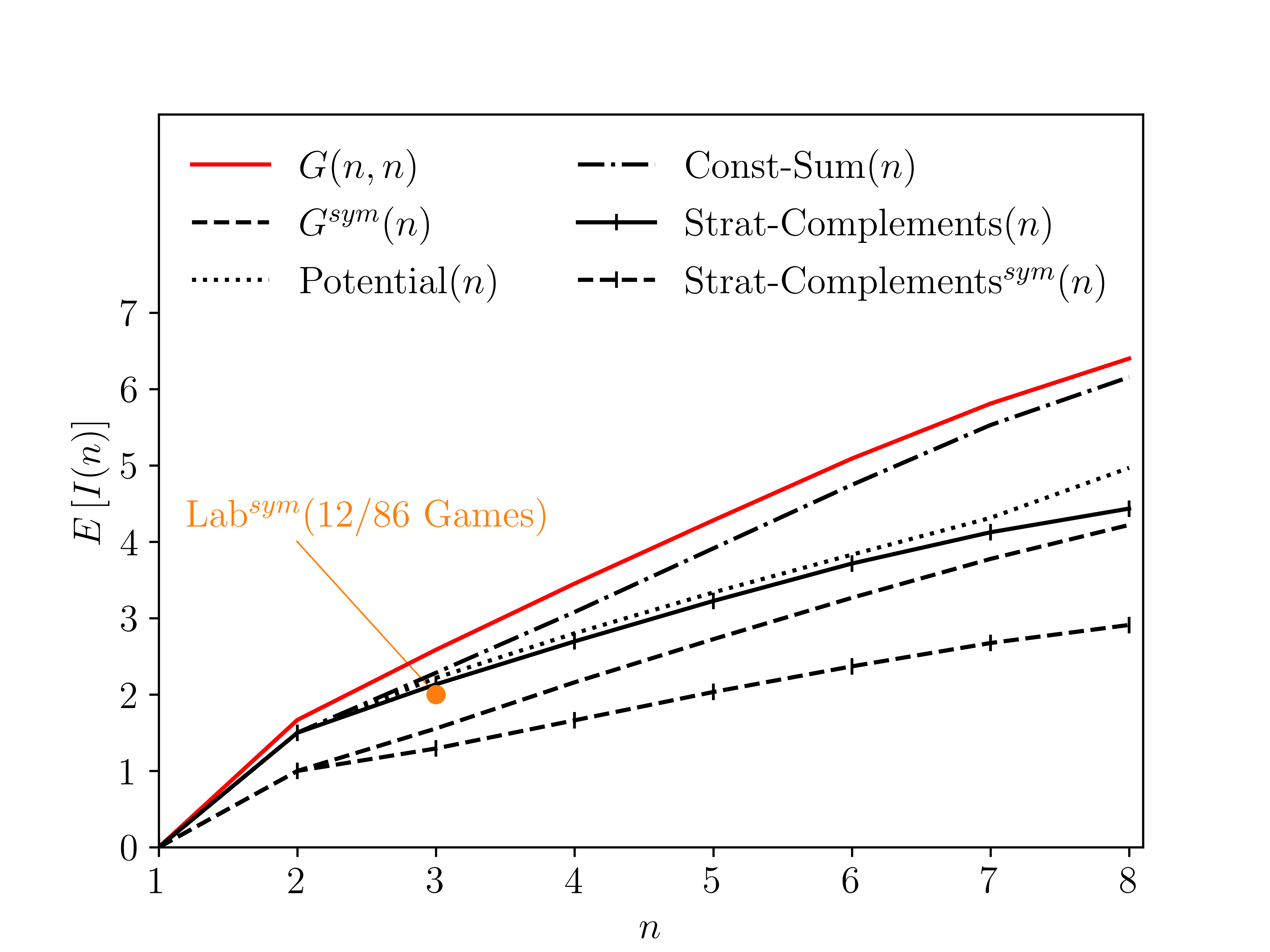}
    \caption{Conditional number of iterations}\label{fig: I_alt_bal}
    \end{subfigure}
    \par
    \medskip
    \par
    \begin{subfigure}[t]{\textwidth}
    \centering
    \vspace{0pt}
    \includegraphics[width=0.49\linewidth]{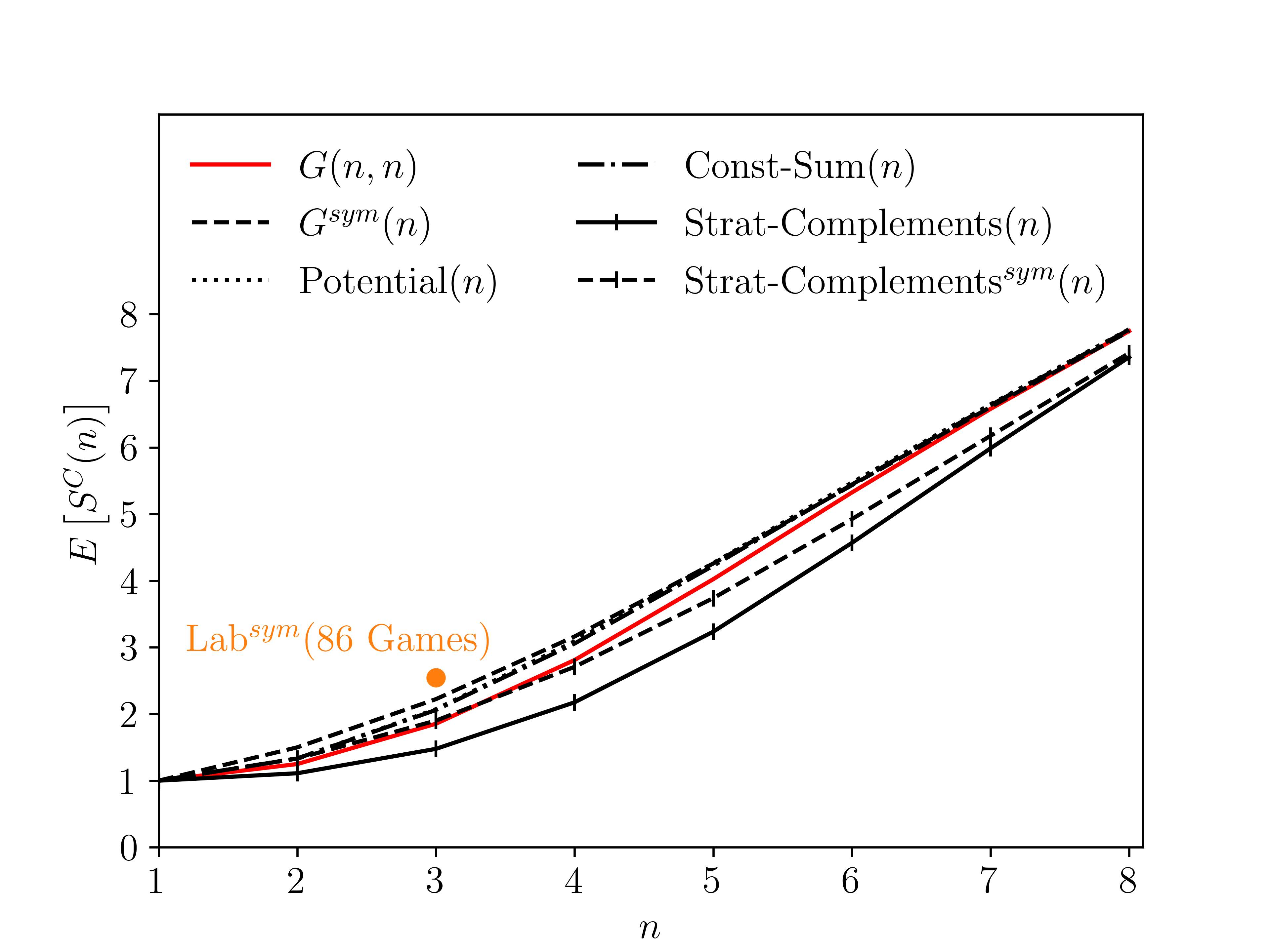}
    \caption{Surviving actions}\label{fig: S_alt_bal}
    \end{subfigure}
    \captionsetup{justification=centering}
    \caption{Three dimensions of dominance solvability in $n \times n$ games for alternative distribution assumptions}
    \label{fig: alt_bal}
    \end{figure}

Panel~(\subref{fig: S_alt_bal}) of the figure illustrates the number of surviving actions for either player when the iterated elimination procedure terminates. We see small differences across the different game structures. As already mentioned, the elimination procedure is somewhat less effective in lab games relative to the random games we study. It is also slightly less effective than in the other classes of games analyzed here. The number of surviving actions increases with the size of the game. For $8 \times 8$ games, in expectation, there is nearly nothing eliminated through the iterative procedure in either of the game classes we consider.


\section{Rationalizability and Mixed Strategies} \label{Rationalizability}


Throughout, we focus on elimination of strictly-dominated actions and consider domination only via \textit{pure }actions. Our notion of strict-dominance solvability is closely related to the rationalizability notion proposed by \cite{borgers_pure_1993}. In contrast to the traditional notion \citep{bernheim_rationalizable_1984, pearce_rationalizable_1984}, he considers only players’ ordinal preferences over strategy profiles to be common knowledge, but not their cardinal preferences. \cite{borgers_pure_1993} characterizes this ``robust'' notion of rationality in terms of a pure-strategy dominance property that, for generic games with distinct payoffs, coincides with the standard pure-strategy dominance relation. Furthermore, he shows that common belief in his rationality notion is outcome equivalent to the procedure of iterated elimination of dominated strategies.

Having said that, experimental evidence suggests that mixed strategies are more cognitively demanding (see \citealp{erev_predicting_1998, shachat_mixed_2002}). Our analysis of \cite{fudenberg_predicting_2019}'s data is consistent with this assertion---their experimental participants were roughly double as likely to play a dominated action when domination was via mixed strategies (see details in the Online Appendix).   
Such ``mistakes'' are likely to be compounded through the iterative elimination procedure, where each step feeds into the next. The suspicion that domination by mixed strategies is more challenging to identify is in line with the focus of much of the robust mechanism design literature on pure-strategy dominance. 

Nonetheless, there is a strong link between the set of surviving actions we identify and the set of \textit{rationalizable} actions, which account for domination by mixed strateges. Recall that the iterative elimination of actions strictly dominated by any arbitrary mixed strategy generates the set of \textit{rationalizable} actions. Since some actions may be strictly dominated only by mixed strategies and not by pure strategies, the set of rationalizable actions is, in general, a subset of the set of actions surviving iterated elimination of actions strictly dominated by pure actions that we study. The literature often considers the iterative elimination of never best responses against (surviving) pure strategies. This procedure culminates in a set of actions that are commonly termed \textit{point-rationalizable} \citep{bernheim_rationalizable_1984}.\footnote{Point-rationalizability is an ordinal concept, and hence its analysis is distribution-free.} Naturally, an action might not be a best response against any pure strategy and still be a best response with some (non-degenerate) belief about the opponent's strategy. Therefore, the set rationalizable actions---and thus, our set of actions surviving iterated elimination of actions strictly dominated by pure actions---is a superset of the set of point-rationalizable actions.



\cite{weinstein_effect_2016} considers deterministic games without indifferences and shows that, for sufficiently risk-averse players, the set of rationalizable actions coincides with the set of pure-strategy rationalizable actions that we consider. For sufficiently risk-loving players, the set of rationalizable actions coincides with the set of point-rationalizable actions. \cite{pei_rationalizable_2019} translate this result to random games.\footnote{They demonstrate existence of payoff distributions for which the set of rationalizable actions coincides with the set of pure-strategy rationalizable actions with high probability, and payoff distributions for which the set of rationalizble actions coincides with the set of point-rationalizable actions with high probability.} They also obtain the distribution of the number of point-rationalizable actions. In particular, for $m \times n$ random games, $m \leq n$, there is a unique point-rationalizable action profile with probability $\frac{m+n-1}{mn}=\Theta\left(m^{-1}\right)$. 
 
\cite{pei_rationalizable_2019}'s analysis focuses mostly on identifying the distribution of the number of point-rationalizable actions. The object of their study is then best responses and the techniques utilized, as well as their results, are different than ours. \cite{pei_rationalizable_2019} are silent about the iterative nature of point-rationalizability. Nonetheless, their bounds combined with ours offer insights on the probability of mixed-strategy dominance solvability. As mentioned, mixed-strategy dominance-solvable games are a superset of our pure-strategy dominance-solvable games and a subset of games with a unique point-rationalizable action profile. Since we showed that $\pi(m,n)=n^{-\Theta(m)}$, it follows that the probability of mixed-strategy dominance solvability is within the interval $[n^{-\Theta(m)},\Theta\left(m^{-1}\right)]$. 
Similarly, we can obtain asymptotic bounds on the number of Column's mixed-strategy rationalizable actions for any $m(n) \leq n$. For balanced $n \times n$ games, \cite{pei_rationalizable_2019} show that a lower bound on the number of point-rationalizable Column actions is given by $\sqrt{\pi n}/2$.\footnote{They also show that, for an arbitrary payoff distribution with a finite third moment, the number of rationalizable actions is close to $n$ with high probability as $n$ grows.} The resulting range for Column's mixed-strategy rationalizable actions is then $[\sqrt{\pi n}/2,n]$. For $m \times n$ games with fixed $m$, the bound on the mixed-strategy rationalizable Column's actions can be directly extended and the resulting range is $[\sqrt{\pi m/2},(\ln{n})^{m-1}/(m-1)!]$.

Figure~\ref{fig: mixed_bal} illustrates these bounds for $n \times n$ games and depicts our simulated variables of interest considering mixed-strategy dominance solvability, or rationalizability, for uniform and normal distributions. In addition, we consider a transformation of the uniform distribution using \cite{fudenberg_predicting_2019}'s estimated risk parameter, $\alpha_{FL}=0.41$.\footnote{Due to computational limitations, we use $10^{4}$ simulations when analyzing mixed strategies. The corresponding figure for imbalanced $3 \times n$ games is in the Online Appendix.} 

\begin{figure}[t!]
    \centering
    \begin{subfigure}[t]{.49\textwidth}
    \centering
    \includegraphics[width=\linewidth]{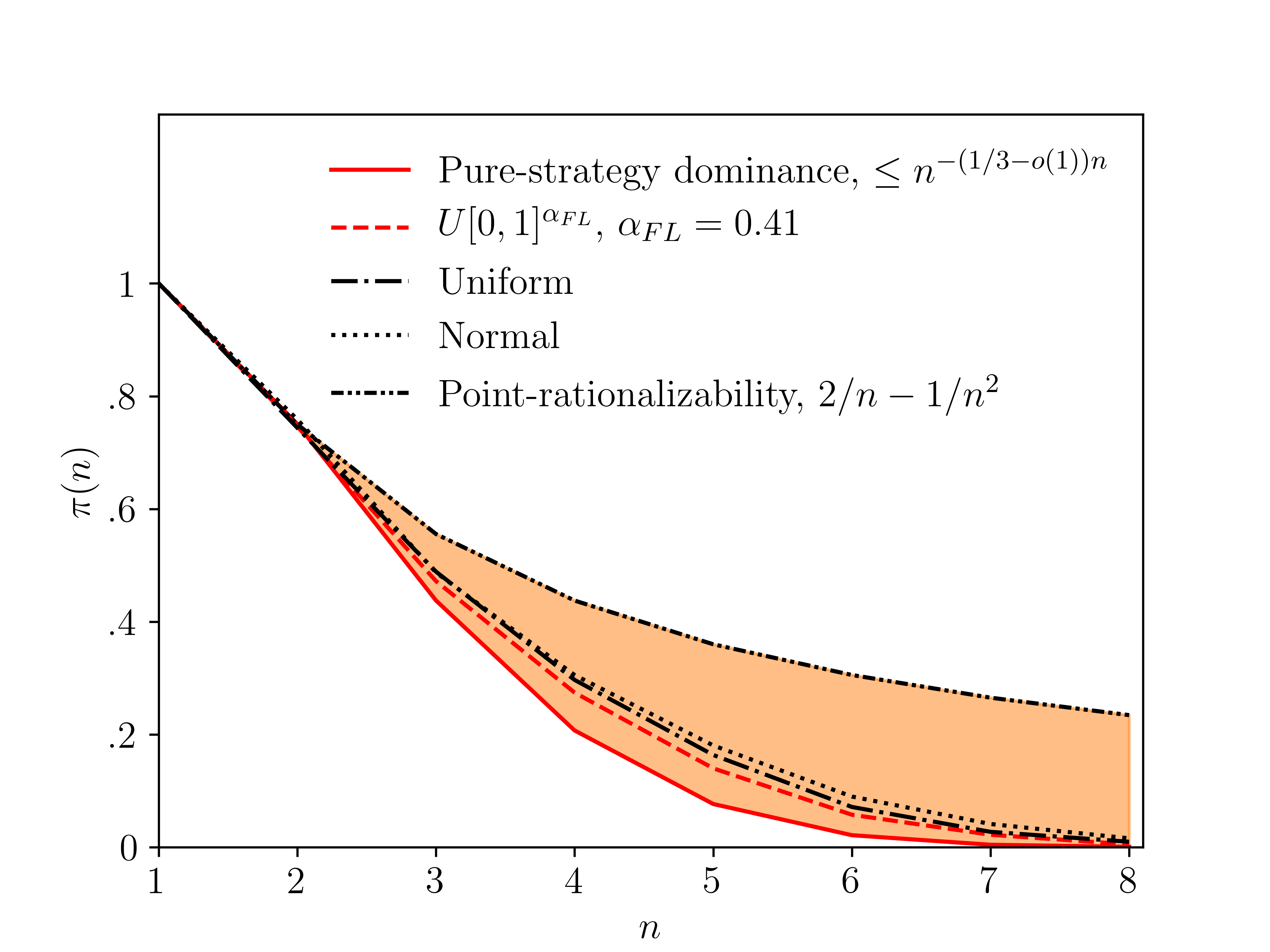}
            \caption{Probability of strict dominance}\label{fig: pi_mixed_bal}
    \end{subfigure}
    \begin{subfigure}[t]{.49\textwidth}
    \centering
    \includegraphics[width=\linewidth]{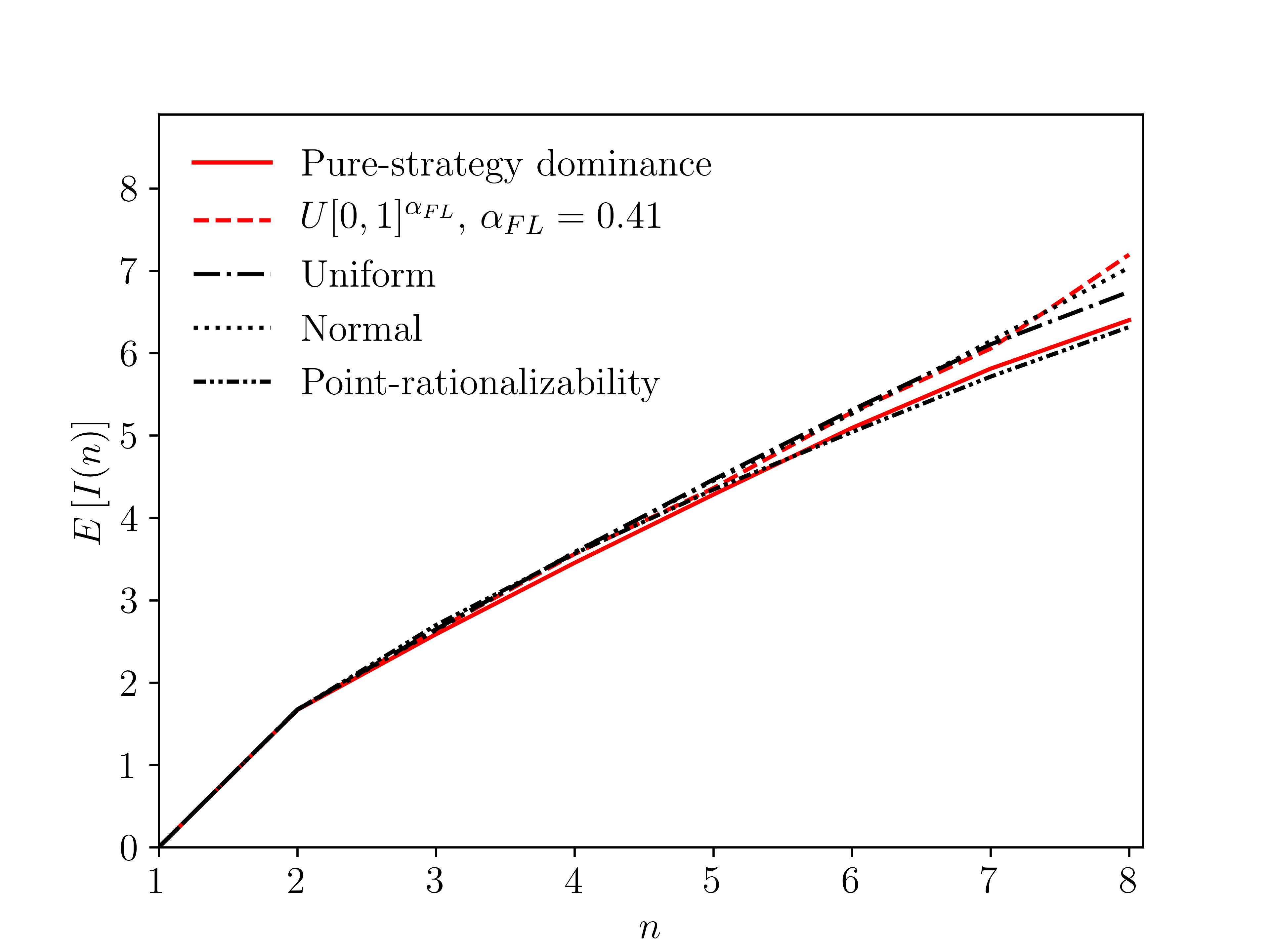}
    \caption{Conditional number of iterations}\label{fig: I_mixed_bal}
    \end{subfigure}
    \par
    \medskip
    \par
    \begin{subfigure}[t]{\textwidth}
    \centering
    \vspace{0pt}
    \includegraphics[width=0.49\linewidth]{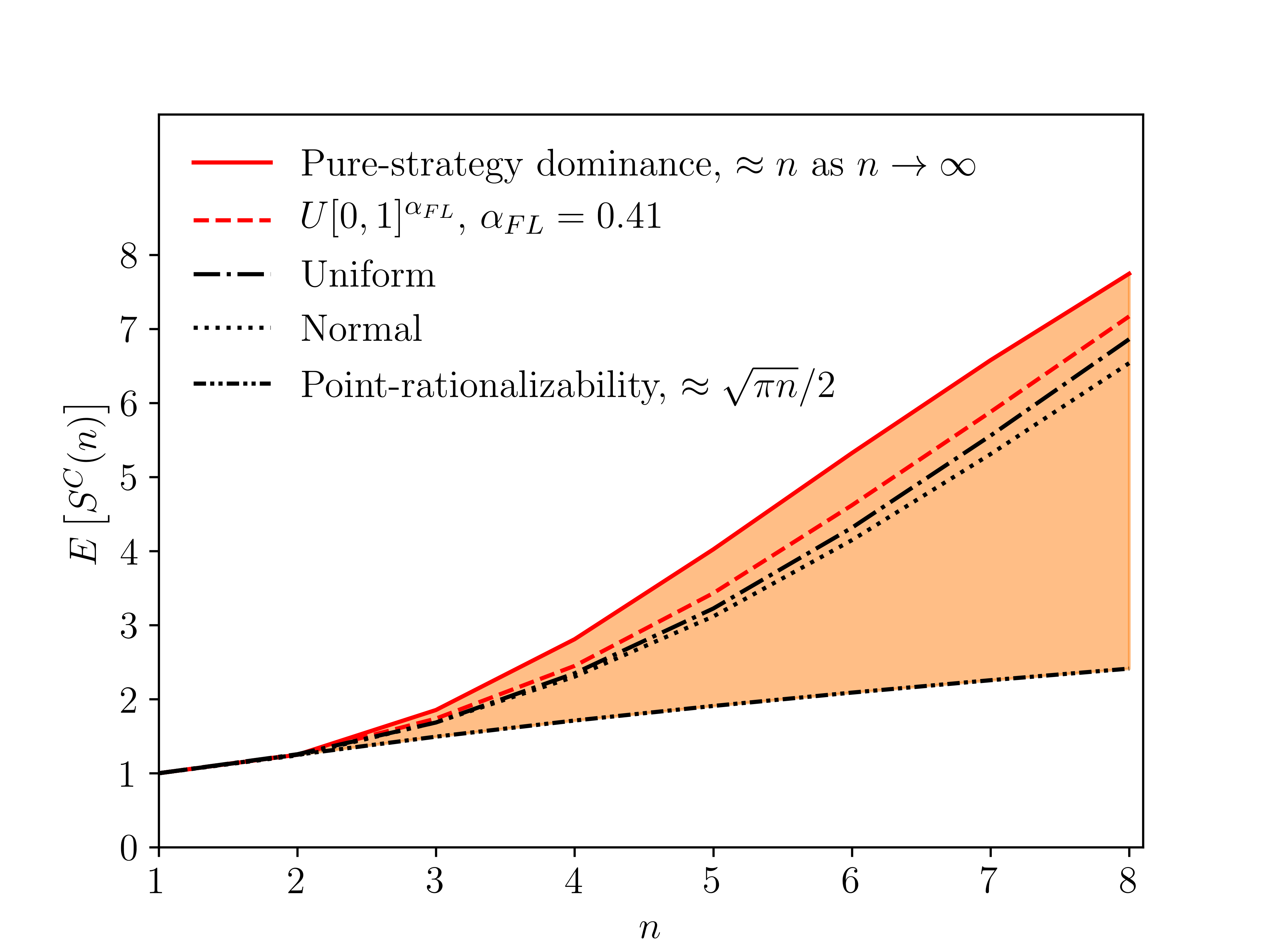}
    \caption{Surviving actions}\label{fig: S_mixed_bal}
    \end{subfigure}
    \captionsetup{justification=centering}
    \caption{Three dimensions of mixed-strategy dominance solvability in balanced $n \times n$ games, where $\alpha_{FL}=0.41$ for \textit{randomly-generated} games is estimated in \cite{fudenberg_predicting_2019}}
    \label{fig: mixed_bal}
    \end{figure}

The figure illustrates that, for these payoff distributions, both the probability of mixed-strategy solvability and the expected number of Column's rationalizable actions are close to those obtained for pure-strategy dominance. This suggests that commonly-used distributions imply sufficient risk aversion to yield similar outcomes from mixed-strategy and pure-strategy dominance. Furthermore, the number of iterations required under the various solvability notions, and accounting for the three payoffs distributions we consider, are nearly identical. In particular, mixed-strategy solvable games appear to require many iterations, reinforcing our results about complexity of solvable games. 


\section{Conclusions}
This paper provides a characterization of several features resulting from iterative elimination of strictly-dominated actions in general random games. We show that \textquotedblleft simple\textquotedblright\ games, ones that are dominance solvable in fairly few steps, are rare. Iterated elimination can help players simplify the game only when players' action sets are sufficiently imbalanced. These insights remain even when restricting attention to various classes of games commonly studied in the literature, or when allowing for domination by mixed strategies. From a technical perspective, we show the usefulness of several new methods from probabilistic combinatorics. 




\section*{Appendix -- Proofs}
\addcontentsline{toc}{section}{Appendix -- Proofs}

\begin{delayedproof}{prop: p_2_n}
We prove the proposition's three statements in turn.
\begin{enumerate}[wide, labelindent=0pt]
\item By Lemma~\ref{lem: u_2_n},  
\begin{align*}
\pi(2,n) &= \Pr(S^R(2,n)=1) = \sum_{k=1}^n \Pr\left(U^C(2,n)=k\right) \cdot \Pr\left(S^R(2,n)=1 \mid
U^C(2,n)=k\right) \\
&= \sum_{k=1}^n \Pr\left(U^C(2,n)=k\right) \cdot \Pr\left(U^R(2,k)=1\right)
= \frac{2}{n!} \cdot \left.\left(\sum_{k=1}^n s(n, k) \cdot
x^k\right)\right|_{x=1/2},
\end{align*}
where the third equality follows from the ordinal randomness assumption. By using
the Pochhammer symbol $x^{(n)} \equiv x(x+1)\ldots(x+n-1)$, we get 
\begin{equation*}
\pi(2,n) = \frac{2}{n!} \cdot \left. x^{(n)} \right|_{x=1/2} = \frac{2}{\Gamma(n+1)} \cdot \frac{\Gamma(n+1/2)}{\Gamma(1/2)},
\end{equation*}
where $\Gamma(\cdot)$ is the gamma function with $\Gamma(1/2)=\sqrt{\pi}$,
the first equality follows from the Proposition 1.3.7 in \cite%
{stanley_permutations_2015}, and the second equality is standard (e.g., 
\citealp{srivastava_generalizations_2013}). By introducing the Wallis ratio 
\begin{equation*}
W(n) \equiv \frac{(2 n-1) ! !}{(2 n) ! !}=\frac{\Gamma\left(n+1/2\right)}{%
\Gamma(1/2) \Gamma(n+1)},
\end{equation*}
we finally obtain $\pi(2,n) = 2W(n)=\dfrac{(2n-1)!!}{2^{n-1}\cdot n!}$.
\item For any $n \geq 1$, by using the identity $\Gamma(x+1)=x\Gamma(x)$,  
\begin{equation*}
\pi(2,n+1) = \frac{n+1/2}{n+1} \cdot \pi(2,n) <\pi(2,n).
\end{equation*}

\item By Stirling's formula applied to the gamma function,  
\begin{align*}
\lim \limits_{n \to \infty}\frac{\Gamma(n+\alpha)}{\Gamma(n) \cdot n^\alpha}%
=1,
\end{align*}
\noindent so that  
\begin{equation*}
\lim \limits_{n \to \infty} n^{1/2} \cdot \pi(2,n) = \frac{2}{\sqrt{\pi}}
\cdot \lim \limits_{n \to \infty} \frac{\Gamma(n+1/2) \cdot n^{1/2}}{%
\Gamma(n+1)} = \frac{2}{\sqrt{\pi}}. \qedhere
\end{equation*}
\end{enumerate}
\end{delayedproof}

\begin{delayedproof}{prop: i_2_n}
We prove the proposition's three statements in order.

\begin{enumerate}[wide, labelindent=0pt]
\item \label{prop: i_2_n_1} By using probabilities derived in Subsection~\ref{subsection: i_2_n},
  
\begin{align*}
\E\left[I(2,n)\right] &= \sum_{i=1}^3 i \cdot \Pr(I(2,n)=i) = 1 \cdot \frac{\sqrt{\pi}}{2^{n}}\cdot \frac{\Gamma(n)}{\Gamma(n+1/2)} + 2
\cdot \frac{n+2^{n-1}-2}{2^n} \cdot \sqrt{\pi} \cdot \frac{\Gamma(n)}{%
\Gamma(n+1/2)} \\
& + 3 \cdot \left(1-\frac{n+2^{n-1}-1}{2^n} \cdot \sqrt{\pi} \cdot \frac{%
\Gamma(n)}{\Gamma(n+1/2)}\right) = 3-\frac{n+2^{n-1}}{2^n} \cdot \sqrt{\pi} \cdot \frac{\Gamma(n)}{%
\Gamma(n+1/2)}.
\end{align*}

\item Note that for any $n \geq 1$,  
\begin{equation*}
A(n+1) \equiv \frac{n+1+2^{n}}{2^{n+1}} \cdot \sqrt{\pi} \cdot \frac{%
\Gamma(n+1)}{\Gamma(n+3/2)} = \frac{n+1+2^n}{2n+2^{n}} \cdot \frac{n}{n+1/2}
\cdot A(n) <A(n),
\end{equation*}
\noindent so that $\E\left[I(2,n)\right]$ is strictly increasing in $n$ by~%
\ref{prop: i_2_n_1}. 

\item From Stirling's formula, we have  $\lim \limits_{n \to \infty} n^{1/2} \cdot\left(3-\E\left[I(2,n)\right]\right) = \dfrac{\sqrt{\pi}}{2}.$\qedhere

\end{enumerate}
\end{delayedproof}

\begin{delayedproof}{prop: s_c_2_n}
The proof of this statement is similar to \cite{hwang_convergence_1998} and
uses the Berry-Esseen theorem to find the convergence rate in the stated central limit result. The difference is that the problem does not belong to the exp-log class immediately. 
For simplicity of notation, we let
\begin{align*}
\mu_n &\equiv \E\left[S^C(2,n)\right]=\ln{n}+\gamma + o(1),\\
\sigma_n &\equiv \sqrt{\Var\left[S^C(2,n)\right]}=\sqrt{\ln{n}}- \left(%
\frac{\pi^2}{12}-\frac{\gamma}{2}\right)\cdot \frac{1}{\sqrt{\ln{n}}}%
+o\left(\frac{1}{\sqrt{\ln{n}}}\right)\text{, and}\\
\varphi_n(t)&=\sum_{j=1}^{n} \Pr\left(S^C(2,n)=j\right) \cdot e^{it(j-\mu_n)/\sigma_n},
\end{align*}
\noindent where $\varphi _n(t)$ denotes the characteristic function of the normed variable $(S^C(2,n)-\mu_n)/{\sigma_n}$. The asymptotic expressions above are derived in the Online Appendix.

\begin{Berry}
Let $F ( x )$ be a non-decreasing function, $G(x)$ a differentiable function
of bounded variation on the real line, $\varphi ( t )$ and $\gamma ( t )$
the corresponding Fourier-Stieltjes transforms: 
\begin{equation*}
\varphi ( t ) = \int _ { - \infty } ^ { \infty } e ^ { i t x } d F ( x ) ,
\quad \gamma ( t ) = \int _ { - \infty } ^ { \infty } e ^ { i t x } d G ( x
).
\end{equation*}
Suppose that $F ( - \infty ) = G ( - \infty )$, $F ( \infty ) = G ( \infty )$%
, $T$ is an arbitrary positive number, and $\left| G ^ { \prime } ( x )
\right| \leq A$. Then for every $b > 1 / ( 2 \pi )$ we have 
\begin{equation*}
\sup _ { - \infty < x < \infty } | F ( x ) - G ( x ) | \leq b \int _ { - T }
^ { T } \left| \frac { \varphi ( t ) - \gamma ( t ) } { t } \right| d t + r
( b ) \frac { A } { T },
\end{equation*}
where $r(b)$ is a positive constant depending only on $b$.
\end{Berry}

We proceed in two steps. In step 1, we reformulate the problem by using the
Berry-Esseen inequality. In step 2, we calculate the characteristic function
and establish the result.\bigskip

\noindent {\textit{Step 1. Reformulated problem}}

Take $G ( x ) = \Phi ( x )$ (so that $A = 1 / \sqrt { 2 \pi }$) and $T
= T _ { n } = c\sigma _ { n }$, where $c>0$ is a sufficiently small
constant. By the Berry-Esseen inequality, it will be sufficient to prove
that 
\begin{align*}
J _ { n } = \int _ { - T _ { n } } ^ { T _ { n } } \left| \frac { \varphi _
{ n } ( t ) - e ^ { - \frac { 1 } { 2 } t ^ { 2 } } } { t } \right| d t = 
\mathcal{O} \left( \frac{1}{\sqrt{\ln{n}}} \right).
\end{align*}

\noindent {\textit{Step 2. Characteristic function}}
\begin{align*}
\varphi_n(t)=\sum_{j=1}^{n} \Pr\left(S^C(2,n)=j\right) \cdot
e^{it(j-\mu_n)/\sigma_n} =A_n(t)+B_n(t),
\end{align*}
where, following algebraic simplification, we get
\begin{align*}
A_n(t) &\equiv e^{-it\mu_n/\sigma_n} \cdot\frac{1}{\Gamma
\left(e^{it/\sigma_n} \right)}\cdot \frac{\Gamma\left(n + e^{it/\sigma_n}
\right)}{\Gamma(n+1)}=e^{-\frac{t^2}{2} + \mathcal{O}%
\left(\frac{|t|+|t|^3}{\sqrt{\ln{n}}}\right)}\text{, and} \\
B_n(t) &\equiv 2 \cdot e^{-it\mu_n/\sigma_n} \cdot\left(\frac{e^{it/\sigma_n} 
}{\Gamma(1/2)}\cdot \frac{\Gamma(n+1/2)}{\Gamma(n+1)} - \frac{1}{\Gamma
\left(e^{it/\sigma_n}/2 \right)}\cdot \frac{\Gamma\left(n +
e^{it/\sigma_n}/2 \right)}{\Gamma(n+1)} \right)\\
&=\frac{2}{\sqrt{\pi}} \cdot e^{-\left(it \cdot\sqrt{\ln{n}} + 
\mathcal{O}\left(\frac{|t|}{\sqrt{\ln{n}}}\right)\right)} \cdot e^{\mathcal{%
O}\left(\frac{|t|}{\sqrt{\ln{n}}}\right)} \cdot e^{-\frac{1}{2} \cdot \ln{n%
} + \mathcal{O}\left(\frac{1}{n}\right)} \\
&- 2 \cdot e^{-\left(it \cdot\sqrt{\ln{n}} + \mathcal{O}\left(\frac{|t|}{%
\sqrt{\ln{n}}}\right)\right)} \cdot e^{-\mathcal{O}\left(1\right)} \cdot
e^{- \frac{1}{2}\cdot \ln{n} + \frac{1}{2}it\cdot \sqrt{\ln{n}} + \mathcal{%
O}\left(t^2\right)}.
\end{align*}
The Online Appendix provides omitted details behind the above derivations.

Note that $B_n(0)=0$ and $B_n(s)=\mathcal{O}\left(\frac{e^{\tau \cdot \sqrt{\ln{n}}}}{%
n^{1/2}}\right)$ uniformly for $|s| \leq \tau$, $s \in \mathcal{C}$, for some
fixed $\tau>0$. By denoting $\kappa_n \equiv \frac{n^{1/2}}{e^{\tau \cdot \sqrt{\ln{n}}}}$ for convenience, we can rewrite $B_n(s)=\mathcal{O}\left(\frac{1}{\kappa_n}\right)$ for $|s| \leq \tau$. Furthermore, by taking a small ball around the origin we easily obtain $B_n(s)=\mathcal{O}\left(\frac{|s|}{\kappa_n}\right)$ for $|s| \leq c<\tau$, where sufficiently small $c>0$ can be taken to be less than $\tau$. Consequently,
\begin{align*}
\varphi _ { n } ( t )=A_n(t) + B_n(t) = e^{-\frac{t^2}{2} + \mathcal{O}\left(\frac{|t|+|t|^3%
}{\sqrt{\ln{n}}}\right)} + \mathcal{O}\left(\frac{|t|}{\kappa_n \cdot \sqrt{\ln n}}\right),
\end{align*}
for $| t | \leq T _ { n }=c\sigma_n$.


Based on the obtained approximation, we can follow the proof of Theorem 1 in \cite%
{hwang_convergence_1998}. That is, using the inequality $\left| e ^ { w } -
1 \right| \leq | w | e ^ { | w | }$ for all complex $w$, we obtain 
\begin{align*}
\left| \frac { \varphi _ { n } ( t ) - e ^ { - \frac { 1 } { 2 } t ^ { 2 } } 
} { t } \right| &= \mathcal{O} \left( \left( \frac { 1 + t ^ { 2 } } { \sqrt{%
\ln{n}} } \right) \exp \left( - \frac { t ^ { 2 } } { 2 } + \mathcal{O}
\left( \frac { | t | + | t | ^ { 3 } } { \sqrt{\ln{n}} } \right) \right) + 
\frac { 1 } { {\kappa_n \cdot \sqrt{\ln{n}}} } \right) \\
&= \mathcal{O} \left( \left( \frac { 1 + t ^ { 2 } } { \sqrt{\ln{n}} }
\right) e ^ { - \frac { 1 } { 4 } t ^ { 2 } } + \frac { 1 } { {\kappa_n
\cdot \sqrt{\ln{n}}} } \right) \quad \left( | t | \leq T _ { n } \right),
\end{align*}
for sufficiently small $0<c<\tau$. 

Thus, 
\begin{align*}
J_n &= \int _ { - T _ { n } } ^ { T _ { n } } \left| \frac { \varphi _ { n }
( t ) - e ^ { - \frac { 1 } { 2 } t ^ { 2 } } } { t } \right| d t = \mathcal{O} \left( \frac { 1 } { \sqrt{\ln{n}} } \int _ { - T _ { n } }
^ { T _ { n } } \left( 1 + t ^ { 2 } \right) e ^ { - \frac { 1 } { 4 } t ^ {
2 } } d t + \frac { 1 } { {\kappa_n} } \right) \\
&= \mathcal{O} \left( \frac { 1 } { \sqrt{\ln{n}} } + \frac { 1 } { {%
\kappa_n} } \right) = \mathcal{O} \left( \frac { 1 } { \sqrt{\ln{n}} } \right),
\end{align*}
\noindent because $\lim\limits_{n \to \infty} \dfrac{\sqrt{\ln{n}}}{\kappa_n} = \lim\limits_{n \to \infty} 
\dfrac{\sqrt{\ln{n}} \cdot e^{\tau \cdot \sqrt{\ln{n}}}}{n^{1/2}} 
= \lim \limits_{n \to \infty} \dfrac{n \cdot e^{\tau \cdot n}}{e^{n^2/2}} =0$. 
\qedhere
\end{delayedproof}

\begin{delayedproof}{lem: u_m_n} 
We prove the lemma's three statements in turn.
\begin{enumerate}[wide, labelindent=0pt]
\item For any $k\geq1$, if there are at most $k$ undominated Column's
actions for some $g(m+1,n)$, then there are at most $k$
undominated Column's actions for the corresponding $g(m, n)$
constructed from $g(m+1,n)$ by removing the $(m+1)$-th action of Row.
Therefore, 
\begin{equation*}
\Pr(U^C(m+1,n) \leq k) \leq \Pr(U^C(m,n) \leq k) \text{ for any }%
k=1,2,\ldots,n.
\end{equation*}
To conclude the proof, $\Pr(U^C(m+1,n) \leq 1) = n^{-m} <
n^{-(m-1)}=\Pr(U^C(m,n) \leq 1)$.
    \item For any $m \geq 2$, consider Column's payoff matrix $C$, 
where rows $\{c_{1\cdot}, c_{2\cdot}, \ldots, c_{m\cdot}\}$ are \textit{i.i.d.} uniform on $S_n$. Let $p^C(m,n)$ be the probability that any particular column is undominated. 

Without loss of generality, focus on the first column $c_{\cdot1}$. There are $n$ possible values for $c_{m1}$. Suppose that $c_{m1}=k$ for some $k \in \lbrack n \rbrack$. This happens with probability $\frac{1}{n}$. Then, the first action is undominated if and only if it is undominated in the $(m-1)\times(n-k+1)$ game formed by removing all columns $j$ with $c_{mj}<k$ and the last row. It happens with probability $p^C(m-1,n-k+1)$. By summing over all possible $k$, $k \in \lbrack n \rbrack$, we easily get
\begin{equation*}
p^C(m,n)=\frac{\sum_{k=1}^n p^C(m-1,n-k+1)}{n}=\frac{\sum_{k=1}^n p^C(m-1,k)}{n}.
\end{equation*}
To conclude, by linearity of the expectation, $\E\left[U^C(m,n)\right]= n \cdot p^C(m,n)$.

Component-wise monotonicity follows immediately from the first part of this lemma and the recurrence relation itself.
\item We show, using induction on $m \geq 2$, that for any $n \geq 1$,
\begin{equation*}
\E\left[U^C(m,n)\right] \leq \frac{(\ln{n})^{m-1}}{(m-1)!}+\frac{(\ln{n})^{m-2}}{(m-2)!}+\ldots+\frac{(\ln{n})^{2}}{2}+\ln{n}+1.
\end{equation*}
For $m=2$, $\E\left[U^C(2,n)\right]=H_n=1+\frac{1}{2}+\ldots+\frac{1}{n} \leq \ln{n}+1$ for any $n \geq 1$. Assume that this statement holds for $m \geq 2$. We prove it for $m+1$.

For any $x>0$, define 
\begin{equation*}
 f(x) \equiv \frac{(\ln{x})^{m-1}}{(m-1)!x}+\frac{(\ln{x})^{m-2}}{(m-2)!x}+\ldots+\frac{(\ln{n})^{2}}{2x}+\frac{\ln{x}}{x}+\frac{1}{x}.
\end{equation*}
The function $f(x)$ is strictly decreasing in $x>1$. 
By the recurrence relation, 
\begin{multline*}
\E\left[U^C(m+1,n)\right] = \frac{\E\left[U^C(m,1)\right]}{1}+\frac{\E\left[U^C(m,2)\right]}{2}+\ldots+\frac{\E\left[U^C(m,n)\right]}{n}\\
\leq f(1)+f(2)+\ldots+f(n)\leq f(1)+\int_1^n f(x)dx
=1+\left(\frac{(\ln{x})^{m}}{m!}+\frac{(\ln{x})^{m-1}}{(m-1)!}+\ldots+\ln{x}\right) \bigg|_{1}^n\\
=\frac{(\ln{n})^{m}}{m!}+\frac{(\ln{n})^{m-1}}{(m-1)!}+\ldots+\frac{(\ln{n})^{2}}{2}+\ln{n}+1,
\end{multline*}
as desired, where the the first inequality follows from the induction hypothesis and the second holds since $f(x)$ is strictly decreasing in $x>1$.

We now use induction on $m \geq 2$ to show that, for any $n \geq 1$,
\begin{equation*}
\E\left[U^C(m,n)\right] \geq \frac{(\ln{n})^{m-1}}{(m-1)!}.
\end{equation*}
For any fixed $m\geq 2$, $p^C(m,n) = \E\left[U^C(m,n)\right]/n$---the probability that the first Column's action is undominated---is decreasing in $n \geq 1$ by its definition.

For $m=2$, $\E\left[U^C(2,n)\right]=H_n=1+\frac{1}{2}+\ldots+\frac{1}{n} > \ln{n}$ for any $n \geq 1$. Assume that the desired statement holds for $m \geq 2$. We prove it for $m+1$.

For any $x>0$, define 
\begin{equation*}
 f(x) \equiv \frac{(\ln{x})^{m-1}}{(m-1)!x},\qquad \text{so that} \qquad
 f'(x)=-\frac{(\ln{x})^{m-1}}{(m-1)!x^2}+\frac{(m-1)(\ln{x})^{m-2}}{(m-1)!x^2}
\end{equation*}
is negative for $x>e^{m-1}$ and positive for $x<e^{m-1}$. Therefore, $f(x)$ has a unique (global) maximum at $x_{\max}=e^{m-1}$, so that
\begin{equation*}
    \max_{x>0} f(x) = f(x_{\max})=f(e^{m-1})=\frac{(m-1)^{m-1}}{(m-1)!e^{m-1}} \leq \frac{1}{e} \leq \frac{1}{2},
\end{equation*}
where the first inequality holds since, for any $m \geq 2$, by denoting $g(m) \equiv \frac{(m-1)^{m-1}}{(m-1)!e^{m-1}}$, we have $g(m+1)=\left(1+\frac{1}{m-1}\right)^{m-1} \cdot \frac{g(m)}{e}\leq g(m) \leq g(2)=\frac{1}{e}$. 

By our observations above, for any $i \leq \floor{e^{m-1}}-1$, we have
\begin{equation*}
\frac{\E\left[U^C(m,i)\right]}{i} = p^C(m,i) \geq p\left(m, \floor{e^{m-1}}\right) \geq f\left(\floor{e^{m-1}}\right) \geq f(x) \text{ for any }x\leq \floor{e^{m-1}},
\end{equation*}
\noindent where the first inequality follows since $p^C(m,n)$ decreases in $n$, the second follows from the induction hypothesis, and the third follows because $f(x)$ is increasing for $x<e^{m-1}$. Therefore, for any $i \leq \floor{e^{m-1}}-1$,
\begin{equation*}
\frac{\E\left[U^C(m,i)\right]}{i} \geq \int_{i}^{i+1}f(x)dx.
\end{equation*}
In addition, for $i=1$, we can, in fact, show that
\begin{equation*}
    \frac{\E\left[U^C(m,1)\right]}{1}-\int_{1}^{2}f(x)dx = 1 - \int_{1}^{2}f(x)dx \geq 1 -  \max_{x>0} f(x) \geq \frac{1}{2}.
\end{equation*}

Similarly, for any $i \geq \ceil{e^{m-1}}$, we have
\begin{equation*}
\frac{\E\left[U^C(m,i)\right]}{i} = p^C(m,i) \geq f\left(i\right) \geq f(x) \text{ for any }x\geq i,
\end{equation*}
where the first inequality follows from the induction hypothesis and the second holds because $f(x)$ is decreasing for $x>e^{m-1}$. Therefore, for any $i \geq \ceil{e^{m-1}}$,
\begin{equation*}
\frac{\E\left[U^C(m,i)\right]}{i} \geq \int_{i}^{i+1}f(x)dx.
\end{equation*}
Finally, for $i=\floor{e^{m-1}}$, as long as $m \geq 2$,
\begin{equation*}
\int_{\floor{e^{m-1}}}^{\ceil{e^{m-1}}}f(x)dx \leq \max_{x>0} f(x) \leq \frac{1}{2}.
\end{equation*}
To sum up,
\begin{align*}
    \frac{\E\left[U^C(m,1)\right]}{1} &\geq \int_{1}^{2}f(x)dx + \frac{1}{2} \geq \int_{1}^{2}f(x)dx +  \int_{\floor{e^{m-1}}}^{\ceil{e^{m-1}}}f(x)dx,\\
    \frac{\E\left[U^C(m,i)\right]}{i} &\geq \int_{i}^{i+1}f(x)dx \quad \text{for} \quad i \neq  \floor{e^{m-1}}.
\end{align*}
Therefore, by the recurrence relation and the set of inequalities above, 
\begin{align*}
\E\left[U^C(m+1,n)\right] &= \frac{\E\left[U^C(m,1)\right]}{1}+\frac{\E\left[U^C(m,2)\right]}{2}+\ldots+\frac{\E\left[U^C(m,n)\right]}{n}\\
&\geq \int_{1}^{n+1} f(x)dx > \int_{1}^{n} f(x)dx = \left(\frac{(\ln{x})^{m}}{m!}\right) \bigg|_{1}^n=\frac{(\ln{n})^{m}}{m!}.\qedhere
\end{align*}
\end{enumerate}
\end{delayedproof}

\setcounter{prop}{3}
\begin{spprop}
\label{prop:s_m_n_star}
Consider a random game $G(m,n)$. Then, for any $m \leq n$, 
\begin{equation*}
\Pr\left(S^R(m,n) < m \right) \leq m(m-1)\cdot\left(\dfrac{m}{n}\right)^{\frac{m-1}{4}}.
\end{equation*}
For any fixed $m$, an asymptotic bound can be improved to $\Pr\left(S^R(m,n) < m \right)=\mathcal{O}\left(n^{-\frac{m-1}{2}}\right)$.
\end{spprop}
\begin{proof}
Let $p(m,n)$ denote the probability that the second Row's action strictly dominates her first after the first round of Column's elimination. By symmetry and Boole's inequality,
\begin{equation*}
\Pr\left(S^R(m,n) < m \right) \leq m(m-1)\cdot p(m,n).
\end{equation*}
 
Without loss of generality, set $c_{1\cdot}=(n,n-1, \ldots,1)$. Define events $E_j$, $j \in \lbrack n \rbrack$, as follows:
\begin{align*}
    E_j &\equiv \mathcal{C}_j \bigcap \mathcal{R}_j, \quad\text{where}\\
\mathcal{C}_j &\equiv \bigcup_{i \geq 2}\left\{C: c_{ij}>\max(c_{i,j-1}, c_{i,j-2},\ldots, c_{i,1})\right\} \quad \text{and} \quad
\mathcal{R}_j \equiv \{R: r_{1j}>r_{2j}\}.
\end{align*}
Note that if $E_j$ happens for some $j \in \lbrack n \rbrack$, then the second Row's action cannot strictly dominate her first one after the first round of Column's elimination. Indeed, if both $\mathcal{C}_j$ and $\mathcal{R}_j$ occur, then the column $j$ stays and $r_{1j}>r_{2j}$. 

Assume that events $\{E_j\}_{j\in \lbrack n \rbrack}$ are mutually independent. Then,
\begin{align*}
p(m,n) &\leq \Pr\left(\bigcap_{j \in \lbrack n \rbrack} \overline{E}_j\right) = \prod_{j \in \lbrack n \rbrack} \left(1- \Pr\left(E_j\right)\right), \quad\text{where}\\
    \Pr\left(E_j\right)&=\frac{1}{2} \cdot \frac{m-1}{j}+\frac{1}{2}\cdot \sum_{k=2}^{m-1} (-1)^{k-1} \binom{m-1}{k} \cdot \left(\frac{1}{j}\right)^k.
\end{align*}
By using the inequality $1-x \leq e^{-x}$ that holds for any $x \geq 0$, we get
\begin{equation*}
 \prod_{j \in \lbrack n \rbrack} \left(1- \Pr\left(E_j\right)\right) \leq e^{-\sum_{j \in \lbrack n \rbrack}\Pr\left(E_j\right)}.  
\end{equation*}

Next, we verify that events $\{E_j\}_{j\in \lbrack n \rbrack}$ are mutually independent. Since rows and columns are mutually independent, it suffices to prove that $\{\mathcal{C}_j\}_{j\in \lbrack n \rbrack}$ are mutually independent. 

Note that for any matrix $C$ with $c_{1\cdot}=(n,n-1, \ldots,1)$, we can map it to the matrix $\{m_{ij}\}_{i \in \lbrack m-1 \rbrack,\, j \in \lbrack n \rbrack}$ defined as $m_{ij}=\left|\{k < j: c_{i+1,k}>{c_{i+1,j}}\}\right| \in \{0,1, \ldots, j-1\}$. This mapping is a bijection. Furthermore, $\mathcal{C}_j$ occurs if and only if $m_{ij}=0$ for some $i \in \lbrack m-1 \rbrack$ (corresponding events are mutually independent).

We can now show the statement, considering two possibilities: 
\begin{enumerate}[wide, labelindent=0pt]
    \item If $m=\mathcal{O}\left(1\right)$, then
    \begingroup
    \allowdisplaybreaks
    \begin{align*}
\Pr\left(S^R(m,n) < m \right) &\leq m(m-1)\cdot p(m,n) \leq m(m-1) \cdot e^{-\sum_{j \in \lbrack n \rbrack}\Pr\left(E_j\right)}\\
&=\mathcal{O}\left(1\right)\cdot e^{-\frac{m-1}{2}\ln{n}+\mathcal{O}\left(1\right)}=\mathcal{O}\left(1\right) \cdot n^{-\frac{m-1}{2}}=\mathcal{O}\left(n^{-\frac{m-1}{2}}\right).
\end{align*}
\endgroup
    \item The main statement is trivial for $m=1$. Consider $m \geq 2$. By applying Bonferroni's inequality up to $k=2$, we have
\begin{equation*}
    \Pr\left(E_j\right) \geq \frac{m-1}{2} \cdot \frac{1}{j} - \frac{(m-1)(m-2)}{4} \cdot \frac{1}{j^2}=\frac{m-1}{2}\left(\frac{1}{j}-\frac{m-2}{2}\cdot \frac{1}{j^2}\right).
\end{equation*}
As $m \leq n$, because $\frac{m-2}{2j} < \frac{1}{2}$ for any $j\geq m-1$,

\begin{multline*}
    \sum_{j \in \lbrack n \rbrack}\Pr\left(E_j\right)\geq \sum_{j \geq m-1}\Pr\left(E_j\right)=\sum_{j \geq m-1}\frac{m-1}{2}\left(\frac{1}{j}-\frac{m-2}{2j}\cdot \frac{1}{j}\right)\\
    >\sum_{j \geq m-1}\frac{m-1}{2}\left(\frac{1}{j}-\frac{1}{2}\cdot \frac{1}{j}\right)=\frac{m-1}{4} \cdot \sum_{j \geq m-1} \frac{1}{j}
    >\frac{m-1}{4} \cdot \int_{m}^n \frac{dx}{x}=\frac{m-1}{4} \cdot \ln{\frac{n}{m}},
\end{multline*}
\noindent so that
\begin{equation*}
\Pr\left(S^R(m,n) < m \right) \leq m(m-1)\cdot e^{-\frac{m-1}{4} \cdot \ln{\frac{n}{m}}}=m(m-1)\cdot\left(\frac{m}{n}\right)^{\frac{m-1}{4}}.\qedhere
\end{equation*}
\end{enumerate}
\end{proof}

\begin{delayedproof}{lem: subgame}
If either $m=1$ or $n=1$, the proof is trivial. Consider, then, $m, n \geq 2$.

We show that there exists an $(m-1) \times n$ subgame that is strict-dominance solvable. Indeed, if there is a strictly dominated action for Row in the original game, then the $(m-1) \times n$ subgame formed by the exclusion of this action is strict-dominance solvable. Otherwise, there is a strictly dominated action for Column such that, in the induced game after the first iteration, Row has a strictly dominated action.  The $(m-1) \times n$ subgame, formed by the exclusion of this action from the original game, is strict-dominance solvable. By a symmetric argument, there exists an $m \times (n-1)$ subgame that is strict-dominance solvable.

We can repeat these steps to prove the desired result by induction.\qedhere
\end{delayedproof}

\begin{delayedproof}{thm:p_m_n}
We first show that $\pi(n,n)\leq n^{-\left(\frac{1}{3}-o(1)\right)n}$. The idea behind the proof is to estimate the probability to eliminate at least $\frac{n}{3}$ rows (actions of Row) or columns (actions of Column). This probability will provide the desired upper bound for the probability of strict-dominance solvability.

Start the standard iterative elimination procedure and stop exactly when at least $\frac{n}{3}$ rows  or at least $\frac{n}{3}$ columns are deleted. To simplify the presentation, we omit all floor and ceiling signs whenever these are not crucial. Without loss of generality, suppose that we deleted $\frac{n}{3}$ rows and at most $\frac{n}{3}$ columns.

Let $X$ be the set of columns that are not yet eliminated. Similarly, $Y$ is defined as the set of rows that are not yet deleted. Their complements $X' \equiv [n]	\setminus X$ and $Y' \equiv [n]	\setminus Y$ correspond to eliminated columns and rows, respectively. By the previous paragraph, $|X| \geq \frac{2n}{3}$ and $|Y|=\frac{2n}{3}$. Also, for any row $r_{i}$ eliminated, $i \in Y'
$, there must exist a row $r_{j}$ not eliminated yet, $j \in Y$, so that $r_{jx}>r_{ix}$ for any $x \in X$, namely $r_{j}$ strictly dominates $r_{i}$ when restricted to columns $X$.

For any row $r_{i}$ eliminated, $i \in Y'$, choose some row $r_{j(i)}$ not eliminated yet, $j(i) \in Y$, so that $r_{j(i)x}>r_{ix}$ for any $x \in X$, and draw a directed edge from $j(i)$ to $i$. We get a collection of $r$ stars of sizes $k_1,k_2,\ldots,k_r$ with centers in $Y$ (not eliminated rows) and leaves $Y'$ (eliminated rows), so that $k_1+k_2+\ldots+k_r=|Y'| = \frac{n}{3}$. 

First, the total number of ways to choose such $X$, $Y$, and stars, is bounded above by

\begin{equation*}
    \binom{n}{|X|} \cdot \binom{n}{|Y|}   \cdot |Y|^{|Y'|} \leq 2^n \cdot 2^n \cdot |Y|^{n-|Y|} \leq 4^n \cdot \left(\frac{2n}{3}\right)^{\frac{n}{3}}.
\end{equation*}

Second, for any such fixed $X$, $Y$, and $r$ stars of sizes $k_1+k_2+\ldots+k_r=|Y'|=\frac{n}{3}$, the probability that for each star, its center dominates all corresponding leaves when restricted to $X$, is exactly
\begin{equation*}
    \left(\frac{1}{k_1+1}\cdot\frac{1}{k_2+1}\cdot\ldots\cdot \frac{1}{k_r+1}\right)^{|X|} \leq \left(\frac{1}{|Y'|+1}\right)^{|X|} \leq \left(\frac{1}{|Y'|}\right)^{|X|} \leq \left(\frac{1}{n/3}\right)^{\frac{2n}{3}}.
\end{equation*}

Based on two previous inequalities, the probability to eliminate at least $\frac{n}{3}$ rows or columns is bounded above by
\begin{equation*}
4^n \cdot \left(\frac{2n}{3}\right)^{\frac{n}{3}} \cdot  \left(\frac{1}{n/3}\right)^{\frac{2n}{3}} = n^{-\left(\frac{1}{3}-o(1)\right)n}\text{, as desired.}
\end{equation*}

In order to prove the main statement of this theorem, we consider two relevant cases. If $m \geq n^{0.9}$, then by Lemma~\ref{lem: subgame} and the inequality for balanced games proved above,
\begin{equation*}
\pi(m,n) \leq \binom{n}{m} \cdot \pi(m,m) \leq \binom{n}{m} \cdot \frac{1}{m^{0.3m}} \leq \binom{n}{m} \cdot \frac{1}{n^{0.27m}}.
\end{equation*}
By using the standard upper bound for the binomial coefficient, we get
\begin{equation*}
  \binom{n}{m} \cdot \frac{1}{n^{0.27m}} \leq \left(\frac{en}{m}\right)^m \cdot \frac{1}{n^{0.27m}} = e^m \cdot \frac{1}{n^{0.17m}}.
\end{equation*}
Otherwise, if $m \leq n^{0.9}$, then by Proposition~\ref{prop:s_m_n},
\begin{equation*}
\pi(m,n) \leq m^2 \cdot \left(\frac{1}{n^{0.1}}\right)^{\frac{m-1}{4}} \leq m^2 \cdot \frac{1}{n^{0.025m}}.
\end{equation*}
To sum up, by taking $C_2=0.01$ and sufficiently large $C_1>0$, for any $m \leq n$, we obtain the desired bound $\pi(m,n) \leq C_1 n^{-C_2m}$.\qedhere
\end{delayedproof}  

\addcontentsline{toc}{section}{References}
\bibliographystyle{ecta}
\bibliography{references}

\end{document}


\appendix
\title{\bfseries{\large Online Appendix for ``Dominance Solvability in Random Games''}}
\author{by Noga Alon, Kirill Rudov, and Leeat Yariv}
\begin{titlingpage}
\begin{changemargin}{0.5in}{0.5in} 
 \maketitle
\begin{abstract}
\normalsize This Online Appendix contains the following four items. First, in Appendix~\ref{appendix_A} we provide proofs omitted from the main text. Second, in Appendix~\ref{appendix_B} we provide more details on enumerative issues arising in our general analysis and particularly on the connection between our problem and permutation patterns. Third, in Appendix~\ref{appendix_C} we discuss how our results extend to multiple players. Fourth, in Appendix~\ref{appendix_D} we include additional figures referred to in the main text.
\end{abstract}
\end{changemargin} 
\end{titlingpage}

\newgeometry{bottom=1in, top=1.5in, right=1in,left=1in}
\section{Omitted Proofs}
\label{appendix_A}

\setcounter{lemma}{3}
\begin{lemma}
\label{lem: i_2_n} Consider a random game $G(2,n)$. Then,

\begin{enumerate}
\item $\Pr(I(2,n)=1)$ is strictly decreasing in $n$ and  $\lim \limits_{n \to \infty} {2^n}\cdot{n^{1/2}}\cdot\Pr(I(2,n)=1)=%
\sqrt{\pi}$;

\item $\Pr(I(2,1)=2)=0$, $\Pr(I(2,2)=2)=\Pr(I(2,3)=2)=2/3$, $%
\Pr(I(2,n)=2)$ is strictly decreasing in $n$ for $n \geq 3$, and $\lim \limits_{n \to \infty} {n^{1/2}}\cdot\Pr(I(2,n)=2)=\dfrac{\sqrt{%
\pi}}{2}$;

\item $\Pr(I(2,1)=3)=\Pr(I(2,2)=3)=0$, $\Pr(I(2,n)=3)$ is strictly
increasing in $n$ for $n \geq 2$, and $\lim \limits_{n \to \infty} {n^{1/2}}\cdot(1-\Pr(I(2,n)=3))=\dfrac{%
\sqrt{\pi}}{2}$. 
\end{enumerate}
\end{lemma}

\begin{proof} We prove the three statements in turn.
\begin{enumerate}[wide, labelindent=0pt]
    \item \label{lem: i_2_n_1} Consider $\Pr(I(2,n)=1)=\dfrac{\sqrt{\pi}}{2^{n}}\cdot \dfrac{\Gamma(n)}{\Gamma(n+1/2)}$, derived in Subsection 3.3. For any $n\geq1$,
     \begin{equation*}
    \Pr(I(2,n+1)=1)= \frac{\sqrt{\pi}}{2^{n+1}}\cdot \frac{\Gamma(n+1)}{\Gamma(n+3/2)}= \frac{n}{2n+1} \cdot \Pr(I(2,n)=1)  < \Pr(I(2,n)=1).
    \end{equation*}
    By Stirling's formula applied to the gamma function,
    \begin{equation*}
     \lim \limits_{n \to \infty} {2^n}\cdot{n^{1/2}}\cdot\Pr(I(2,n)=1)=\sqrt{\pi} \lim \limits_{n \to \infty} \frac{\Gamma(n) \cdot n^{1/2}}{\Gamma(n+1)} =\sqrt{\pi}. 
    \end{equation*}
      \item \label{lem: i_2_n_2} Similarly, consider $\Pr(I(2,n)=2)=\dfrac{n+2^{n-1}-2}{2^n} \cdot \sqrt{\pi} \cdot \dfrac{\Gamma(n)}{\Gamma(n+1/2)}$, derived in Subsection 3.3. It is straightforward to check that $\Pr(I(2,2)=2)=\Pr(I(2,3)=2)=2/3$ and $\Pr(I(2,1)=2)=0$. For any $n\geq3$,
      \begin{equation*}
    \Pr(I(2,n+1)=2) = \frac{n+2^{n}-1}{2n+2^{n}-4}  \cdot \frac{n}{n+1/2} \cdot \Pr(I(2,n)=2) < \Pr(I(2,n)=2).
    \end{equation*}
    By Stirling's formula applied to the gamma function,
    \begin{equation*}
    \lim \limits_{n \to \infty} {n^{1/2}}\cdot\Pr(I(2,n)=2)=\frac{\sqrt{\pi}}{2} \cdot \lim \limits_{n \to \infty} \frac{n+2^{n-1}-2}{2^{n-1}} \cdot \lim \limits_{n \to \infty} \frac{\Gamma(n) \cdot n^{1/2}}{\Gamma(n+1)} =\frac{\sqrt{\pi}}{2}.
    \end{equation*}
     \item Finally, consider $\Pr(I(2,n)=3)=1-\dfrac{n+2^{n-1}-1}{2^n} \cdot \sqrt{\pi} \cdot \dfrac{\Gamma(n)}{\Gamma(n+1/2)}$, also derived in Subsection 3.3. It is straightforward to check that $\Pr(I(2,1)=3)=\Pr(I(2,2)=3)=0$. By~\ref{lem: i_2_n_1} and~\ref{lem: i_2_n_2}, $\Pr(I(2,n)=3)$ is strictly increasing in $n$ for $n \geq 2$. By Stirling's formula,
     \begin{equation*}
    \lim \limits_{n \to \infty} {n^{1/2}}\cdot\left(1-\Pr(I(2,n)=3)\right)=\frac{\sqrt{\pi}}{2} \cdot \lim \limits_{n \to \infty} \frac{n+2^{n-1}-1}{2^{n-1}} \cdot \lim \limits_{n \to \infty} \frac{\Gamma(n) \cdot n^{1/2}}{\Gamma(n+1)} =\frac{\sqrt{\pi}}{2}.\qedhere
    \end{equation*}
\end{enumerate}
\end{proof}

\begin{lemma}
\label{lem: s_c_2_n_1} Consider a random game $G(2,n)$. Then,

\begin{enumerate}
\item $\lim \limits_{n \to \infty} n^{1/2} \cdot \Pr\left(S^C(2,n)=1\right)
= \dfrac{2}{\sqrt{\pi}}$; 

\item for any fixed $k \geq 2$, $\lim \limits_{n \to \infty} \dfrac{n}{(\ln{%
n})^{k-1}} \cdot \Pr\left(S^C(2,n)=k\right) = \dfrac{1}{(k-1)!} \cdot
\left(1-\dfrac{1}{2^{k-1}}\right) $; 

\item for $k(n) \sim \ln{n}$, $\lim \limits_{n \to \infty} (\ln{n})^{1/2}
\cdot \Pr\left(S^C(2,n)=k(n)\right)=\dfrac{1}{\sqrt{2\pi}}$.
\end{enumerate}
\end{lemma}

\begin{proof} For these results, we use two theorems regarding the unsigned Stirling numbers of the first kind.

\begin{Hwang}
For any $\eta>0$, the unsigned Stirling numbers of the first kind $s(n,k)$ satisfy asymptotically 
\begin{equation*}
 \frac{s(n,k)}{n!}=\frac{1}{n} \cdot \frac{( \ln n + \gamma ) ^ { k - 1 }}{(k-1)!}+\mathcal{O}\left(\frac{( \ln n ) ^ { k }}{k!\cdot n^2}\right) \quad (n \to \infty),   
\end{equation*}
uniformly for $1\leq k \leq \eta \ln{n}$.
\end{Hwang}

\begin{Erdos}
The Stirling numbers of the first kind form log-concave sequences. In addition, the signless Stirling number $s(n,k)$ is maximized at $k(n)=\argmax_{k \in \{\floor{H_n}, \ceil{H_n}\}}s(n,k)$, where $H_n=1+\frac{1}{2}+\ldots+\frac{1}{n}$ is the $n$-th harmonic number. That is, $k(n) \sim \ln{n}$.
\end{Erdos}

We now show the three statements of the lemma in sequence.
\begin{enumerate}[wide, labelindent=0pt]
    \item It follows immediately from Proposition 1.
    \item It follows immediately from Hwang's theorem for fixed $k \geq 2$. In fact, the original theorem pertaining to this case was offered by \cite{wilf_asymptotic_1993}.
    \item By Hwang's theorem,

\begin{equation*}
\tag{$\dagger$}
\label{dagger1}
\frac{s(n,k(n))}{n!} \sim \frac{1}{n} \cdot \frac{(\ln{n})^{\ln{n}-1}}{\Gamma(\ln{n})}\quad (n \to \infty,\; k(n) \sim \ln{n}).
\end{equation*}
By applying Stirling's formula to the gamma function,
\[\Gamma ( z ) = \sqrt { \frac { 2 \pi } { z } } \left( \frac { z } { e } \right) ^ { z } \left( 1 + \mathcal{O} \left( \frac { 1 } { z } \right) \right).\]
\noindent For $\Gamma(\ln{n})$, we get
\begin{equation*}
\tag{$\dagger\dagger$}
\label{dagger2}
    \Gamma(\ln{n}) \sim \frac{\sqrt{2\pi}}{n} \cdot (\ln{n})^{\ln{n}-\frac{1}{2}} \quad (n \to \infty),
\end{equation*}
\noindent so that from equations~\ref{dagger1} and~\ref{dagger2} we have
\[\frac{s(n,k(n))}{n!} \sim \frac{1}{\sqrt{2\pi}}\cdot(\ln{n})^{-\frac{1}{2}}\quad (n \to \infty,\; k(n) \sim \ln{n}).\]
\noindent Thus,
\begin{equation*}
\Pr\left(S^C(2,n)=k(n)\right) = \frac{s(n,k(n))}{n!} \cdot \left(1-\frac{1}{2^{k(n)-1}}\right) \sim  \frac{1}{\sqrt{2\pi}}\cdot(\ln{n})^{-\frac{1}{2}}\quad (n \to \infty,\; k(n) \sim \ln{n})
\end{equation*}
\noindent and $\lim\limits_{n \to \infty} (\ln{n})^{1/2} \cdot \Pr\left(S^C(2,n)=k(n)\right)=\dfrac{1}{\sqrt{2\pi}}$ as desired.\qedhere
\end{enumerate}
\end{proof}

\begin{lemma}
\label{lem: s_c_2_n_2} Consider a random game $G(2,n)$. Then,

\begin{enumerate}
\item for any $n \geq 1$, $\E\left[\left|S^C(2,n)\right|\right]= W(n) \cdot %
\big(2 - (\psi(n+1/2)-\psi(1/2))\big)+H_n$, where $\psi(z)=\dfrac{%
\Gamma^{\prime }(z)}{\Gamma(z)}$ is the digamma function and $W(n)=\dfrac{%
\Gamma(n+1/2)}{\Gamma(n+1)\cdot\Gamma(1/2)}$; 

\item $\E\left[\left|S^C(2,n)\right|\right]$ is strictly increasing in $n$; 

\item $\E\left[\left|S^C(2,n)\right|\right]=\ln{n} +\gamma - \dfrac{1}{%
\sqrt{\pi}} \cdot \dfrac{\ln{n}}{n^{1/2}}+ \mathcal{O}\left(\dfrac{1}{%
n^{1/2}}\right)$, where $\gamma$ is the Euler-Mascheroni
constant.
\end{enumerate}
\end{lemma}

\begin{proof}
We prove each of the three statements in turn.

\begin{enumerate}[wide, labelindent=0pt]
    \item Recall that
\begin{align*}
\Pr\left(\left|S^C(2,n)\right|=1\right) &= \frac{1}{n!} \cdot \sum_{k=1}^n \frac{s(n, k)}{2^{k-1}}\text{, and}\\
\Pr\left(\left|S^C(2,n)\right|=k\right) &= \frac{s(n, k)}{n!} \cdot \left(1-\frac{1}{2^{k-1}}\right) \quad \text
{for any } 2\leq k \leq n,
\end{align*}
\noindent so that
\begin{align*}
\E\left[\left|S^C(2,n)\right|\right] &= \frac{1}{n!} \cdot \sum_{k=1}^n \frac{s(n, k)}{2^{k-1}} + \sum_{k=2}^{n} k \cdot \frac{s(n, k)}{n!} \cdot \left(1-\frac{1}{2^{k-1}}\right)= \frac{1}{n!} \cdot \sum_{k=1}^n \frac{s(n, k)}{2^{k-1}} \\
&+ \sum_{k=1}^{n} k \cdot \frac{s(n, k)}{n!} \cdot \left(1-\frac{1}{2^{k-1}}\right)= 2W(n) + \sum_{k=1}^{n} k \cdot \frac{s(n, k)}{n!} - \sum_{k=1}^{n} k \cdot \frac{s(n, k)}{2^{k-1} \cdot n!}.
\end{align*}

First, note that $\displaystyle\sum_{k=1}^{n} k \cdot \dfrac{s(n, k)}{n!} =  H_n$ (see, e.g., Theorem 2 in \citealp{benjamin_stirling_2002}).

Second, we can differentiate the identity $\displaystyle\sum_{k=1}^{n} {s(n, k)} \cdot x^{k}=\dfrac{\Gamma(n+x)}{\Gamma(x)}$
to derive an explicit expression for $\displaystyle\sum_{k=1}^{n} k \cdot \dfrac{s(n, k)}{2^{k-1} \cdot n!}$. Namely,
\begin{equation*}
\sum_{k=1}^{n} k\cdot {s(n, k)} \cdot x^{k-1} =  \frac{d}{dx}\left(\frac{\Gamma(n+x)}{\Gamma(x)}\right) = \frac{\Gamma(n+x)}{\Gamma(x)} \cdot (\psi(n+x)-\psi(x)), 
\end{equation*}
\noindent where $\psi(z)=\dfrac{\Gamma'(z)}{\Gamma(z)}$, so that $\displaystyle\sum_{k=1}^{n} k \cdot \dfrac{s(n, k)}{2^{k-1} \cdot n!} = W(n) \cdot (\psi(n+1/2)-\psi(1/2))$.

By collecting all terms, we get
\begin{equation*}
\E\left[\left|S^C(2,n)\right|\right] =  W(n) \cdot \big(2 - (\psi(n+1/2)-\psi(1/2))\big)+H_n,
\end{equation*}
where $\psi(n+1/2) = - \gamma - 2 \ln 2 + \displaystyle\sum_{k = 1}^{n} \dfrac{2}{2k-1} = - \gamma + H_{n-1/2}$ and $\psi(1/2) = - \gamma - 2 \ln 2$.

\item For any $n \geq 1$,
\begin{align*}
\E\left[\left|S^C(2,n+1)\right|\right]&=  W(n+1) \cdot \big(2 - (\psi(n+3/2)-\psi(1/2)\big)+H_{n+1}\text{ and}  \\
\E\left[\left|S^C(2,n)\right|\right]&=  W(n) \cdot \big(2 - (\psi(n+1/2)-\psi(1/2)\big)+H_{n},
\end{align*}
\noindent so that
\begin{align*}
&\E\left[\left|S^C(2,n+1)\right|\right] - \E\left[\left|S^C(2,n)\right|\right] = \frac{1}{n+1} \\
& + W(n)\cdot \Bigg(-\frac{1}{n+1} + (\psi(n+1/2)-\psi(1/2)) - \frac{n+1/2}{n+1} \cdot \big(\psi(n+3/2)-\psi(1/2)\big)\Bigg)\\
&> W(n) \cdot \Bigg(\sum_{k=1}^{n} \frac{2}{2k-1} - \frac{n+1/2}{n+1} \cdot \sum_{k=1}^{n+1} \frac{2}{2k-1}\Bigg) = \dfrac{W(n)}{n+1} \cdot \Bigg(\sum_{k=1}^{n} \frac{1}{2k-1} - 1 \Bigg) \geq 0,
\end{align*}
\noindent where the first inequality follows from $W(n)=\dfrac{\pi(2,n)}{2}<1$ for any $n \geq 1$. 

\item Note that 
\begin{align*}
    H_n &= \ln{n} +\gamma + \mathcal{O}\left(\frac{1}{n}\right), \qquad
    \dfrac{\Gamma(n+1/2)}{\Gamma(n+1)} = \frac{1}{n^{1/2}} + \mathcal{O}\left(\frac{1}{n^{3/2}}\right),\qquad\text{and}\\
    \psi(n+1/2)&=\ln{(n+1/2)}+\mathcal{O}\left(\frac{1}{n}\right),
\end{align*}
\noindent so that
\begin{align*}
\E\left[\left|S^C(2,n)\right|\right]&= \frac{1}{\sqrt{\pi}} \cdot \left(\frac{1}{n^{1/2}} + \mathcal{O}\left(\frac{1}{n^{3/2}}\right)\right) \cdot \left(2 + \gamma + 2\ln{2}-\ln{(n+1/2)}-\mathcal{O}\left(\frac{1}{n}\right)\right) \\
&+ \ln{n} +\gamma + \mathcal{O}\left(\frac{1}{n}\right) = \ln{n} +\gamma - \frac{1}{\sqrt{\pi}} \cdot \frac{\ln{n}}{n^{1/2}}+ \mathcal{O}\left(\frac{1}{n^{1/2}}\right).
\end{align*}

In particular, $\lim\limits_{n \to \infty} \left(\E\left[\left|S^C(2,n)\right|\right] - \ln{n}\right) = \gamma$.\qedhere
\end{enumerate}
\end{proof}

\begin{lemma}
\label{lem: s_c_2_n_3} Consider a random game $G(2,n)$. Define the
polygamma function $\psi^{(m)}(x)$ of order $m$ as the $m$-th derivative of
the digamma function $\psi(x)\equiv \psi^{(0)}(x)$. In addition, let $%
H^{(m)}_{n} \equiv 1 + \dfrac{1}{2^m}+\ldots+\dfrac{1}{n^m}$ be the
generalized harmonic number of order $m$ of $n$. Then,

\begin{enumerate}
\item for any $n \geq 1$, 
\begin{multline*}
\Var\left[S^C(2,n)\right]= H_n-H^{(2)}_n + \dfrac{W(n)}{2} \cdot \Bigg(4 - 2 \cdot \big(%
\psi(n+1/2)-\psi(1/2)\big)\\
-\Big( \big(\psi(n+1/2)-\psi(1/2)\big)^2+\big(\psi^{(1)}(n+1/2)-%
\psi^{(1)}(1/2)\big) \Big) - 8\cdot H_n + 4 \cdot H_n \cdot \big(\psi(n+1/2)-\psi(1/2)\big) \Bigg)\\ 
- \left(W(n) \cdot \big(2 - (\psi(n+1/2)-\psi(1/2))\big)\right)^2.
\end{multline*}

\item $\Var\left[S^C(2,n)\right]=\ln{n}+\gamma - \dfrac{\pi^2}{6}+\dfrac{3}{%
2\cdot\sqrt{\pi}}\cdot\dfrac{(\ln{n})^2}{n^{1/2}} - \dfrac{5-2\ln{2}%
-3\gamma}{\sqrt{\pi}}\cdot \dfrac{\ln{n}}{n^{1/2}} +\mathcal{O}\left(\dfrac{%
1}{n^{1/2}}\right)$.
\end{enumerate}
\end{lemma}

\begin{proof} We prove two statements in sequence.

\begin{enumerate}[wide, labelindent=0pt]
    \item Similar to the proof of Lemma~\ref{lem: s_c_2_n_2},
\begingroup
\allowdisplaybreaks
\begin{align*}
\E\left[S^C(2,n)^2\right] &= \frac{1}{n!} \cdot \sum_{j=1}^n \frac{s(n, j)}{2^{j-1}} + \sum_{j=2}^{n} j^2 \cdot \frac{s(n, j)}{n!} \cdot \left(1-\frac{1}{2^{j-1}}\right)= \frac{1}{n!} \cdot \sum_{j=1}^n \frac{s(n, j)}{2^{j-1}} \\
&+ \sum_{j=1}^{n} j^2 \cdot \frac{s(n, j)}{n!} \cdot \left(1-\frac{1}{2^{j-1}}\right)= 2W(n) + \sum_{j=1}^{n} j^2 \cdot \frac{s(n, j)}{n!} - \sum_{j=1}^{n} j^2 \cdot \frac{s(n, j)}{2^{j-1} \cdot n!} \\
&= 2W(n) + \sum_{j=1}^{n} j^2 \cdot \frac{s(n, j)}{n!} -\sum_{j=1}^{n} j \cdot \frac{s(n, j)}{2^{j-1} \cdot n!} - \frac{1}{2} \cdot \sum_{j=1}^{n} j(j-1) \cdot \frac{s(n, j)}{2^{j-2} \cdot n!}.
\end{align*}
\endgroup

First, we have $\displaystyle\sum_{j=1}^{n} j^2 \cdot \dfrac{s(n, j)}{n!} = H_n+H^2_n-H^{(2)}_n$ (\citealp{gontcharoff_du_1944}). Second, note that $\displaystyle\sum_{j=1}^{n} j \cdot \dfrac{s(n, j)}{2^{j-1} \cdot n!} = W(n) \cdot (\psi(n+1/2)-\psi(1/2))$.

Third, we can differentiate twice the identity $\displaystyle\sum_{j=1}^{n} {s(n, j)} \cdot x^{j}=\dfrac{\Gamma(n+x)}{\Gamma(x)}$ to derive an explicit expression for $\displaystyle\sum_{j=1}^{n} j(j-1) \cdot \dfrac{s(n, j)}{2^{j-2} \cdot n!}$ as follows: 
\begin{align*}
\sum_{j=1}^{n} j(j-1)\cdot {s(n, j)} \cdot x^{j-2} &=  \frac{d^2}{dx^2}\left(\frac{\Gamma(n+x)}{\Gamma(x)}\right) = \frac{d}{dx}\left(\frac{\Gamma(n+x)}{\Gamma(x)} \cdot (\psi(n+x)-\psi(x))\right)\\
&= \frac{\Gamma(n+x)}{\Gamma(x)} \cdot \Big( \big(\psi(n+x)-\psi(x)\big)^2 + \big(\psi^{(1)}(n+x)-\psi^{(1)}(x)\big) \Big), 
\end{align*}
\noindent where $\psi ^ { ( m ) } ( z ) = \dfrac { d ^ { m } } { d z ^ { m } } \psi ( z )$ the polygamma function of order $m$, so that
\begin{align*}
\sum_{j=1}^{n} j(j-1) \cdot \frac{s(n, j)}{2^{j-2} \cdot n!} &= W(n) \cdot \Big( \big(\psi(n+1/2)-\psi(1/2)\big)^2 +\big(\psi^{(1)}(n+1/2)-\psi^{(1)}(1/2)\big) \Big).
\end{align*}
\noindent Thus,
\begin{align*}
\E\left[S^C(2,n)^2\right] &= 2W(n) + \left(H_n+H^2_n-H^{(2)}_n\right) -W(n) \cdot (\psi(n+1/2)-\psi(1/2))\\
&-\frac{W(n)}{2} \cdot \Big( \big(\psi(n+1/2)-\psi(1/2)\big)^2 +\big(\psi^{(1)}(n+1/2)-\psi^{(1)}(1/2)\big) \Big).
\end{align*}

By Lemma~\ref{lem: s_c_2_n_2},
\begin{align*}
\E\left[S^C(2,n)\right]=  W(n) \cdot \big(2 - (\psi(n+1/2)-\psi(1/2))\big)+H_n,
\end{align*}
\noindent so that by using $\Var\left[S^C(2,n)\right]=\E\left[S^C(2,n)^2\right] - \E\left[S^C(2,n)\right]^2$, we get the desired expression.
\item By using asymptotic expressions for each term, we get
\begin{align*}
\Var\left[S^C(2,n)\right]&= \ln{n}+\gamma - \frac{\pi^2}{6} +\frac{1}{2 \cdot \sqrt{\pi}} \cdot \left(\frac{1}{n^{1/2}} + \mathcal{O}\left(\frac{1}{n^{3/2}}\right)\right) \\
& \times \Bigg(-(\ln(n+1/2))^2+4 \cdot\ln(n+1/2) \cdot \ln{n}-2\ln(n+1/2) \\
&-(4\ln{2}+2\gamma)\ln{(n+1/2)}-8 \ln{n} + 4 \gamma \cdot \ln{n} \\
&+(8\ln{2}+4\gamma) \ln{(n+1/2)} + \mathcal{O}(1)\Bigg)\\
&=\ln{n}+\gamma - \frac{\pi^2}{6}+\frac{3}{2\cdot\sqrt{\pi}}\cdot\frac{(\ln{n})^2}{n^{1/2}} - \frac{5-2\ln{2}-3\gamma}{\sqrt{\pi}}\cdot \frac{\ln{n}}{n^{1/2}}  +\mathcal{O}\left(\frac{1}{n^{1/2}}\right),
\end{align*}
where $\lim_{n \to \infty} H^{(2)}_n=\dfrac{\pi^2}{6}$.\qedhere
\end{enumerate}
\end{proof}

\setcounter{prop}{2}
\begin{prop}
\label{prop: s_c_2_n} Consider a random game $G(2,n)$. Then, 
\begin{equation*}
\Pr \left( S^{C}(2,n)-\E\left[ S^{C}(2,n)\right] \leq x\cdot \sqrt{\Var\left[
S^{C}(2,n)\right] }\right) =\Phi (x)+\mathcal{O}\left( \frac{1}{\sqrt{\ln n}%
}\right) ,
\end{equation*}%
\noindent where $\Phi (\cdot )$ is the distribution function of the standard
normal distribution, 
\begin{align*}
\E\left[ S^{C}(2,n)\right] =\ln {n}+\gamma +o(1)\text{,\quad and \quad}
\sqrt{\Var\left[ S^{C}(2,n)\right] }=\sqrt{\ln {n}}- \frac{{\pi ^{2}}-6\gamma}{12\sqrt{\ln {n}}}+o\left( \frac{%
1}{\sqrt{\ln {n}}}\right).
\end{align*}
\end{prop}
\begin{proof}
The proof of this statement is similar to \cite{hwang_convergence_1998} and
uses the Berry-Esseen theorem to find the convergence rate in the stated
central-limit result. The difference is that the problem does not belong to
the exp-log class immediately (one can observe several exp-terms below after
some manipulations). Still, similar reasoning can be applied to establish
the result. For notation simplicity, we let
\begin{align*}
\mu_n &\equiv \E\left[S^C(2,n)\right]=\ln{n}+\gamma + o(1),\\
\sigma_n &\equiv \sqrt{\Var\left[S^C(2,n)\right]}=\sqrt{\ln{n}}- \left(%
\frac{\pi^2}{12}-\frac{\gamma}{2}\right)\cdot \frac{1}{\sqrt{\ln{n}}}%
+o\left(\frac{1}{\sqrt{\ln{n}}}\right)\text{, and}\\
\varphi_n(t)&=\sum_{j=1}^{n} \Pr\left(S^C(2,n)=j\right) \cdot e^{it(j-\mu_n)/\sigma_n},
\end{align*}
\noindent where $\varphi _n(t)$ denotes the characteristic function of the normed variable $(S^C(2,n)-\mu_n)/{\sigma_n}$.

\begin{Berry}
Let $F ( x )$ be a nondecreasing function and $G(x)$ a differentiable function
of bounded variation on the real line. The corresponding Fourier-Stieltjes transforms $\varphi ( t )$ and $\gamma ( t )$ are then: 
\begin{equation*}
\varphi ( t ) = \int _ { - \infty } ^ { \infty } e ^ { i t x } d F ( x ) ,
\quad \gamma ( t ) = \int _ { - \infty } ^ { \infty } e ^ { i t x } d G ( x
).
\end{equation*}
Suppose that $F ( - \infty ) = G ( - \infty )$, $F ( \infty ) = G ( \infty )$%
, $T$ is an arbitrary positive number, and $\left| G ^ { \prime } ( x )
\right| \leq A$. Then, for every $b > 1 / ( 2 \pi )$, we have 
\begin{equation*}
\sup _ { - \infty < x < \infty } | F ( x ) - G ( x ) | \leq b \int _ { - T }
^ { T } \left| \frac { \varphi ( t ) - \gamma ( t ) } { t } \right| d t + r
( b ) \frac { A } { T },
\end{equation*}
where $r(b)$ is a positive constant depending only on $b$.
\end{Berry}

We proceed in two steps. In step 1, we reformulate the problem by using the
Berry-Esseen inequality. In step 2, we calculate the characteristic function
and use it to establish the result.\bigskip

\noindent {\textit{Step 1. Reformulated problem}}

Let $G ( x ) = \Phi ( x )$ (so that $A = 1 / \sqrt { 2 \pi }$) and $T
= T _ { n } = c\sigma _ { n }$, where $c>0$ is a sufficiently small
constant. By the Berry-Esseen inequality, it will be sufficient to prove
that 
\begin{align*}
J _ { n } = \int _ { - T _ { n } } ^ { T _ { n } } \left| \frac { \varphi _
{ n } ( t ) - e ^ { - \frac { 1 } { 2 } t ^ { 2 } } } { t } \right| d t = 
\mathcal{O} \left( \frac{1}{\sqrt{\ln{n}}} \right).
\end{align*}

\noindent {\textit{Step 2. Characteristic function}}
\begin{align*}
\varphi_n(t)&=\sum_{j=1}^{n} \Pr\left(S^C(2,n)=j\right) \cdot
e^{it(j-\mu_n)/\sigma_n} \\
&=e^{-it\mu_n/\sigma_n} \cdot \left( \frac{1}{n!} \cdot \sum_{j=1}^n \frac{%
s(n, j)}{2^{j-1}} \cdot e^{it/\sigma_n} + \sum_{j=2}^{n} \frac{s(n, j)}{n!}
\cdot \left(1-\frac{1}{2^{j-1}}\right) \cdot e^{ itj/\sigma_n}\right) \\
&=e^{-it\mu_n/\sigma_n} \cdot \left( \frac{1}{n!} \cdot \sum_{j=1}^n \frac{%
s(n, j)}{2^{j-1}} \cdot e^{it/\sigma_n} + \sum_{j=1}^{n} \frac{s(n, j)}{n!}
\cdot \left(1-\frac{1}{2^{j-1}}\right) \cdot e^{itj/\sigma_n}\right) \\
&= e^{-it\mu_n/\sigma_n} \cdot \left(\sum_{j=1}^{n} \frac{s(n, j)}{n!} \cdot
e^{itj/\sigma_n}\right) \\
&+ 2 \cdot e^{-it\mu_n/\sigma_n} \cdot \left(\frac{1}{n!} \cdot \sum_{j=1}^n 
\frac{s(n, j)}{2^{j}} \cdot e^{it/\sigma_n} - \sum_{j=1}^{n} \frac{s(n, j)}{%
n!} \cdot \frac{1}{2^{j}} \cdot e^{itj/\sigma_n} \right) =A_n(t)+B_n(t),
\end{align*}
where 
\begin{align*}
A_n(t) &\equiv e^{-it\mu_n/\sigma_n} \cdot\frac{1}{\Gamma
\left(e^{it/\sigma_n} \right)}\cdot \frac{\Gamma\left(n + e^{it/\sigma_n}
\right)}{\Gamma(n+1)}\text{, and} \\
B_n(t) &\equiv 2 \cdot e^{-it\mu_n/\sigma_n} \cdot\left(\frac{e^{it/\sigma_n} 
}{\Gamma(1/2)}\cdot \frac{\Gamma(n+1/2)}{\Gamma(n+1)} - \frac{1}{\Gamma
\left(e^{it/\sigma_n}/2 \right)}\cdot \frac{\Gamma\left(n +
e^{it/\sigma_n}/2 \right)}{\Gamma(n+1)} \right).
\end{align*}

We start by finding the asymptotic expression for $A_n(t)$. By denoting $%
e^{it/\sigma_n} \equiv 1 + \varepsilon_n$ with $\varepsilon_n = \dfrac{it}{%
\sigma_n}-\dfrac{t^2}{2\sigma^2_n}+\mathcal{O}\left(\dfrac{|t|^3}{\sigma^3_n}%
\right)$ and using Stirling's formula,

\begingroup
\allowdisplaybreaks
\begin{multline*}
\ln \frac { \Gamma \left( n + e ^{it/\sigma_n} \right) } { \Gamma ( n + 1 ) 
} = \ln{\Gamma \left( n + 1 + \varepsilon _ { n } \right)} - \ln{\Gamma
\left( n + 1 \right)}\\
= (n+1/2+\varepsilon _ { n })\cdot \ln{(n+1+\varepsilon _ { n })} -
(n+1+\varepsilon _ { n }) + \frac{1}{2}\ln {2\pi} 
- (n+1/2)\cdot \ln{(n+1)} + (n+1) - \frac{1}{2} \ln {2\pi} + 
\mathcal{O}\left(\frac{1}{n}\right) \\
= \varepsilon _ { n } \ln (n+1+ \varepsilon _ { n }) + (n+1/2)\ln\left(1+%
\frac{\varepsilon _ { n }}{n+1}\right) - \varepsilon _ { n } + \mathcal{O}%
\left(\frac{1}{n}\right) 
= \varepsilon _ { n } \ln n + \mathcal{O}\left(\frac{1}{n}\right) \\
= \left(\dfrac{it}{\sigma_n}-\dfrac{t^2}{2\sigma^2_n}+\mathcal{O}\left(%
\dfrac{|t|^3}{\sigma^3_n}\right)\right) \cdot \ln n + \mathcal{O}\left(%
\frac{1}{n}\right) 
= \left(\dfrac{it}{\sqrt{\ln{n}}}-\dfrac{t^2}{2\ln{n}}+\mathcal{O}\left(%
\dfrac{|t|^3}{(\ln{n})^{3/2}}\right)\right) \cdot \ln n + \mathcal{O}\left(%
\frac{1}{n}\right) \\
= it \cdot \sqrt{\ln{n}} -\dfrac{t^2}{2} + \mathcal{O}\left(\dfrac{|t|^3}{%
\sqrt{\ln{n}}}\right).
\end{multline*}
\endgroup

\noindent In addition, 
\begin{align*}
\ln \Gamma \left(e^{it/\sigma_n} \right) &= \ln \Gamma\left(1 + \mathcal{O}%
\left(\dfrac{|t|}{\sqrt{\ln{n}}}\right)\right) = \mathcal{O}\left(\dfrac{|t|%
}{\sqrt{\ln{n}}}\right)\qquad\text{and}\\
\frac{it\mu_n}{\sigma_n}&=it \cdot \frac{\ln{n}+\gamma + o(1)}{\sqrt{\ln{n}}%
- \left(\frac{\pi^2}{12}-\frac{\gamma}{2}\right)\cdot \frac{1}{\sqrt{\ln{n}}%
}+o\left(\frac{1}{\sqrt{\ln{n}}}\right)} = it \cdot\sqrt{\ln{n}} + 
\mathcal{O}\left(\dfrac{|t|}{\sqrt{\ln{n}}}\right).
\end{align*}

\noindent Collecting all terms, we get 
\begin{align*}
A_n(t) = e^{-\left(it \cdot\sqrt{\ln{n}} + \mathcal{O}\left(\frac{|t|}{%
\sqrt{\ln{n}}}\right)\right)} \cdot e^{-\mathcal{O}\left(\frac{|t|}{\sqrt{%
\ln{n}}}\right)} \cdot e^{it \cdot \sqrt{\ln{n}} -\frac{t^2}{2} +\mathcal{O%
}\left(\frac{|t|^3}{\sqrt{\ln{n}}}\right)} = e^{-\frac{t^2}{2} + \mathcal{O}%
\left(\frac{|t|+|t|^3}{\sqrt{\ln{n}}}\right)}.
\end{align*}

Next, we find the asymptotic expression for $B_n(t)$. Using similar
calculations, 
\begin{multline*}
\ln \frac { \Gamma \left( n + 1/2 \right) } { \Gamma ( n + 1 ) } = \ln{%
\Gamma \left( n + 1/2 \right)} - \ln{\Gamma \left( n + 1 \right)} \\
= n\cdot \ln{(n+1/2)} - (n+1/2) + \frac{1}{2}  \ln {2\pi} 
- (n+1/2)\cdot \ln{(n+1)} + (n+1) - \frac{1}{2}  \ln {2\pi} + 
\mathcal{O}\left(\frac{1}{n}\right) \\
= -\frac{1}{2} \cdot \ln{(n+1)} + n\ln\left(1-\frac{1}{2}\cdot \frac{1}{%
n+1}\right) + 1/2 + \mathcal{O}\left(\frac{1}{n}\right) 
=-\frac{1}{2} \cdot \ln{n}+\mathcal{O}\left(\frac{1}{n}\right),
\end{multline*}
\noindent and 
\begin{multline*}
\ln \frac { \Gamma \left( n + e ^{it/\sigma_n}/2 \right) } { \Gamma ( n + 1
) } = \ln{\Gamma \left( n + 1/2 + \varepsilon _ { n }/2 \right)} - \ln{%
\Gamma \left( n + 1 \right)} \\
= (n+\varepsilon _ { n }/2)\cdot \ln{(n+1/2+\varepsilon _ { n }/2)} -
(n+1/2+\varepsilon _ { n }/2) + \frac{1}{2} \ln {2\pi} \\
- (n+1/2)\cdot \ln{(n+1)} + (n+1) - \frac{1}{2} \ln {2\pi} + 
\mathcal{O}\left(\frac{1}{n}\right) \\
= (\varepsilon _ { n }/2-1/2) \ln{(n+1)} + (n+\varepsilon _ { n }/2) \ln
\left(1+\frac{\varepsilon _ { n }/2-1/2}{n+1}\right)
+ (1/2-\varepsilon _ { n }/2) + \mathcal{O}\left(\frac{1}{n}\right) \\
= - \frac{1}{2} \ln{n} + \frac{1}{2}it\cdot \sqrt{\ln{n}} + \mathcal{%
O}\left(t^2\right).
\end{multline*}

\noindent Furthermore, 
\begin{align*}
\ln \Gamma \left(e^{it/\sigma_n}/2 \right) &= \ln \Gamma\left(1/2 + 
\mathcal{O}\left(\dfrac{|t|}{\sqrt{\ln{n}}}\right)\right) = \mathcal{O}%
\left(1\right)\qquad\text{and}\\
\frac{it}{\sigma_n}&=it\cdot \frac{1}{{\sqrt{\ln{n}}- \left(\frac{\pi^2}{12}%
-\frac{\gamma}{2}\right)\cdot \frac{1}{\sqrt{\ln{n}}}+o\left(\frac{1}{\sqrt{%
\ln{n}}}\right)}} = \mathcal{O}\left(\dfrac{|t|}{\sqrt{\ln{n}}}\right).
\end{align*}

\noindent Collecting all terms, we get 
\begin{align*}
B_n(t) &= \frac{2}{\sqrt{\pi}} \cdot e^{-\left(it \cdot\sqrt{\ln{n}} + 
\mathcal{O}\left(\frac{|t|}{\sqrt{\ln{n}}}\right)\right)} \cdot e^{\mathcal{%
O}\left(\frac{|t|}{\sqrt{\ln{n}}}\right)} \cdot e^{-\frac{1}{2} \cdot \ln{n%
} + \mathcal{O}\left(\frac{1}{n}\right)} \\
&- 2 \cdot e^{-\left(it \cdot\sqrt{\ln{n}} + \mathcal{O}\left(\frac{|t|}{%
\sqrt{\ln{n}}}\right)\right)} \cdot e^{-\mathcal{O}\left(1\right)} \cdot
e^{- \frac{1}{2}\cdot \ln{n} + \frac{1}{2}it\cdot \sqrt{\ln{n}} + \mathcal{%
O}\left(t^2\right)}.
\end{align*}

Note that $B_n(0)=0$ and $B_n(s)=\mathcal{O}\left(\frac{e^{\tau \cdot \sqrt{\ln{n}}}}{%
n^{1/2}}\right)$ uniformly for $|s| \leq \tau$, $s \in \mathcal{C}$, for some
fixed $\tau>0$. By denoting $\kappa_n \equiv \frac{n^{1/2}}{e^{\tau \cdot \sqrt{\ln{n}}}}$ for convenience, we can rewrite $B_n(s)=\mathcal{O}\left(\frac{1}{\kappa_n}\right)$ for $|s| \leq \tau$. Furthermore, by taking a small circle around the origin we easily obtain $B_n(s)=\mathcal{O}\left(\frac{|s|}{\kappa_n}\right)$ for $|s| \leq c<\tau$, where sufficiently small $c>0$ can be taken less than $\tau$. Consequently,
\begin{align*}
\varphi _ { n } ( t )=A_n(t) + B_n(t) = e^{-\frac{t^2}{2} + \mathcal{O}\left(\frac{|t|+|t|^3%
}{\sqrt{\ln{n}}}\right)} + \mathcal{O}\left(\frac{|t|}{\kappa_n \cdot \sqrt{\ln n}}\right),
\end{align*}
for $| t | \leq T _ { n }=c\sigma_n$.

In fact, we can use Levy's convergence theorem to obtain the convergence
result. However, we still need to use the Berry-Esseen inequality to find
the convergence rate.

Based on the obtained approximation, we can follow the proof of Theorem 1 in \cite%
{hwang_convergence_1998}. That is, using the inequality $\left| e ^ { w } -
1 \right| \leq | w | e ^ { | w | }$ for all complex $w$, we obtain 
\begin{align*}
\left| \frac { \varphi _ { n } ( t ) - e ^ { - \frac { 1 } { 2 } t ^ { 2 } } 
} { t } \right| &= \mathcal{O} \left( \left( \frac { 1 + t ^ { 2 } } { \sqrt{%
\ln{n}} } \right) \exp \left( - \frac { t ^ { 2 } } { 2 } + \mathcal{O}
\left( \frac { | t | + | t | ^ { 3 } } { \sqrt{\ln{n}} } \right) \right) + 
\frac { 1 } { {\kappa_n \cdot \sqrt{\ln{n}}} } \right) \\
&= \mathcal{O} \left( \left( \frac { 1 + t ^ { 2 } } { \sqrt{\ln{n}} }
\right) e ^ { - \frac { 1 } { 4 } t ^ { 2 } } + \frac { 1 } { {\kappa_n
\cdot \sqrt{\ln{n}}} } \right) \quad \left( | t | \leq T _ { n } \right),
\end{align*}
for sufficiently small $0<c<\tau$. 

Thus, 
\begin{align*}
J_n &= \int _ { - T _ { n } } ^ { T _ { n } } \left| \frac { \varphi _ { n }
( t ) - e ^ { - \frac { 1 } { 2 } t ^ { 2 } } } { t } \right| d t = \mathcal{O} \left( \frac { 1 } { \sqrt{\ln{n}} } \int _ { - T _ { n } }
^ { T _ { n } } \left( 1 + t ^ { 2 } \right) e ^ { - \frac { 1 } { 4 } t ^ {
2 } } d t + \frac { 1 } { {\kappa_n} } \right) \\
&= \mathcal{O} \left( \frac { 1 } { \sqrt{\ln{n}} } + \frac { 1 } { {%
\kappa_n} } \right) = \mathcal{O} \left( \frac { 1 } { \sqrt{\ln{n}} } \right),
\end{align*}
\noindent because $\lim\limits_{n \to \infty} \dfrac{\sqrt{\ln{n}}}{\kappa_n} = \lim\limits_{n \to \infty} 
\dfrac{\sqrt{\ln{n}} \cdot e^{\tau \cdot \sqrt{\ln{n}}}}{n^{1/2}} 
= \lim \limits_{n \to \infty} \dfrac{n \cdot e^{\tau \cdot n}}{e^{n^2/2}} =0$. Due to Step 1, this concludes the proof.\qedhere
\end{proof}

\newpage
\section{Enumerative Issues}
\label{appendix_B}
To illustrate the enumerative challenges posed by games in which both players have many agents, consider games in which Row has $m=3$ actions. We now focus on the basic problem of finding the probability that Column has no strictly dominated actions. Fix the first row of Column's payoff matrix to be the identity permutation: $c_{1\cdot }=e_{n}\equiv (1,2, \ldots, n)$. There are no strictly dominated actions for Column
if and only if a pair of permutations $(c_{2\cdot },c_{3\cdot })$ avoids the 
permutation pattern in which $c_{2j}>c_{2i}$ and $c_{3j}>c_{3i}$ for some $j>i$. 
This particular avoidance imposes restrictions on the
pair of \textit{i.i.d.} uniform $(c_{2\cdot },c_{3\cdot })$ that has been
the main object of interest for \cite{hammett_how_2008}. Formally, our
problem can be equivalently reformulated in terms of what is termed \textit{%
parallel permutation patterns }in enumerative combinatorics (see our Lemma~%
\ref{lem: bruhat} below). This problem lies at the
research frontier of that literature (\citealp{hammett_how_2008,
gunby_asymptotics_2019}). To make things worse, in general, permutation
patterns induced by strict dominance are different from those studied in the
literature on permutation avoidance.

As a numerical demonstration, Table~\ref{table: issues} displays the number
of possible Column's matrices with one fixed payoff row for $m=3$ and $n\in
\lbrack 6]$ corresponding to exactly $k$ strictly undominated actions, $k\in
\lbrack n]$.\footnote{%
These numbers correspond to exact computations for all $(n!)^{2}$ possible
combinations.} These numbers can be viewed as a generalization of the
unsigned Stirling numbers of the first kind for $m=3$. In particular, the
underlined sequence corresponds to the number of incidents in which Column
has no strictly dominated actions, as described above. The table suggests
properties similar to those of the standard Stirling numbers. Values appear
to be log-concave (and unimodal) and asymptotically normal with faster
convergence rates. This hints at the qualitative similarities
between the general $m$ by $n$ case and the particular $2$ by $n$ case
studied in Section 3 of the main text.

\renewcommand\thetable{\thesection.\arabic{table}}

\begin{table}[ht]
\centering
\begin{tabular}{c|cccccc}
\multicolumn{1}{l}{$n$} &  &  &  &  &  &  \\ \cline{1-1}
1 & \underline{1} &  &  &  &  &  \\ 
2 & 1 & \underline{3} &  &  &  &  \\ 
3 & 4 & 15 & \underline{17} &  &  &  \\ 
4 & 36 & 147 & 242 & \underline{151} &  &  \\ 
5 & 576 & 2460 & 4775 & 4690 & \underline{1899} &  \\ 
6 & 14400 & 63228 & 134909 & 164193 & 109959 & \underline{31711} \\ \hline
\multicolumn{1}{l}{} & 1 & 2 & 3 & 4 & 5 & 6 \\ \cline{2-7}
\multicolumn{1}{l}{} & \multicolumn{6}{c}{$k$}%
\end{tabular}
\captionsetup{justification=centering}
\caption{$(n!)^2 \cdot \Pr(U^C(3,n)=k)$, $k \in \lbrack n \rbrack$: exact
calculations (the \protect\underline{underlined} sequence corresponds to sequence
\href{https://oeis.org/A007767}{A007767} in the \href{https://oeis.org/}{OEIS})}
\label{table: issues}
\end{table}

The combinatorics community has accumulated knowledge of many number
sequences summarized in the On-Line Encyclopedia of Integer Sequences (\href{https://oeis.org/}%
{OEIS}). It is worth noting that none of the sequences corresponding to any
dimension of our analysis has been enumerated before in the \href{https://oeis.org/}%
{OEIS}. This suggests that our fundamental combinatorial object has not been
studied previously.

To state Lemma~\ref{lem: bruhat}, we need to introduce three additional
definitions related to the literature on permutation patterns. First, we say
that for $\sigma _ { 1 } , \ldots , \sigma _ { d } \in S _ { n }$ and $%
\sigma _ { 1 } ^ { \prime } , \ldots , \sigma _ { d } ^ { \prime } \in S _ {
m }$, $\left( \sigma _ { 1 } , \ldots , \sigma _ { d } \right)$ \textit{%
avoids} $\left( \sigma _ { 1 } ^ { \prime } , \ldots , \sigma _ { d } ^ {
\prime } \right)$ if there does not exist indices $c _ { 1 } < \cdots < c _
{ m }$ such that $\sigma _ { i } \left( c _ { 1 } \right) \sigma _ { i }
\left( c _ { 2 } \right) \cdots \sigma _ { i } \left( c _ { m } \right)$ is
order-isomorphic to $\sigma _ { i } ^ { \prime }$ for all $i$ (e.g., see %
\citealp{gunby_asymptotics_2019, klazar_counting_2000}). Second, we say that 
$\pi \preceq \sigma$ in the \textit{weak Bruhat Order} if there is a chain $%
\sigma=\omega_{1} \rightarrow \omega_{2} \rightarrow \cdots \rightarrow
\omega_{s}=\pi$, where each $\omega_{t}$ is a \textit{simple reduction} of $%
\omega_{t-1}$, i.e. obtained from $\omega_{t-1}$ by transposing two \textit{%
adjacent} elements $\omega_{t-1}(i), \omega_{t-1}(i+1)$ with $%
\omega_{t-1}(i)>\omega_{t-1}(i+1)$. Equivalently (see Lemma 4.1 in %
\citealp{hammett_how_2008}), $\pi \preceq \sigma$ in the \textit{weak Bruhat
Order} if $I(\pi) \subseteq I(\sigma)$, where for any $\omega \in S_n$, the
inversion set $I(\omega)=\{(i,j) \mid i<j\text{ with }\omega^{-1}(i)>%
\omega^{-1}(j)\}$ is defined to be the set of all inversions in $\omega$.
Finally, for any $\sigma \in S_n$, let $\sigma^\star \in S_n$ denote its 
\textit{complement}, i.e. $\sigma^{\star}(i)=n+1-\sigma(i)$.

\begin{lemma}
\label{lem: bruhat} Consider a random game $G(3,n)$. Then, for any $n \geq 1$%
, 
\begin{equation*}
\prod_{i=1}^{n}(H_i / i)\leq \Pr
(U^C(3,n)=n)=\Pr\left((c_{3\cdot}^{\star})^{-1} \preceq
c_{2\cdot}^{-1}\right) \leq (0.362)^{n}.
\end{equation*}
\end{lemma}

\begin{proof}
As in Lemma 1, we can set $c_{1\cdot}$ to $e_n$. This is without loss of generality.

All Column's actions are undominated if and only if $(c_{2\cdot}, c_{3\cdot})$ avoids $(12,12)$. This holds whenever, for any $i<j$ with $c_{2i}<c_{2j}$, we have $c_{3i}>c_{3j}$, or equivalently $c_{3i}^\star<c_{3j}^\star$. In other words, the set $I(c_{2\cdot}^{-1})$ of inversions of $c_{2\cdot}$ contains the set $I\left((c_{3\cdot}^\star)^{-1}\right)$ of inversions of $\left(c_{3\cdot}^\star\right)^{-1}$, i.e. $I\left((c_{3\cdot}^\star)^{-1}\right) \subseteq I(c_{2\cdot}^{-1})$. This occurs if and only if $(c_{3\cdot}^{\star})^{-1} \preceq c_{2\cdot}^{-1}$. 

Certainly, $\Pr\left((c_{3\cdot}^{\star})^{-1} \preceq c_{2\cdot}^{-1}\right)=\Pr\left(\pi \preceq \sigma\right)$, where $\sigma, \pi \in S_n$ are selected independently and uniformly at random. Probability bounds for this problem have been studied by \cite{hammett_how_2008}.\qedhere
\end{proof}

\newpage
\section{Many Players}
\label{appendix_C}
Consider a random $n$-person game $G(m_1,m_2,\ldots,m_n)$, where player $k$ has $m_k$ actions, $k \in [n]$. Then, player $k$ has the number of her undominated actions $U^k\left(m_1,m_2,\ldots,m_n\right)$ that is equivalent to $U^R\left(m_k, \prod_{i \neq k} m_i\right)$ already examined in Section 4.1. Because the number of other players' profiles $\prod_{i \neq k} m_i$ is multiplicative and becomes large even in small games, it is problematic for any player to eliminate any of her actions.

\renewcommand\thefigure{\thesection.\arabic{figure}}
\begin{figure}[h]
    \centering
    \begin{subfigure}[t]{.49\textwidth}
    \centering
    \includegraphics[width=\linewidth]{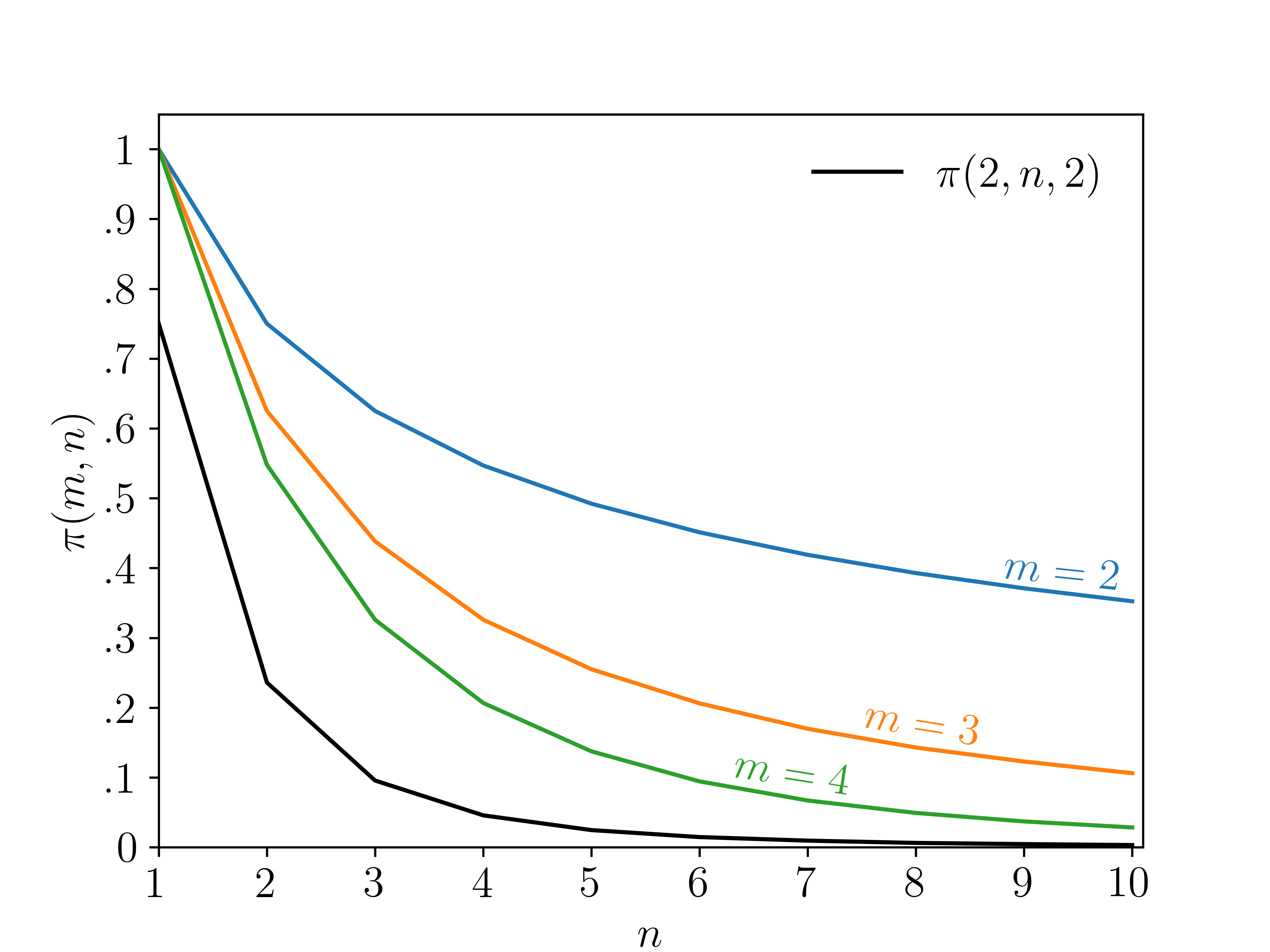}
            \caption{Probability of strict dominance}\label{fig: pi_three}
    \end{subfigure}
    \begin{subfigure}[t]{.49\textwidth}
    \centering
    \includegraphics[width=\linewidth]{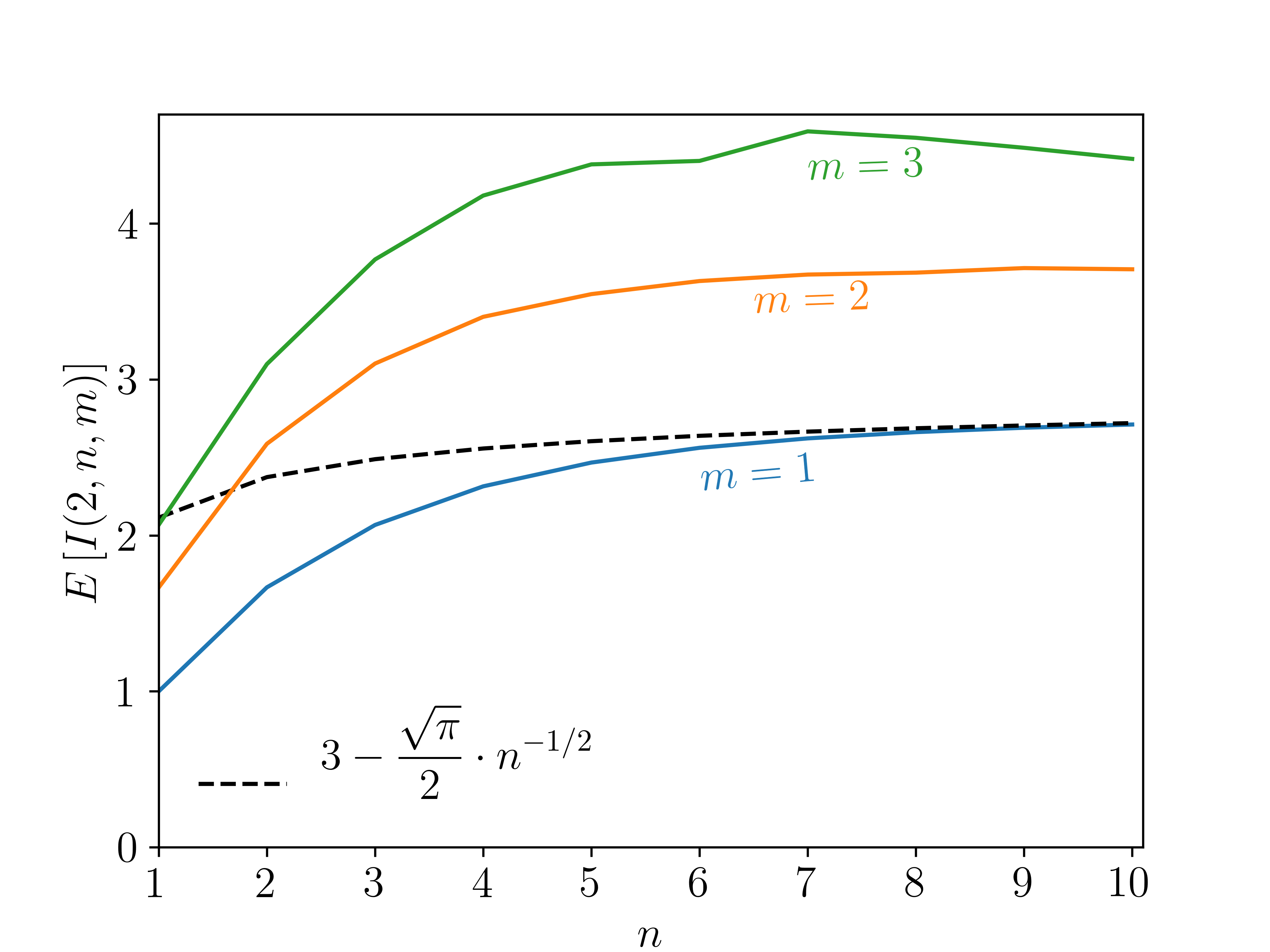}
    \caption{Conditional number of iterations}\label{fig: I_three}
    \end{subfigure}
    \par
    \medskip
    \par
    \begin{subfigure}[t]{\textwidth}
    \centering
    \vspace{0pt}
    \includegraphics[width=0.49\linewidth]{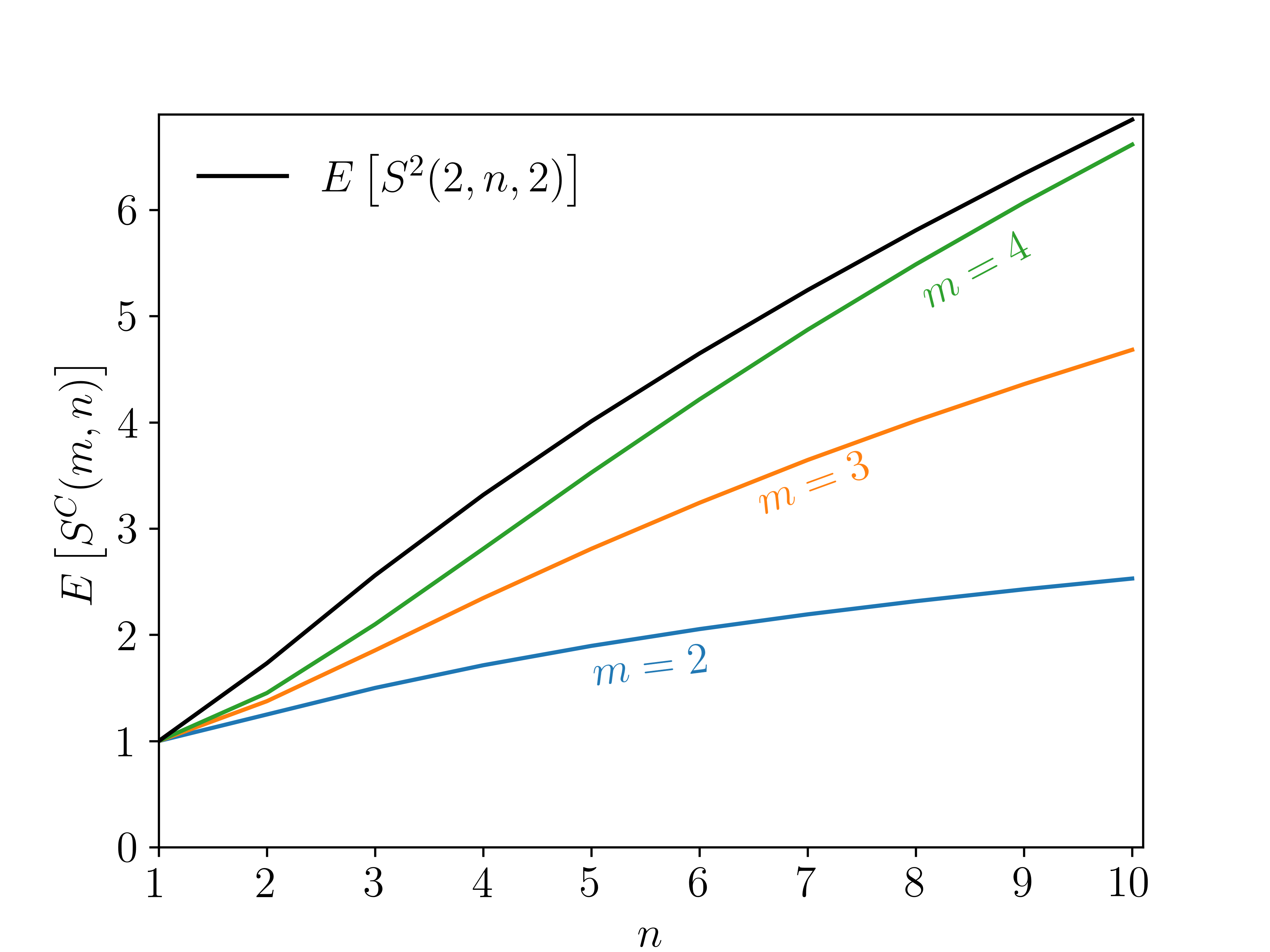}
    \caption{Surviving actions}\label{fig: S_column_three}
    \end{subfigure}
    \caption{Three dimensions of dominance solvability in games with three players}
    \label{ThreePlayers}
    \end{figure}   

Therefore, results pertaining to dominance solvability features are straightforward even for small games with at least three players and closely trace results for two-player games in which players' action sets are comparably large. First, getting strict-dominance solvability is challenging even for small games since it is unlikely players can eliminate any action at all. Second, surviving games are likely to coincide with the games prior to the process of elimination---iterated elimination is statistically ineffective at simplifying games. Last, since it is even more demanding for any player to eliminate many actions at each elimination round, we expect the conditional number of iterations to be large even in small games. Figure~\ref{ThreePlayers} illustrates these insights for three-player games. In particular, it shows that for random three-player games $G(2,n,m)$ with $G(2,n,1)$ being equivalent to a two-player game $G(2,n)$, the conditional number of iterations is component-wise increasing and exceeds 3 even for small game dimensions.

Because we have qualitative results even for small games, asymptotics is less interesting for random games with many, at least three, players. However, similar arguments to those used in the paper can be exploited to derive analogues of our main results.

\newpage
\section{Additional Figures}
\label{appendix_D}

\subsection{Alternative Distributional Assumptions in Imbalanced Games}

Figure~\ref{fig: alt_imbal} is the analogue of Figure 5 in the main text and displays the three dimensions of dominance solvability in $3 \times n$ games, rather than $n \times n$ games, for the various distributional assumptions we inspect.

\setcounter{figure}{0}
\begin{figure}[h]
\centering
\begin{subfigure}[t]{.49\textwidth}
\centering
\includegraphics[width=\linewidth]{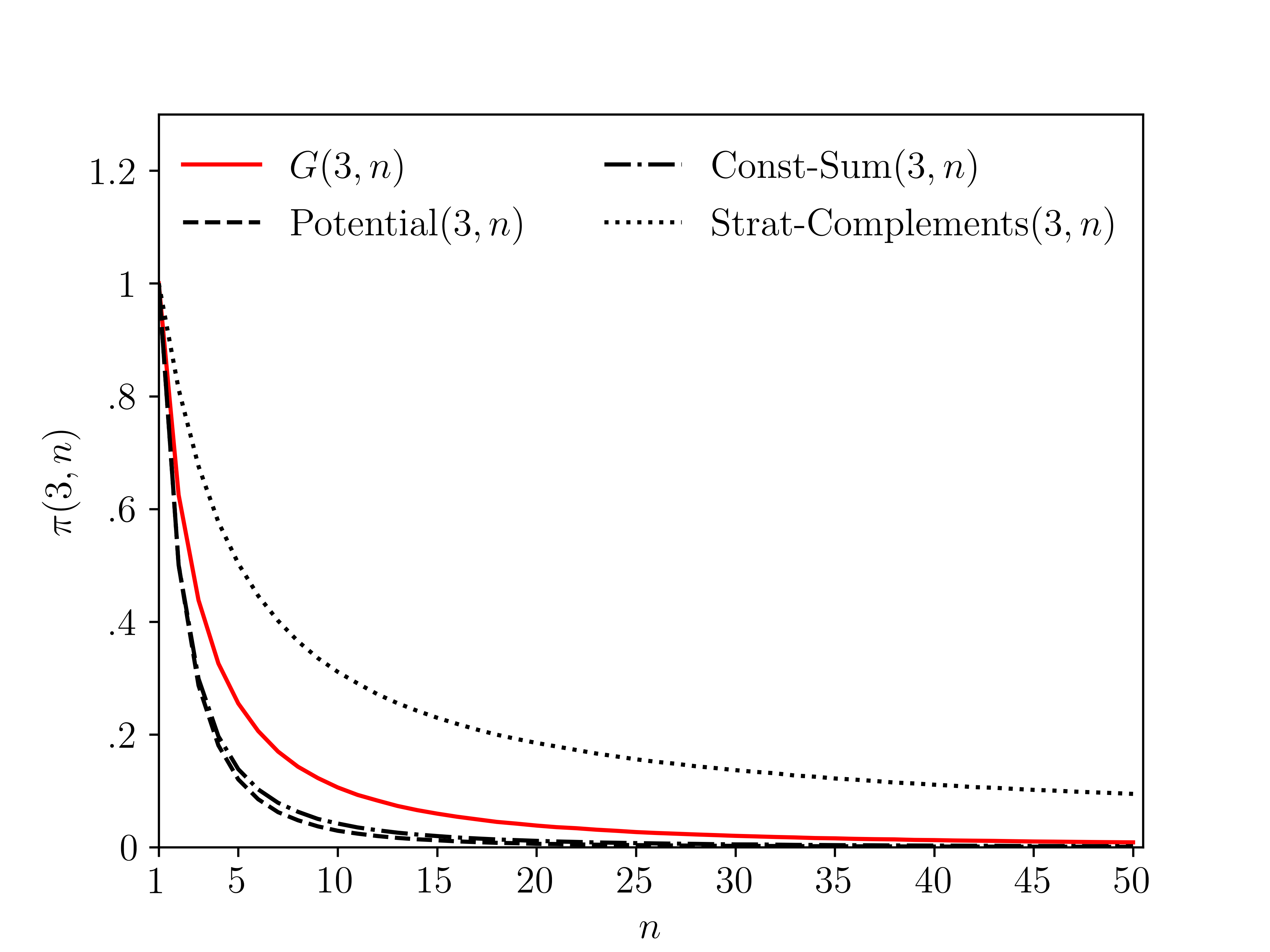}
        \caption{Probability of strict dominance}\label{fig: pi_alt_imbal}
\end{subfigure}
\begin{subfigure}[t]{.49\textwidth}
\centering
\includegraphics[width=\linewidth]{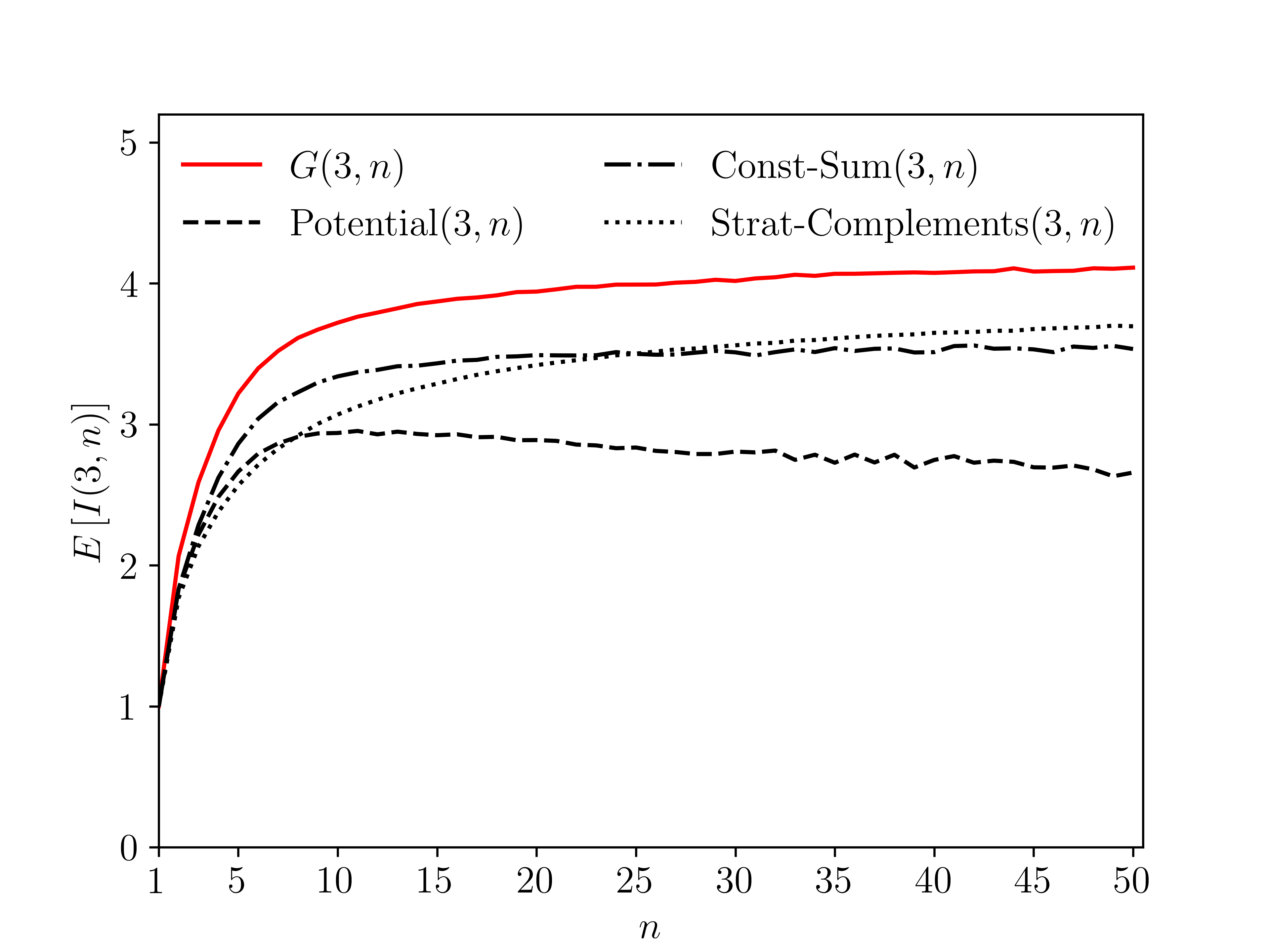}
\caption{Conditional number of iterations}\label{fig: I_alt_imbal}
\end{subfigure}
\par
\medskip
\par
\begin{subfigure}[t]{\textwidth}
\centering
\vspace{0pt}
\includegraphics[width=0.49\linewidth]{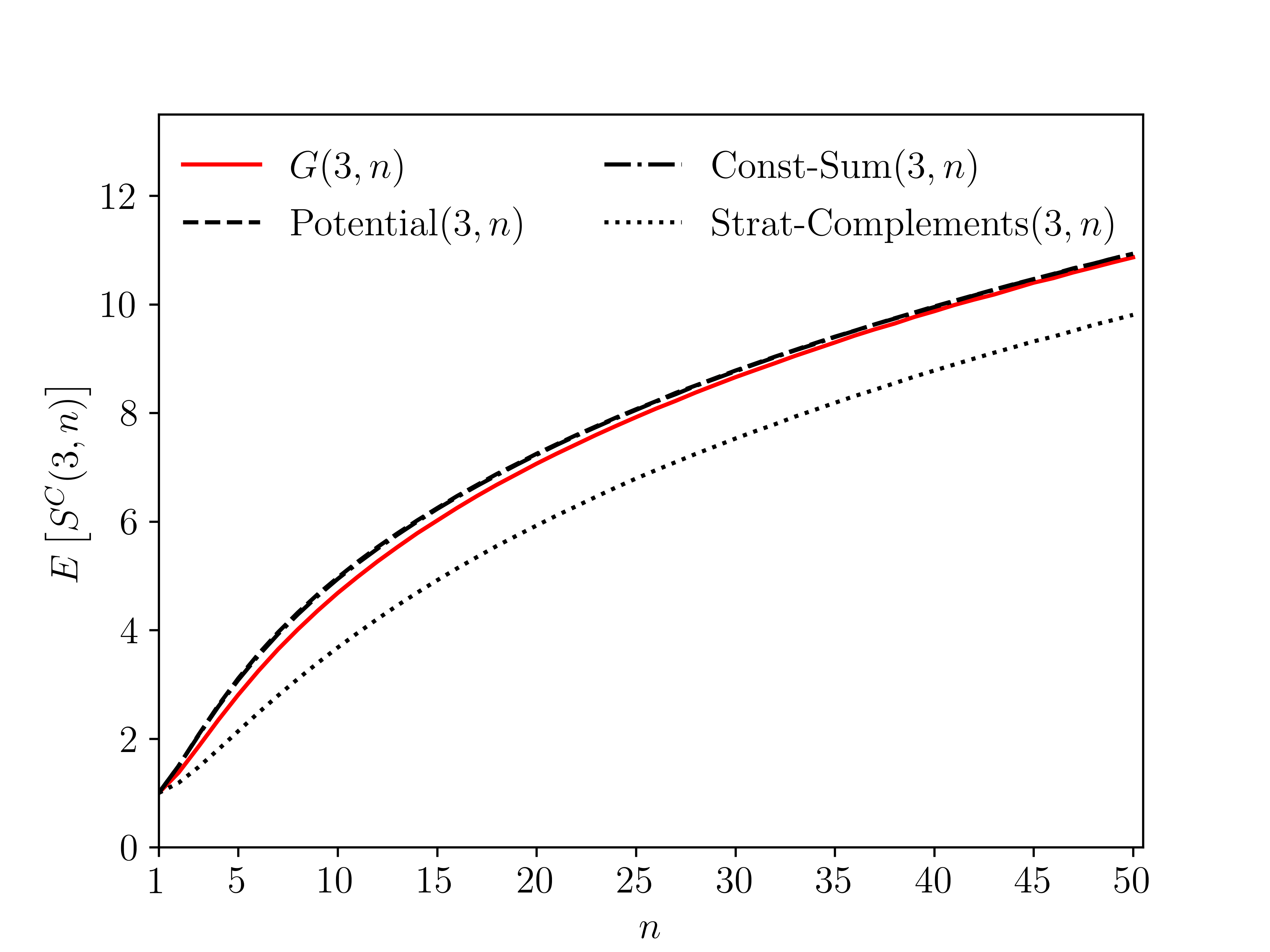}
\caption{Surviving actions}\label{fig: S_alt_imbal}
\end{subfigure}
\captionsetup{justification=centering}
\caption{Three dimensions of dominance solvability in $3 \times n$ games for alternative distributional assumptions}
\label{fig: alt_imbal}
\end{figure}

As can be seen, even when action sets are imbalanced, our insights regarding the efficacy of the iterative elimination procedure remain similar for uniformly random games and games governed by other distributions, corresponding to commonly-studied payoff structures.

One exception pertains to the number of iterations required conditional on dominance solvability in imbalanced potential games. As Column's action set expands, that number declines. We conjecture the underlying reason is the following. Conditional on a potential game being dominance solvable, when Column has multiple actions left after the first iterative round, Row must have some actions to eliminate---otherwise, the game would not be dominance solvable to begin with. This imposes restrictions on Column's payoffs that are absent when Row and Column's payoffs are determined independently. A deeper  investigation of the features of dominance solvability in potential games is left for future research.


\subsection{Dominance via Mixed Strategies in Lab Games}

Figure~\ref{fig: mixed_lab} presents the likelihood of playing an action dominated by a mixed strategy in laboratory games, using data from \cite{fudenberg_predicting_2019}, and in laboratory games, using data compiled by \cite{wright2014level}. We consider CRRA utilities, highlighting the risk-aversion parameters estimated by  \cite{fudenberg_predicting_2019} for those games.

As can be seen, in both samples, participants have a substantially harder time eliminating actions dominated by mixed strategies.

\begin{figure}[ht]
\begin{subfigure}[t]{.49\textwidth}
\centering
\includegraphics[width=\linewidth]{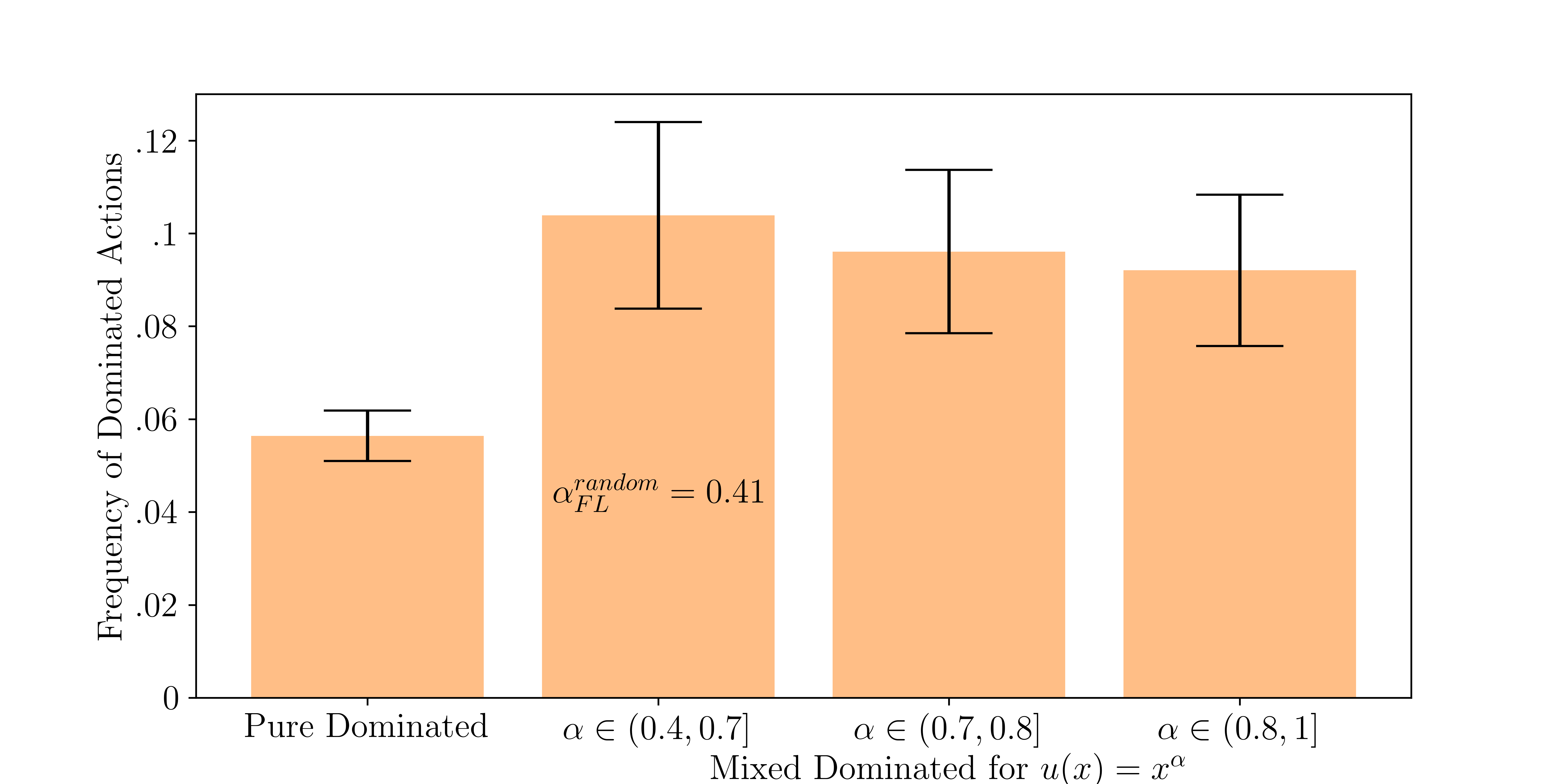}
        \caption{\textit{Random} games \citep{fudenberg_predicting_2019}, $\alpha_{FL}^{random}=0.41$ for \textit{random} games is estimated by \cite{fudenberg_predicting_2019}}\label{fig: fl_random}
\end{subfigure}
\begin{subfigure}[t]{.49\textwidth}
\centering
\includegraphics[width=\linewidth]{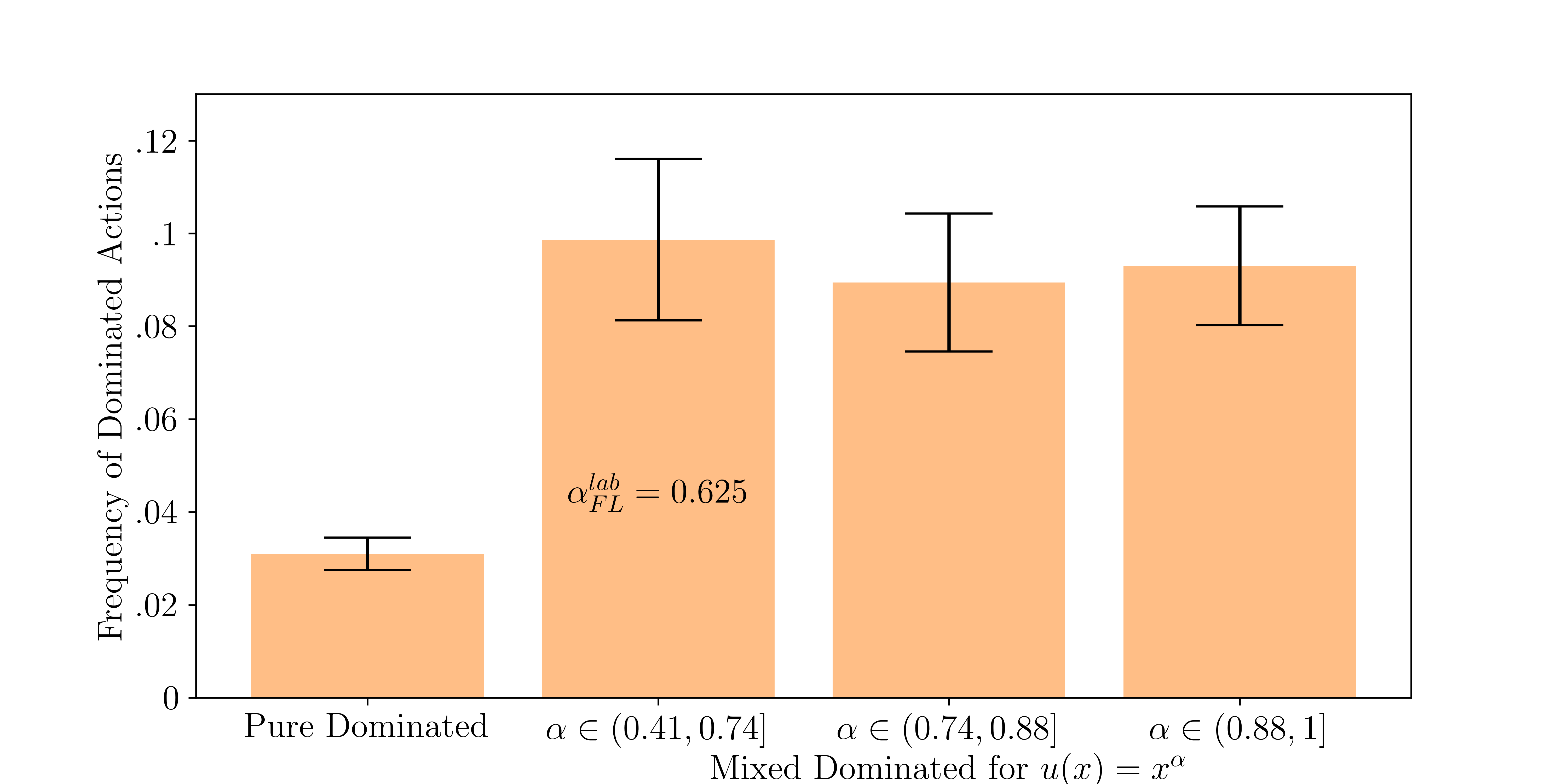}
        \caption{\textit{Laboratory} games \citep{wright2014level}, $\alpha_{FL}^{lab}=0.625$ for \textit{laboratory} games is estimated by \cite{fudenberg_predicting_2019}}\label{fig: fl_lab}
\end{subfigure}
    \captionsetup{justification=centering}
    \caption{Frequency of Row's dominated decisions in games with (1) exactly one Row's strictly dominated action and no other weakly-dominated actions and (2) exactly one Row's mixed-strategy dominated action and no weakly-dominated actions}
    \label{fig: mixed_lab}
\end{figure}


\subsection{Mixed-strategy Dominance Solvability in Imbalanced Games}

Figure~\ref{fig: mixed_bal} is the analogue of Figure 6 in the main text for imbalanced $3 \times n$ games. As can be seen, all of the insights demonstrated in the text for $n \times n$ games carry over.

\newpage

\begin{figure}[ht]
\centering
\begin{subfigure}[t]{.49\textwidth}
\centering
\includegraphics[width=\linewidth]{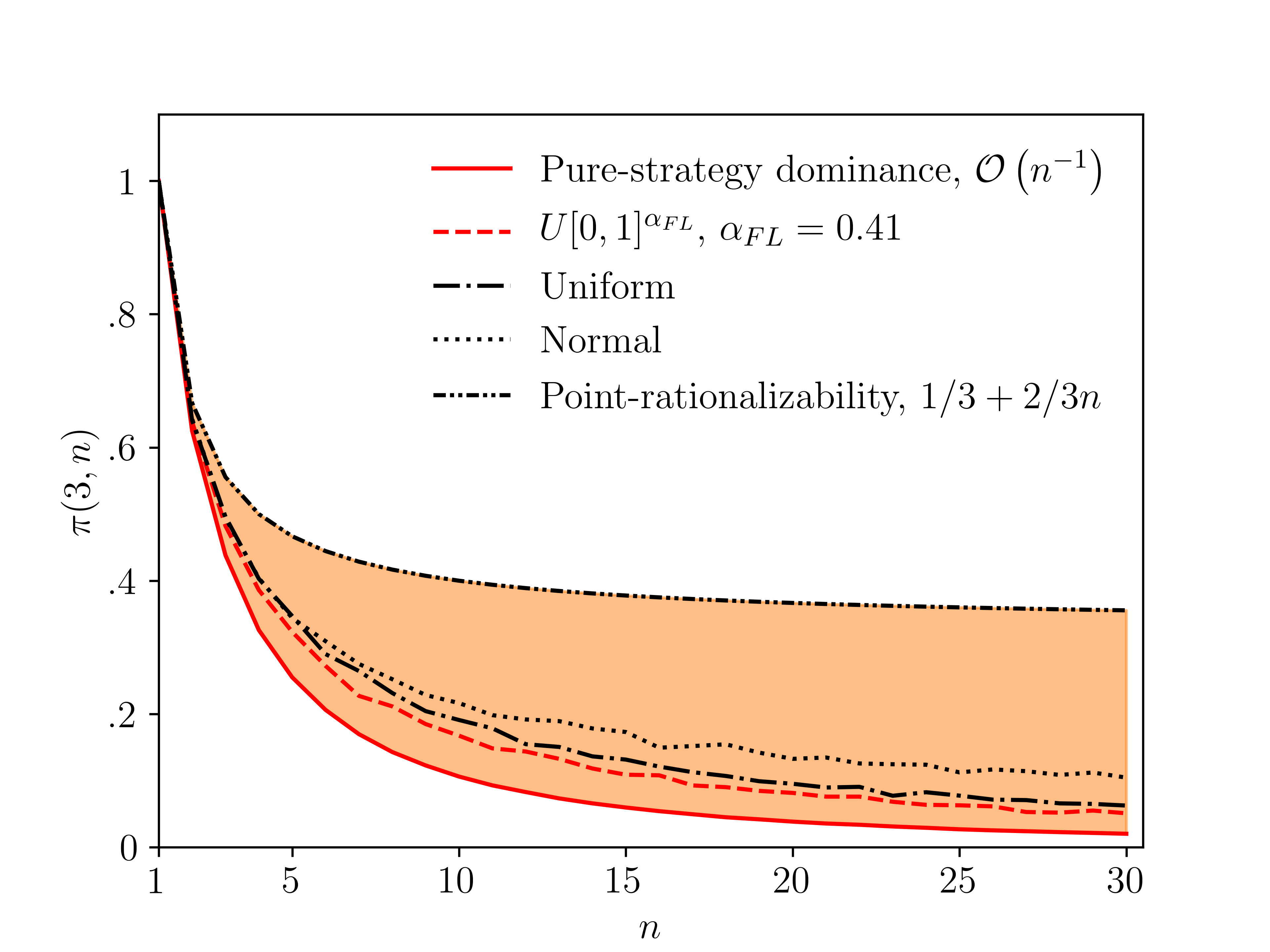}
        \caption{Probability of strict dominance}\label{fig: pi_mixed_imbal}
\end{subfigure}
\begin{subfigure}[t]{.49\textwidth}
\centering
\includegraphics[width=\linewidth]{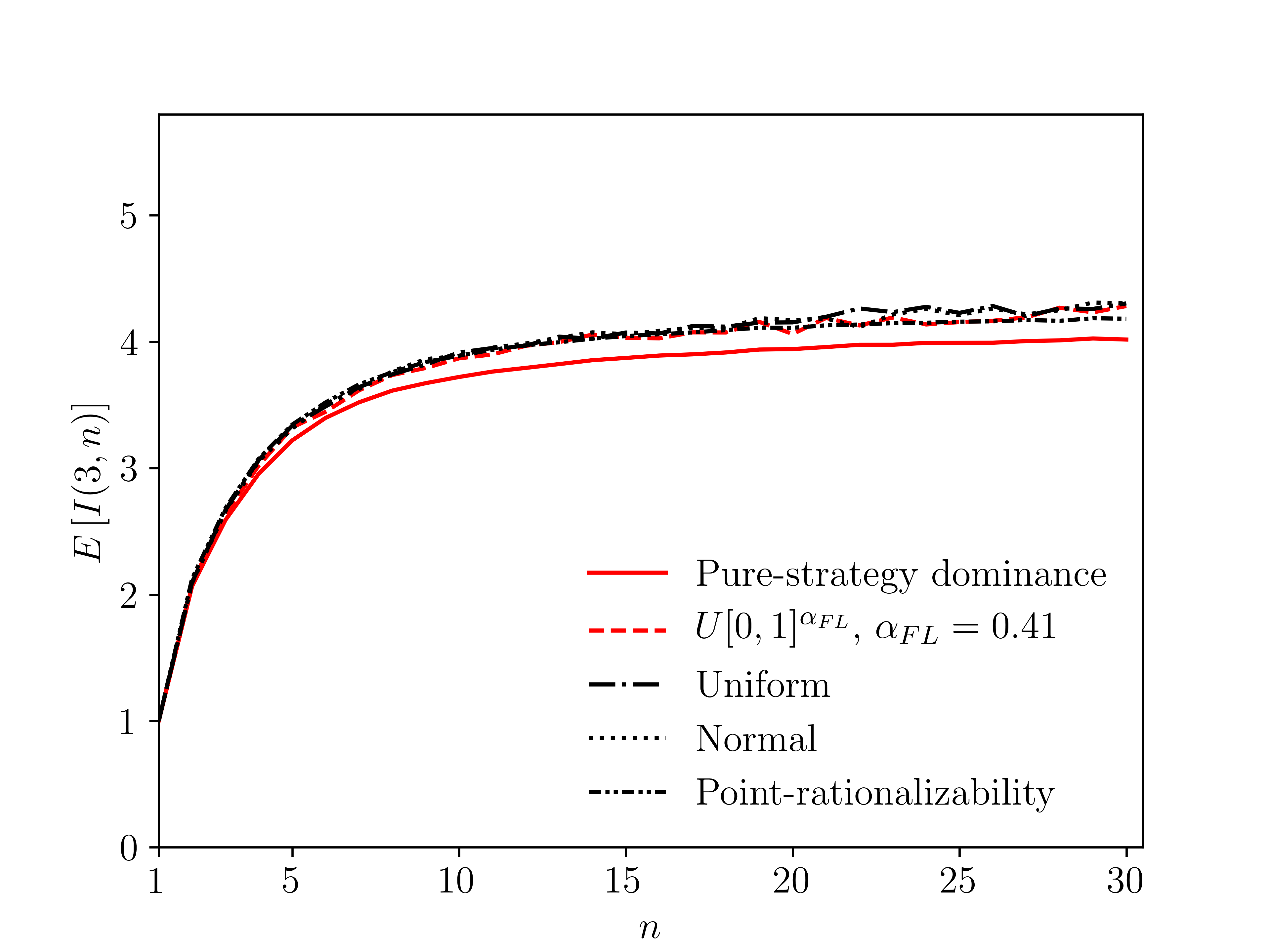}
\caption{Conditional number of iterations}\label{fig: I_mixed_imbal}
\end{subfigure}
\par
\medskip
\par
\begin{subfigure}[t]{\textwidth}
\centering
\vspace{0pt}
\includegraphics[width=0.49\linewidth]{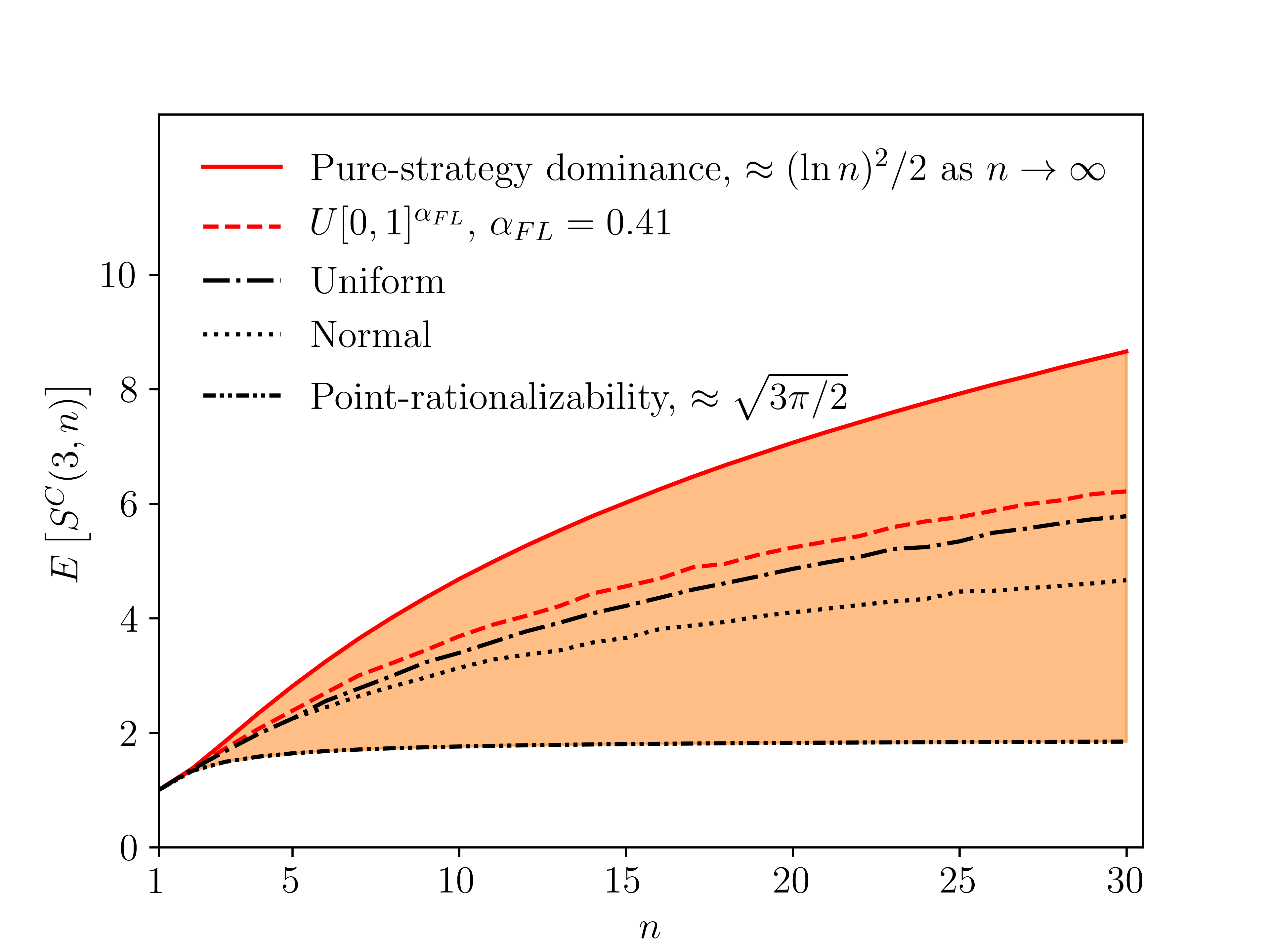}
\caption{Surviving actions}\label{fig: S_mixed_imbal}
\end{subfigure}
\captionsetup{justification=centering}
\caption{Three dimensions of mixed-strategy dominance solvability in $3 \times n$ games, where $\alpha_{FL}=0.41$ for \textit{randomly-generated} games is estimated by \cite{fudenberg_predicting_2019}}
\label{fig: mixed_bal}
\end{figure}

\newpage
\addcontentsline{toc}{section}{References}
\bibliographystyle{ecta}
\bibliography{references}